\documentclass{article}
\usepackage[utf8]{inputenc}
\usepackage{amsthm,amssymb,amsfonts,amsmath}
\usepackage{mathdots}
\usepackage{bm,bbold}
\usepackage{appendix}
\usepackage{mathabx}
\usepackage{stmaryrd}
\usepackage{tikz}
\usetikzlibrary{decorations.pathreplacing,calligraphy}
\usepackage{array}

\usepackage{hyperref}
\newcommand{\arrs}[2]{#1_{1\ldots #2}}
\newcommand{\arrsP}[3]{#1_{#3_1\ldots #3_#2}}
\newcommand{\arrsM}[3]{#1_{#2 \ldots #3}}

\newcommand{\symmze}[2]{\bm{\Sigma}^{(#1)}_{#2}}
\newcommand{\formalPS}[1]{\llbracket #1 \rrbracket}
\newcommand{\symbDistArgs}[2]{\left(#2|#1\right)}

\newcommand{\evaluateP}[2]{\mathfrak{e}_{#1}^{#2}}
\newcommand{\evaluatePD}[1]{\mathfrak{E}_{#1}}

\newcommand{\intD}[2]{\mathfrak{I}_{#1}^{#2}}
\newcommand{\extD}[2]{\mathfrak{A}_{#1}^{#2}}%former extenstion
\newcommand{\restrD}[2]{\mathfrak{R}_{#1}^{#2}}

\newcommand{\intDH}[2]{\widehat{\mathfrak{I}}_{#1}^{#2}}
\newcommand{\extDH}[2]{\widehat{\mathfrak{A}}_{#1}^{#2}}
\newcommand{\restrDH}[2]{\widehat{\mathfrak{R}}_{#1}^{#2}}

\newcommand{\restrictH}[2]{\mathcal{R}_{#1}^{#2}}

\newcommand{\MltM}[1]{\mathfrak{M}_{#1}}
\newcommand{\MltMn}[2]{\mathfrak{M}_{#1}^{(#2)}}

\newcommand{\MltS}[1]{\mathfrak{M}^{#1}}

\newcommand{\permK}[2]{\mathfrak{p}^{(#1)}_{#2}}
\newcommand{\DiagM}[1]{{D_{#1}}^*}
\newcommand{\Fourier}[2]{\mathfrak{F}^{#2}_{#1}}
\newcommand{\intSimplex}[1]{J^{(#1)}}
\newcommand{\intUSimplex}[2]{J^{(#1)_{#2}}}
\newcommand{\timeorder}[1]{{\mathrm{\bf{T}}}\left\{#1\right\}}

\newcommand{\TRegion}[1]{\mathcal{T}_{#1}}

\newcommand{\RR}[1]{\mathbb{R}^{#1}}
\newcommand{\SRR}[1]{\mathcal{S}\left(\RR{#1}\right)}

\newcommand{\Mlt}[1]{\mathcal{O_M}\left(\RR{#1}\right)}
\newcommand{\MltT}[1]{\mathcal{O_M}\left(\TRegion{#1}\right)}

\newcommand{\SMltT}[2]{\mathcal{S}\mathcal{O_M}\left(\RR{#1}, \TRegion{#2}\right)}

\newcommand{\SMlt}[2]{\mathcal{SO_M}\left(\RR{#1},\RR{#2}\right)}
\newcommand{\SpRR}[1]{\mathcal{S}'\left(\RR{#1}\right)}
\newcommand{\DS}{\mathcal{D}_{\mathcal{S}}}
\newcommand{\DSS}[2]{\DS^{#1}\SRR{#2}}
\newcommand{\DSMlt}[2]{\DS^{#1}\Mlt{#2}}
\newcommand{\DSSMlt}[3]{\DS^{#1}\SMlt{#2}{#3}}

\newcommand{\LL}{\mathcal{L}}
\newcommand{\LLfirst}{\widecheck{\mathcal{L}}}

\newcommand{\LDS}{\LL(\mathcal{D}_{\mathcal{S}})}
\newcommand{\permNKS}[2]{\mathfrak{p}_{#1}^{*(#2)}}

\newcommand{\symmgr}[1]{\mathfrak{S}_{#1}}
\newcommand{\WightTf}[3]{\mathcal{W}^{(#3)}_{\arrs{#2}{#1}}}

\newcommand{\PWightmanRng}[2]{W^{\mathrm{R}[#2]}_{#1}}

\newcommand{\PWightmanRnag}[3]{W^{\mathrm{R}[#3]}_{#1;\arrs{#2}{#1}}}

\newcommand{\GreenUng}[2]{\mathcal{G}^{[#2]}_{#1}}

\newcommand{\GreenUnag}[3]{\mathcal{G}^{[#3]}_{#1;\arrs{#2}{#1}}}
\newcommand{\WightmanRnagP}[4]{\mathcal{W}^{\mathrm{R}[#3]}_{#1;\arrsP{#2}{#1}{#4}}}

\newcommand{\WightmanRnag}[3]{\mathcal{W}^{\mathrm{R}[#3]}_{#1;\arrs{#2}{#1}}}
\newcommand{\WightmanUnag}[3]{\mathcal{W}^{[#3]}_{#1;\arrs{#2}{#1}}}
\newcommand{\WightmanUng}[2]{\mathcal{W}^{[#2]}_{#1}}

\newcommand{\WightmanRnaV}[3]{\mathcal{W}^{\mathrm{R}(#3)}_{#1;\arrs{#2}{#1}}}
\newcommand{\WightmanRnaVarr}[3]{\mathcal{W}^{\mathrm{R}(\arrs{#3}{#1})}_{#1;\arrs{#2}{#1}}}

\newcommand{\WightmanRnaVarrG}[4]{\mathcal{W}^{\mathrm{R}(\arrs{#3}{#1});#4}_{#1;\arrs{#2}{#1}}}

\newcommand{\PWightmanRnaV}[3]{W^{\mathrm{R}(#3)}_{#1;\arrs{#2}{#1}}}
\newcommand{\PWightmanRnaVarr}[3]{W^{\mathrm{R}(\arrs{#3}{#1})}_{#1;\arrs{#2}{#1}}}
\newcommand{\PWightmanRnaVarrS}[3]{W^{\mathrm{R},\Sigma(\arrs{#3}{#1})}_{#1;\arrs{#2}{#1}}}

\newcommand{\WightmanRnagAd}[3]{\mathcal{W}^{\mathrm{R}[#3]}_{#1;\arrs{#2}{#1}}}
\newcommand{\WightmanUnagAd}[3]{\mathcal{W}^{[#3]}_{#1;\arrs{#2}{#1}}}
\newcommand{\GreenUnagAd}[3]{\mathcal{G}^{[#3]}_{#1;\arrs{#2}{#1}}}
\newcommand{\GreenEUnagAd}[3]{\widetilde{\mathcal{G}}^{[#3]}_{#1;\arrs{#2}{#1}}}
\newcommand{\hotimes}{\widehat{\otimes}}
\newcommand{\tlambda}{\widetilde{\lambda}}
\newcommand{\ii}{\mathrm{i}}
\newcommand{\uwhat}[1]{\underline{\widehat{#1}}}
\numberwithin{equation}{section}
\newtheorem{thm}{Theorem}[section]
\newtheorem{prop}[thm]{Proposition}

\newtheorem{lem}[thm]{Lemma}
\newtheorem{=>}[thm]{Corollary}

\newtheorem{notn}[thm]{Notation}

\theoremstyle{definition} 
\newtheorem{rmk}[thm]{Remark}
\newtheorem{Def}[thm]{Definition}
\newtheorem{exmp}[thm]{Example}

\newcommand{\NN}{\mathbb{N}}
\newcommand{\sign}{\mathrm{sign}}

\title{Ultraviolet finite non-local Hamiltonian perturbation Quantum Field Theories and their weak adiabatic limit}
\author{Aleksei Bykov }
\date{November 2022}

\begin{document}

\maketitle
\begin{abstract}
    The non-local quantum field theories attract interest in the mathematical and physical community as candidates for effective description of the reality taking into account the quantum gravity effects. The standard methods, developed for the conventional local quantum field theories, often are not applicable. In this paper the Hamiltonian approach, originally proposed for quantization of the non-local effective theory inspired by non-commutative geometry of the Doplicher-Fredenhagen-Roberts Quantum Spacetime, is generalized to a large class of ultraviolet finite non-local Hamiltonian perturbative Quantum Field Theories. General properties of such theories are discussed. The weak adiabatic limit existence is proved for the whole class of theories. 
\end{abstract}
\section{Introduction}
Local quantum field theory is a powerful framework to build models of physical reality. In particular, the Standard Model is known for its extremely precise predictions. Its modifications, known as ``Beyond Standard Model Physics", are also local quantum field theories. For the models of quantum gravity, local quantum field theories serve as low-energy effective theories. 
 Despite this success, it is very unlikely that a physical theory including all known interactions can be built within that framework. In particular, it was argued by many authors, starting from the early days of the quantum field theory \cite{Bronstein, Klein}, that quantum gravity can not be local. The physical reason of that is the interplay between Heisenberg uncertainty relations and the role played by the energy and momentum in General Relativity. As a result, any attempt of a too much localized measurement inevitably produces a dramatic change of the geometry such as a black hole formation, making the measurement impossible. Assuming that the main object of the theory is a system of (in principle) measurable observables, we conclude that the language of local QFT is not appropriate. There is a related long-standing conjecture that the gravity-induced non-locality may be a cure of the ultra-violet divergences of local QFT \cite{Deser}, since the later are caused by the singular products of point-localised observables, or, equivalently, by point-localized elementary scatterings.
 \par 
 For this reason, non-local quantum field theories attract interest. In particular, a lot of non-local models are considered as effective theories of quantum field theory on quantum (non-commutative) spacetimes, e.g. \cite{DFR,NekrasovNCQFT}. 
 \par 
\begin{rmk}
 \label{rmk:Intro/TwoLocalities}
 Locality plays a crucial role in both the mathematical and physical aspects of Quantum Field Theory. The meaning of this word is, however, slightly different. In mathematics, in particular in algebraic Quantum Field Theory, by locality one usually means commutativity of the observables related to measurements performed in space-like separated regions of spacetime. In physics one is usually concerned with the locality of the interactions. The precise meaning of this varies from approach to approach, but one expects that the terms describing the interaction (e.g. the higher than quadratic part of the Hamiltonian or Lagrangian) are local non-linear functionals of quantum fields. In the language of Feynman graphs, it may be interpreted as the localization of all ``elementary scatterings" at points of spacetime. These two meanings are not unrelated. For example, in causal perturbation theory the ``physical" locality together with causality implies the ``algebraic one" \cite{pAQFT}. Throughout this paper we by default understand non-locality in the ``physical" sense.    
 \end{rmk}
 The main problem of this area of mathematical physics is the lack of a universal understanding of what a non-local quantum field theory is. The ambiguities arise even at the perturbation theory level.
 In conventional quantum field theory, there are different approaches, involving various mathematical constructions but describing essentially the same physics. This equivalence is lost in the absence of locality, as well as the applicability of the most popular methods just can not be used anymore. For example, the approaches based on the functional integrals, widely used in physics, fail to provide a unitary theory in the non-local case \cite{NekrasovNCQFT}\footnote{In this paper we consider the physical Lorentzian signature of the spacetime metrics only. In the Euclidean signature, for which the functional integral approach, as usual, is much more reasonable and is widely studied, e.g. \cite{EucNCQFT}. Yet, the Wick rotation trick, normally used to pass from one signature to another does not work for non-local theories \cite{BahnsWick} at least in a straightforward way.}.
 Mathematical approaches to perturbative quantum field theory \cite{EG73}, based on the Einstein causality can not be applied directly to non-local theories.
 \par 
 Therefore, less standard ideas turn out to be useful for non-local theories.
 In this paper we deal with the so-called Hamiltonian approach, first suggested in \cite{DFR} and further developed in \cite{BahnsPhD,BahnsEtAl}. The basic idea is to apply the standard Hamiltonian perturbation theory to a non-local interaction part of the Hamiltonian. As usual, it produces a time-ordered product of interacting Hamiltonians, integrated over spacetime (see Section \ref{HpQFT}).
 Unlike the local theories, here the time-ordered product makes sense without UV regularization. The integration over spacetime is instead divergent and has to be regularized through the multiplication of all coupling constants with the adiabatic cut-off function depending on the spacetime coordinates and vanishing far away from the origin. This is a typical trick of perturbative quantum field theory, unavoidable in general due to the Haag theorem \cite{HaagThm}. To obtain a physically meaningful result one has to pass to the limit of the adiabatic cut-off function going to the unity, known as the adiabatic limit.  
 \par 
 Among the advantages of the Hamiltonian approach are the explicit expression for the scattering operator and its manifest unitarity\footnote{See, however, Remark \ref{rmk:HpQFT/loc-alg}.}
 . A notable restriction of this method is that all theories constructed in this way (if they are not local) have a selected reference frame, in which residual locality persists (see Remark \ref{rmk:HpQFT/loc-alg}). 
 \par
 Originally in \cite{BahnsEtAl,BahnsPhD} the Hamiltonian approach was applied to a class of non-local quantum field theories built as an effective description of quantum field theories in the Doplicher-Frednhagen-Roberts quantum (non-commutative) spacetime defined in \cite{DFR}. The latter is a simple Lorentz-covariant realization of the physically motivated spacetime uncertainty relations also derived in \cite{DFR}.  A set of Feynman rules for this theory was derived in \cite{BahnsPhD} and it was shown that the non-locality smears out the typical ultraviolet singularities of local theories. However, it was observed that the adiabatic limit of the scattering operator (so-called strong adiabatic limit) does not exist in general due to the external lines correction. This is not a surprise. In fact, in the local case, the strong adiabatic limit exists only when particular renormalization conditions are specified \cite{EG73}.
 \par
  In the author's Ph.D. thesis \cite{PhDThesis} a class of non-local quantum ultraviolet finite  Hamiltonian perturbative Quantum Field Theories (HpQFT) was defined. This class includes the theories on the Doplicher-Fredenhagen-Roberts quantum spacetime considered in \cite{BahnsPhD,BahnsEtAl}. It was shown that the weak adiabatic limit always exists as in the local case.
  \par 
  In this paper, we reproduce this result using significantly different techniques for the computations. Namely, instead of unbounded operators (operator-valued functions, operator-valued distributions, etc.) on the Hilbert space, we deal with continuous operators on the spaces of Schwartz functions. Then the vectors, vector-valued, operator-valued functions, distributions, and parameter-dependent distributions can be treated uniformly as linear continuous operators between the spaces of Schwartz functions. Furthermore, we introduce the second quantization map (see Subsection \ref{DS/2Q}), allowing us to do all computations with the spaces of a fixed finite number of particles only. Na\"ive formulas, such as the Wick theorem and its generalizations for operator-valued distributions can be provided with a direct sense in this formalism. Here we use only a small part of this framework. In the subsequent paper devoted to the strong adiabatic limit, a suitable generalization of this formalism will be used much more intensively. From this perspective, besides giving an improved exposition of the constructions of \cite{PhDThesis} this paper may be considered as an illustration of this new formalism.
  \par
  The paper is organized as follows. The rest of this section contains a summary of the notation and terminology used. In the second Section we explain how different objects, such as vector-valued, operator-valued, and parameter-dependent (as well as mixes of these kinds) distributions can be treated uniformly as linear operators between appropriately chosen spaces. In the third Section, the standard domain of the quantum field theory $\DS$ is defined. We also construct classes of operators, operator-valued functions, distributions, etc. on $\DS$ and introduce the second quantization procedure.  In the Fourth Section, motivated by the time-dependent perturbation field theory in quantum mechanics, we describe the class of non-local quantum field theories we deal with. In the last, fifth Section, the weak adiabatic limit existence is formulated and proved. The Appendix contains several formulations of Feynman rules in the adiabatic limit convenient for practical computations.
  
  \paragraph{Acknowledgements}
  This work is a reformulation and continuation of the PhD project and later partially supported by the PostDoc fellowship, both granted by the University of Rome ``Tor Vergata" and performed under the advisory of G. Morsella.

\subsection{Conventions, notation and terminology}\label{Intro/prelim}
\paragraph{General}
Three-dimensional vectors semantically related to the physical coordinate or momentum space are denoted with the vector symbol, e.g. $\vec{p}\in\RR{3}$. Other finite-dimensional vectors are denoted with bold letters, e.. $\bm{x}\in\RR{n}$.
 \par
 We use the standard notation 
$$
\partial^{\alpha}=\partial_{x}^{\alpha}=\frac{\partial^{|\alpha|}}{\partial x_1^{\alpha_1} \cdots \partial x_n^{\alpha_n}}
$$
for the partial derivatives, where $\alpha=(\alpha_1,\ldots,\alpha_n)\in\NN_0^{n}$ is a multi-index and $|\alpha|=\sum_{j=1}^n \alpha_j$. 
\par 
The permutation group of $k$ elements is denoted with $\symmgr{k}$. A permutation $\sigma\in\symmgr{k}$ is understood as a functions $\{1,\ldots,k\}\rightarrow \{1,\ldots,k\}$, $i\mapsto \sigma(i)=\sigma_i$. 
\par 
We use the following abbreviated notations for lists of repeating variables:
$$\arrs{x}{n} \Longleftrightarrow x_1,x_2,\ldots, x_n$$ 
$$\arrsM{x}{n}{k} \Longleftrightarrow x_n,x_{n+1},\ldots, x_k,$$
$$\arrsP{x}{k}{\sigma} \Longleftrightarrow x_{\sigma(1)},\ldots,x_{\sigma(k)}, \, (\sigma\in\symmgr{k}),$$
$$
(x_i)_{i=1\ldots n} \Longleftrightarrow x_1,\ldots, x_n,
$$
$$
(x_{i,j})_{i=1\ldots n,j=1\ldots m_j} \Longleftrightarrow x_{1,1},x_{1,2},\ldots, x_{1,m_1},x_{2,1},x_{2,2},\ldots,x_{2,m_2},\ldots,x_{n,m_n}.
$$
The last two notations are used for more composite construction depending on $i$ (or $(i,j)$) in place of $x_i$ (or $x_{i,j}$). Inside integrals we use 
 $$
 d^{nk}\arrs{\bm{x}}{k}=\prod_{j=1}^k d^n \bm{x}_j
 $$ 
 for $\arrs{\bm{x}}{n}\in\RR{n}$. Note that we always indicate the dimensions of the space explicitly. We write $\arrs{\bm{x}}{k}\in\RR{n}$ to express that  $\bm{x}_i\in\RR{k}$ for each $i=1,\ldots,n$ in contrast with $(\arrs{\bm{x}}{k})\in\RR{n\cdot k}$.
\par
For two topological spaces $X$, $Y$ we use $\LL(X,Y)$ to denote the space of all linear continuous operators from $X$ to $Y$. We set $\LL(X)=\LL(X,X)$.
\par
If $f$ is a function on $\RR{n+m}$ we may interchangeably write $f(\bm{x},\bm{y})$ or $f(\bm{z})$ for $\bm{x}\in\RR{n},\bm{y}\in\RR{m}$and $\bm{z}\in\RR{n+m}$.
\par 
We use the ordering convention for non-commutative products
$$
\prod_{i=n}^{m} A_i=A_n A_{n+1}\cdots A_{n+m}.
$$
\paragraph{Quantum field theories}
By a quantum field theory we in general mean a (local or non-local, not necessarily Lorentz-covariant) bosonic real scalar perturbative quantum field theory. All results can be easily generalized to the case of several particle species and the Fermionic statistics. The spin plays no role in absence of the Lorentz invariance, so different polarizations can be treated as separate species. The same can be said about the antiparticles in absence of the locality.
\par 
By a wave function we mean a vector of state of such theory with a fixed number of particles and presented as a function of the momenta.
\par
The Quantum Field theories we consider are not necessarily Lorentz-invariant, both because no reasonable non-local Lorentz-invariant interaction is known and due to immanent breaking of Lorentz invariance by the non-local Hamiltonian formalism (see Remark \ref{rmk:HpQFT/loc-alg}). Yet, in \cite{BahnsPhD,BahnsEtAl} the free theory is assumed to be Lorentz-invariant. We instead allow it to be broken from the beginning as it produces no additional technical problems. For this reason, we need to deal with generalized dispersion relations. We require that both the energy and its inverse are smooth functions of the momentum which have at most polynomial growth at infinity together with all their partial derivatives. The space of such functions $\Mlt{3}$ is defined in Subsection \ref{Framework/Mlt}.
\begin{Def}\label{def:dispRel}
By a \emph{(massive) dispersion function} we mean a positive function $\omega_0\in\Mlt{3}$ such that $(\vec{k}\mapsto \frac{1}{\omega_0(\vec{k})})\in\Mlt{3}$, and\footnote{The $M>0$ condition is introduced to eliminate exotic dispersion functions of polynomial decaying at the infinity. The possibility of zero energy at a finite point is already excluded by the smoothness of the inverse energy.}
$$
M=\inf_{\vec{k}\in\RR{3}}\omega_0({\vec{k}})>0.
$$
The constant $M$ is called the \emph{mass} of the dispersion function. We always assume $\omega_0(-\vec{p})=\omega_0(\vec{p})$ for the sake of simplicity. This condition can be easily relaxed.
\end{Def}
\paragraph{Feynman graphs and Feynman rules}
We use the Feynman graphs to keep track of combinatorics in operator products computations.
\par
The terminology used is mostly standard and intuitive, but it makes sense to summarize it in one place. We introduce various versions of the Feynman rules both in the main text and in Appendix based on Feynman graphs with different labeling and auxiliary structures.
\par
In general, a \emph{Feynman graph} is a graph with marked vertices and edges (interchangeably called lines). We always exclude the graphs containing an edge starting and ending at the same vertex. 
\par 
We consider \emph{partially} or \emph{totally ordered} and \emph{unordered} Feynman graphs, which come with partial, total or no order on the set of vertices respectively. When dealing with ordered graphs, we use terms like \emph{precede}, \emph{earlier}, \emph{earliest} and so on (respectively, \emph{follows}, \emph{later}, \emph{latest} etc.) to describe the relative order of vertices and use the symbols $\prec$ and $\succ$ respectively for these two relations\footnote{This terminology have clear interpretation if we keep in mind that each vertex corresponds to an operator. Then earlier vertices correspond to operators which act first. In computations of the Green functions (and the scattering amplitudes \cite{PhDThesis}) the order can be physically interpreted as the order of the ``elementary scatterings" in the physical time. The physical interpretation is lost in the Wightman functions computation, where some parts of the operator product are anti-time ordered.}. 
\par
A total order, when prescribed, induces the natural (up to an overall flip) orientation of the edges, making the graph directed acyclic. For definiteness, we choose the orientation from the earlier to the later vertices. Partial orders in general do not define orientations of edges, but we consider only such of them that do, and moreover are generated by an acyclic orientation of the edges\footnote{In other words, instead of partially ordered graphs we should be speaking of directed acyclic ones. But the partial order plays a major role in the Feynman rules, which we have emphasized in the terminology.}. In the figures we always draw earlier vertices to the left of the later ones.
\par 
We distinguish external and internal vertices. The external vertices may be incident to only one line, for internal ones there are no limitations\footnote{Of course, the contribution of a diagram may vanish if the corresponding term is absent in the interaction.}. 
\par 
The vertices can be enumerated or not, and the numeration can agree with the order or not. When it does, rather counter-intuitively, in ordered graphs we enumerate the vertices from the \emph{last to the earliest}\footnote{This choice of numeration is dictated by the traditional right-to-left composition notation and numeration of the factors in products from left to right}.
\par
To the lines we often assign a flux of the momentum and sometimes the energy flow. The flow is directed, i.e. it is incoming for some vertices and outgoing for others.  Graphically we draw an arrow parallel to the line to show the flow direction. For ordered graphs we assume that the flux always goes from the earlier to the later vertex. We call the difference of the sum of all outgoing and the sum of all incoming momenta (respectively,  energy) at a vertex the \emph{momentum} (respectively, \emph{energy}) \emph{defect} of the vertex. In some cases we make a distinguishing between the on-shell (i.e. constrained by a dispersion relation) and off-shell (i.e. an independent free variable) energies, and respectively the on-shell and off-shell energy defects.
\par 
We say that a connected component of a Feynman graph is a \emph{vacuum component} if it contains no external vertices.
\par

The \emph{Feynman rules} consist of the following ingredients: 
\begin{itemize}
    \item Description of the classes of relevant Feynman graphs (e.g. ordered graphs with a fixed number of internal and external enumerated vertices with no vacuum energy corrections);
    \item Labeling rules, i.e. the list of free variables assigned to each line (e.g. spatial momentum) or each vertex (e.g. timestamp);
    \item Prescription o the factors corresponding to each line and each vertex;
    \item The integration and summation range for the free parameters and additional overall factor if necessary.
\end{itemize}
To each graph corresponds an expression constructed as a product of factors corresponding to the components of the graph according to the Feynman rules and the overall factor (if it is present), integrated and summed with respect to the free parameters and
divided by the order of the automorphisms group of the graph\footnote{We assume that the automorphisms preserve all defined structure of the graphs (including order and numeration of vertices if they are defined) except for free parameters assignment. The combinatoric factors are not at all important for the weak adiabatic limit, so we do not discuss the symmetry factor, but it can always be reconstructed. Yet, we note that this factor is different from the one of \cite{BahnsPhD} because of different normalization of the vertex factors. Our conventions instead agree with the standard physics textbooks, e.g. \cite{PS} to facilitate eventual comparisons.}. 
\par 
We say that some object can be computed by Feynman rules if it is equal to the sum of expressions corresponding to all relevant graphs according to that Feynman rules.
\paragraph{Formal power series}
Let $g$ be a formal parameter. 
For a vector space $X$, the space of $X$-valued formal power series in terms of $g$ is defined as the vector space of $X$-valued functions on $\NN_0$,
$$
X\formalPS{g}=X^{\NN_0}=\{x:\NN_0\rightarrow X\}
$$
with pointwise linear operations. Elements of $X\formalPS{g}$ we symbolically present as
$$
x(g)=\sum_{n=0}^{\infty}x_{(n)}g^n. 
$$
\par
If $X$, $Y$ and $Z$ are vector spaces and $*:X\times Y\rightarrow Z$ is a bilinear operation, then for $x\in X\formalPS{g}$, $y\in Y\formalPS{g}$ we set 
$$
(x*y)(g)=\sum_{n=0}^{\infty}\left(\sum_{m=0}^{n}x_{(m)}*x_{(n-m)}\right)g^n\in Z\formalPS{g}.
$$
The limits involving formal power series should be understood in the sense of pointwise convergence of the underlying functions on $\NN_0$. Consequently, integrals of power series are defined.
\par 
To avoid confusion we write
$(x(g))_{g^n}$ instead of $x_n$

\paragraph{Test functions, distributions, and related objects}
We use the standard symbol for the space of Schwartz function $\SRR{n}$ and the space of tempered distributions $\SpRR{n}=\LL(\SRR{n},\mathbb{C})$ dual to it (here $n\in\mathbb{N}_0$ with $\SRR{0}=\SpRR{0}=\mathbb{C}$ for convenience). The seminorms of $\SRR{n}$ are denoted as follows:
$$
||f||_{\SRR{n},m,k}=\sup_{\bm{x}\in\RR{n},|a|\leq k}(1+||\bm{x}||)^m|\partial^{a}f(\bm{x})|, \,\forall f\in\SRR{n},\forall m,k\in\mathbb{N}_0.
$$
We often omit the adjective tempered, as no other classes of distributions appear in this paper.
We use the square brackets to denote the evaluation of a distribution on a test function (or, more generally, the action of an operator on a vector\footnote{In selected cases we omit the square bracket in this sense.}). The round brackets after a symbol of distribution are used in the symbolic integral notation which is explained in Subsection \ref{Framework/Symbolic}.  
\par 
A lot of non-standard notation is introduced in Section \ref{Framework}. In particular, the operator $\evaluateP{}{}$ appears in  Subsection \ref{Framework/VectorTest}; the operation $\extDH{}{}{}$ is explained Subsection \ref{Framework/OpVal}; the operations $\intD{}{}$ and $\restrD{}{}$ is introduced in Subsection \ref{Framework/Restrict}; the operations $\extD{}{}$ and the tensor product $\otimes$ are defined in Subsection \ref{eq:Framework/extProd}. Most of the other operators used in the paper can be found in Subsection \ref{Framework/exmp}
\par
\par 
 For the secondary quantization map $A\mapsto \widehat{A}$, and an explanation of the underlining notation see Subsection {\ref{DS/2Q}}. As a general rule, the hated underlined symbols denote the objects in their standard sense. For example, one can think of $\uwhat{\phi}_0$ as the usual free quantum field, defined as an operator-valued distribution. The original object $\phi_0\in\LLfirst(\DSS{}{4},\DS{})$ (notation of Section \ref{DS} is used) does not have a direct sense but is convenient to operate with.  

%%%%%%%%%%%%%%%%%%%%%
\section{Framework for operator-valued parameter-dependent distributions}\label{Framework}
As usual in Quantum Field Theory, operator-valued distributions are essential for construction. Standard general formalism (e.g. \cite{SimonReed, NN}) is, however, not suitable for our needs, because we deal with a certain class of operators and distributions only. In particular, instead of the Hilbert space, we may always deal with finer spaces of Schwartz functions $\SRR{n}$, $n\in\NN$. This section is devoted to $\SRR{n}$ and $\LL(\SRR{n},\SRR{m})$-valued functions and (in general, parameter-dependent) distributions. In Section \ref{DS} we explain how this construction leads to the usual notion of distributions, valued in unbounded operators on the Hilbert space.
\par 
\par 
The material presented here is not original, but nowhere presented in the form suitable enough for us. In the parts devoted to parameter-dependent distributions and restrictions of distributions we mostly follow \cite{NN}. The later concept is an older ``cheap" version of the microlocal analysis\cite{Lars}, allowing to restrict the distribution to fixed values of the parameters. There are several reasons why we stick to it. First of all, it is just simpler to verify the existence conditions in the situations we deal with (especially taking into account that otherwise we would be forced to use generalized to the tempered distributions case). The second reason is that this way our framework is more uniform, allowing to treat vector-valued and parameter-dependent distributions together. Finally, the strong adiabatic limit, which is not considered here, but will be presented in the subsequent paper \cite{PAP1}, requires more general spaces, and hence the microlocal analysis would require further generalization. Instead, we aim to use as few properties of the space $\SRR{n}$ as possible\footnote{In broad terms, we use only that it is Fr`echet and that the partial evaluations are continuous in a siutable sense}.
\par
In description of vector-valued functions we partially follow  \cite{Treves}. The part on the multiplier spaces follows \cite{Amman}.
\subsection{Vector-valued test functions}\label{Framework/VectorTest}
We start by the following observation\footnote{See chapter 40 of \cite{Treves}}.
\begin{rmk}\label{rmk:Framework/VectorTest}
For any $n,m\in\NN$ there is an isomorphism between the following topological vector spaces
\begin{enumerate}
    \item The space $\SRR{n+m}$;
    \item The space $$\mathcal{S}(\RR{n},\SRR{m})= $$
    $$
    \left\{f: \RR{n}\rightarrow \SRR{m}\Big|     ||f||_{\mathcal{S}(\RR{n},\SRR{m}),s,k,l,r} < \infty, \quad \forall  l,r,s,k\in\mathbb{N}_0\right\},
    $$
    where for a function $f:  \RR{n}\rightarrow \SRR{m}$ and $l,r,s,k\in\mathbb{N}_0$ we set
    \begin{equation}
    ||f||_{\mathcal{S}(\RR{n},\SRR{m}),s,k,l,r}=\sup_{\bm{x}\in\RR{n},|a|\leq k}(1+||\bm{x}||)^s||(\partial^{\alpha}f)(\bm{x})||_{\SRR{m},l,r}
    \end{equation}
    with the topology induced by these seminorms.
    \end{enumerate}
    The isomorphism is given by $\evaluateP{n+1\ldots n+m}{n+m}\in \LL(\SRR{n+m},\mathcal{S}(\RR{m},\SRR{n}))$:
    $$
    \evaluateP{n+1\ldots n+m}{n+m}f(\bm{x})(\bm{y})=f(\bm{x},\bm{y}), \forall f\in\SRR{n+m}, \forall \bm{x}\in \RR{n}, \forall \bm{y} \in \RR{m}
    $$
    $$
    \left(\evaluateP{n+1\ldots n+m}{n+m}\right)^{-1}g(\bm{x},\bm{y})=g(\bm{x})(\bm{y}), \, \forall g\in \mathcal{S}(\RR{m},\SRR{n}),\forall \bm{x}\in \RR{n}, \forall \bm{y}\in \RR{m}.
    $$
\end{rmk}
\begin{proof}
The linear map $\evaluateP{n+1\ldots n+m}{n+m}$ is clearly well-defined.
We have to check that $(\evaluateP{n+1\ldots n+m}{n+m})^{-1}$ is indeed valued in $\SRR{n+m}$, or more precisely the joint continuity of $(\evaluateP{n+1\ldots n+m}{n+m})^{-1}g$ together with its partial derivatives for any $g\in \mathcal{S}(\RR{m},\SRR{n})$. For that we write
$$
|g(\bm{x}')(\bm{y}')-g(\bm{x})(\bm{y})|\leq |g(\bm{x}')(\bm{y})-g(\bm{x})(\bm{y})|+|g(\bm{x}')(\bm{y}')-g(\bm{x}')(\bm{y})|\leq 
$$
$$||g||_{\mathcal{S}(\RR{m},\SRR{n}),0,1,0,0}||\bm{x}-\bm{x}'||+
||g||_{\mathcal{S}(\RR{m},\SRR{n}),0,0,0,1}||\bm{y}-\bm{y}'||,
$$
$$
\forall \bm{x},\bm{x}'\in\RR{n}, \bm{y},\bm{y}'\in\RR{m}
$$
The right-hand sides can be made anyhow small for sufficiently small $||\bm{x}-\bm{x}'||$ and $||\bm{y}-\bm{y}'||$, so $g$ is jointly smooth. The derivatives can be treated in the same way. The continuity of $\evaluateP{n+1\ldots n+m}{n+m}$ and $\evaluateP{n+1\ldots n+m}{n+m}$ follows by the standard inequalities
\begin{equation}
(1+||(\bm{x},\bm{y})||)\leq (1+||\bm{x}||)(1+||\bm{y}||), \forall \bm{x}\in\RR{n}, \forall \bm{y}\in\RR{m},
\label{eq:xy<x}  
\end{equation}
\begin{equation}
(1+||\bm{x}||)\leq (1+||(\bm{x},\bm{y})||), \forall \bm{x}\in\RR{n}, \forall \bm{y}\in\RR{m},
\label{eq:x<xy}
\end{equation}
$$
(1+||\bm{y}||)\leq (1+||(\bm{x},\bm{y})||), \forall \bm{x}\in\RR{n}, \forall \bm{y}\in\RR{m}.
$$
\end{proof}
We conclude that the space $\SRR{n+m}$ can be interpreted as the space of $\SRR{m}$-valued  Schwartz functions on $\RR{n}$. The identification is not unique, in general we define $\evaluateP{i_{1}\ldots i_{m}}{n+m}$ extracting the arguments on positions $i_{1},\ldots,i_{m}$ as parameters. This operator is also of great use in the technical definition below in this section. 
\par 
In what follows we do all real computations with more understandable space $\SRR{n+m}$, but sometimes use the aforementioned interpretation to streamline the presentation and to simplify comparison with standard techniques.

\subsection{Vector-valued and parameter-dependent distributions}\label{Framework/VV=PD}
There are two related spaces important for Section \ref{DS}.
First one may consider the space of \emph{vector-valued distributions} $\mathcal{S}'(\RR{n},\SRR{m})$ which is natural to identify with $\LL(\SRR{n},\SRR{m})$.
The second one is the space of \emph{parameter-dependent distributions of Schwartz class}\footnote{We use terminology of \cite{NN}} $\mathcal{S}(\RR{m},\SpRR{n})$. Following \cite{NN}, we define it as a space of maps $\RR{m}\rightarrow \SpRR{n}$ 
$$
F: \RR{m} \ni \bm{x} \rightarrow F_y\in \SpRR{n}
$$
such that
$$
(y\mapsto F_y[f]) \in \SRR{m}, \quad \forall f\in \SRR{n}.
$$
Indeed, it is nothing but a tautological rewriting of differentiability and fast decay (together with all partial derivatives) properties of $F$, considered as a distribution-valued function on $\RR{n}$\footnote{Recall that we assume the weak-* topology for $\SpRR{n}$, so all limits including $F$ should be computed after evaluation on a test function.}.
\par 
These two notions coincide. 
\begin{rmk}[Variant of Exercise 2.43 of \cite{NN}]\label{rmk:Framework/VV=PD}
Fix $n,m\in\NN$. Then there is a natural bijection between $\mathcal{S}'(\RR{n},\SRR{m})$ and $\mathcal{S}(\RR{m},\SpRR{n})$. It is given by identifying
$$
F_{\bm{x}}[f]=F[f](\bm{x}), \quad, \forall \bm{x}\in \RR{m},\forall f\in \SRR{n},
$$
where in the left and the right-hand sides $F$ is treated as a parameter-dependent and vector-valued distribution.
\end{rmk}
\begin{proof}
The only thing to prove is that a parameter-dependent distribution of class $\SRR{m}$ defines a continuous map $\SRR{n}\rightarrow\SRR{m}$. As $\SRR{n}$ and $\SRR{m}$ are Fr\'echet spaces, by the closed graph theorem, it is enough to show for any sequence $f_j\in\SRR{n}$, $j=1,\ldots,n$ the conditions
\begin{equation}
f_j\underset{j\rightarrow \infty}{\longrightarrow}{f}\in \SRR{n},
\label{eq:proof:rmk:Framework/VV=PD/f} 
\end{equation}
\begin{equation}
(y\mapsto F_y[f_j])\underset{j\rightarrow \infty}{\longrightarrow}{g}\in \SRR{m},
\label{eq:proof:rmk:Framework/VV=PD/F[f]}
\end{equation}
imply $g=(y\mapsto F_y[f])$. Since convergence in $\SRR{m}$ implies the pointwise one,  (\ref{eq:proof:rmk:Framework/VV=PD/f}-\ref{eq:proof:rmk:Framework/VV=PD/F[f]}) and the continuity of $F_y$ for each $y\in\RR{m}$ give
$$
g(y)=\lim_{j\rightarrow\infty}F_y[f_j]=F_y[f],\, \forall y\in\RR{m}
$$ 
as needed.
\end{proof}
From now on we make no distinguishing between vector-valued and parameter-dependent distributions and use the notation $\LL(\SRR{n},\SRR{m})$ for both. 
\begin{rmk}

Combining Remarks \ref{Framework/VV=PD} and \ref{rmk:Framework/VectorTest}, for $m,n,k\in\NN_0$ we can interpret $\LL(\SRR {n},\SRR{m+k})$ as the space of parameter-dependent  $\SRR{m}$-valued distributions on $\RR{n}$ of class $\SRR{k}$. 

\end{rmk}
\subsection{Operator-valued (parameter-dependent) distributions}\label{Framework/OpVal}
\begin{rmk}
Take $n,m,k\in\NN_0$.
There is a correspondence between elements of the space $\LL(\SRR{n+m},\SRR{k})$ and $\LL(\SRR{m},\SRR{k})$\emph{-valued distributions on $\RR{n}$}. It is given by defining for each $F\in\LL(\SRR{m},\SRR{k})$ an operator-valued distribution $\evaluatePD{1,\ldots,n}F\in \LL(\SRR{n},\LL(\SRR{m},\SRR{k}))$ by
$$
\evaluatePD{1,\ldots,n}F[f][\psi]=F[f\otimes \psi],\quad \forall f\in\SRR{n},\forall \psi\in\SRR{m}.
$$
Here we assume the pointwise topology on $\LL(\SRR{m},\SRR{k})$.
\end{rmk} 
\begin{proof}
For $f,\psi$ as in the statement, the map $\psi\mapsto f\otimes \psi$ is continuous, so $$\evaluatePD{1,\ldots,n}F[f]\in\LL(\SRR{m},\SRR{k}), \quad, \forall f\in\SRR{n}.$$
Similarly, we get continuity of the map $f\mapsto F[f][\psi]$ for any $\psi\in\SRR{m}$ making $\evaluatePD{1,\ldots,n}F$ into a distribution. 
\par
Conversely, any distribution $G\in \LL(\SRR{n},\LL(\SRR{m},\SRR{k}))$ defines a separately continuous bilinear map 
$$\SRR{n}\times \SRR{m}\rightarrow \SRR{k},$$
$$\SRR{n}\times \SRR{m} \ni (f,\psi)\mapsto G[f][\psi].$$
Each such map is jointly continuous (corollary of Theorem 34.1 in \cite{Treves}) and defines ${\evaluatePD{1,\ldots,n}}^{-1}G$ such that
$$
{\evaluatePD{1,\ldots,n}}^{-1}G[f\otimes \psi]=G[f][\psi].
$$
\end{proof}
\begin{rmk}
Combining this result with Remark \ref{rmk:Framework/VectorTest}, for $n,m,k,l\in\NN_0$ we may interpret $\LL(\SRR{n+m},\SRR{k+l})$ as the space of \emph{parameter-dependent operator-valued distributions on $\RR{n}$, taking values in $\LL(\SRR{m},\SRR{k})$ of class $\SRR{l}$}.
\end{rmk}
As in Subsection \ref{Framework/VectorTest} the identification is not unique and, in general, we define $\evaluatePD{i_1,\ldots,i_n}F$ analogously. 
\par 
As before, we prefer to perform computations within the spaces of linear operators, but use the interpretation presented above to make their sense more clear. 

\subsection{Restriction of a distribution to fixed values of arguments}\label{Framework/Restrict}
Occasionally we prefer to treat some of the arguments interchangeably as distributional one parametric ones. It is easy to make the latter into the former using the following operation.
\begin{notn}
Take $n,k,m\in\mathbb{N}_0$, $i_1,\ldots,i_k\in \{1,\ldots, m+k\}$, $i'_1,\ldots,i'_k\in \{1,\ldots, n+k\}$ and 
 $F\in\LL(\SRR{n},\SRR{k+m})$. We  define
 \begin{equation}\label{eq:Framework/RestrExt}
     \intD{\arrs{i}{k}}{\arrs{i'}{k}}F\in \LL(\SRR{k+n},\SRR{m}), \quad \intD{\arrs{i}{k}}{\arrs{i'}{k}}F[f](\bm{y})=
 \end{equation}
 $$
 \int \evaluateP{\arrs{i}{k}}{m+k}\left[F\left[\left(\evaluateP{\arrs{i'}{k}}{n+k}\left[f\right]\right)\left(\bm{x}\right)(\bm{y})\right]\right](\bm{x})(\bm{y})d^k\bm{x}, \quad \forall \bm{y}\in\RR{k},\,\forall f\in\SRR{n}.
$$
\end{notn}
The meaning of this operator is that it makes the parameters on the places $i_1,\ldots,i_k$ into distributional arguments inserted on positions $i'_1,\ldots,i'_k$. 
To understand the meaning of the rather formal construction above let us consider a simple example.
\begin{exmp}
Take $m\in\NN$, $h\in\SRR{m}=\LL(\SRR{0},\SRR{m})$. Then
$$\intD{1,\ldots,m}{1,\ldots,m}h[f]=\int{h(\bm{x})f(\bm{x})d^m\bm{x}}.$$
\end{exmp}
We introduce the partially defined inverse operation
$$
\restrD{\arrs{i'}{k}}{\arrs{i}{k}}=\left(\intD{\arrs{i}{k}}{\arrs{i'}{k}}\right)^{-1}.
$$
If for $G\in\LL(\SRR{n+k},\SRR{m})$ the operator $\restrD{\arrs{i'}{k}}{\arrs{i}{k}}G$
exists, we call it a \emph{restriction of $G$} because effectively it allows us to evaluate $G$ at fixed values of its distributional arguments. 
\begin{rmk}
Here we again follow \cite{NN}. Note that the microlocal analysis suggests a more sophisticated notion of restriction. One may show that whenever both are applicable, the results coincide, but the existence conditions are not in general equivalent.
\end{rmk}

\begin{prop}[Analog of Proposition 2.11 in \cite{NN}] For the existence of $$\restrD{\arrs{i'}{k}}{\arrs{i}{k}}G $$ it is sufficient to find a map $F: \SRR{n}\rightarrow \SRR{m+k}$ (not a priori continuous) such that $G[f]=\intD{\arrs{i}{k}}{\arrs{i'}{k}}F[f]$ where the right-hand sides is formally defined by (\ref{eq:Framework/RestrExt}).   
\end{prop}
\begin{proof}
The reasoning is similar to Remark \ref{rmk:Framework/VV=PD}.
Without any loss of generality, we assume
$$
i_{j}=i'_{j}=j, \quad j=1,\ldots,k.
$$
Assume $G\in\LL(\SRR{n+k},\SRR{m})$ and $F$ is as in the statement. Our goal is to prove that $F$ is continuous. 
To do that we evaluate $G$ on a tensor product
\begin{equation}\label{eq:proof:Framework/Restr}
G[h\otimes g](\bm{y})=\intD{1,\ldots,k}{1,\ldots,k}F[h\otimes g](\bm{y})=
\int{F[h](\bm{x},\bm{y})g(\bm{x})d^k\bm{x}}, 
\end{equation}
$$
\forall g\in\SRR{k},
\forall h\in\SRR{n},\,\forall \bm{y}\in\RR{m},
$$
and apply the closed graph theorem. Indeed, assume that there is a sequence $h_{j}\in\SRR{n}$, $j\in\NN$ such that
\begin{equation}\label{eq:proof:Framework/Restr/closedGraph}
    \lim_{j\rightarrow \infty}h_j=h, \, \lim_{j\rightarrow \infty}F[h_j]=u.
\end{equation}
On the one hand, for any $g\in\SRR{k}$ by (\ref{eq:proof:Framework/Restr},\ref{eq:proof:Framework/Restr/closedGraph}) we get
\begin{equation}\label{eq:proof:Framework/Restr/outetLim}
    \lim_{j\rightarrow \infty}G[h_j\otimes g]=G[h\otimes g].
\end{equation}
On the other hand, the continuity of $G$ and $h\mapsto h\otimes g$ for any $g\in\SRR{k}$ together with (\ref{eq:proof:Framework/Restr/closedGraph}) gives
\begin{equation}\label{eq:proof:Framework/Restr/innerLim}\lim_{j\rightarrow \infty}G[h_j\otimes g](\bm{y})=
\int{u(\bm{x},\bm{y})g(\bm{x})d^k\bm{x}}.
\end{equation}
Combining (\ref{eq:proof:Framework/Restr/outetLim}) with  (\ref{eq:proof:Framework/Restr/innerLim}), we conclude that whenever (\ref{eq:proof:Framework/Restr/closedGraph}) is satisfied,
$$
\int{u(\bm{x},\bm{y})g(\bm{x})d^k\bm{x}}=\int{F[h](\bm{x},\bm{y})g(x)d^k\bm{x}}, \,\forall g\in \SRR{k}, \forall g\in\SRR{k},\forall \bm{y}\in\RR{m}.
$$
By assumption $F[h],u\in\SRR{m+k}$, so the last identity is possible only if $F[h]=u$ as desired.
\end{proof}

\subsection{Augmentations and products}\label{Framework/Ext}
Developing Remark \ref{rmk:Framework/VectorTest} we may act by an operator from $\LL(\SRR{n},\SRR{m})$ on a vector from $\SRR{n+k}$ treating the rest $k$ arguments as parameters. This leads to the following construction.
\begin{rmk}\footnote{Variant of Theorem 2.1.3 in \cite{Lars}.}
Let $n,m,k\in\mathbb{N}$, $F\in\LL(\SRR{n},\SRR{m})$. 
Then we may introduce 
$\extD{1,\ldots,k}{1,\ldots,k}F\in \LL(\SRR{n+k},\SRR{m+k}) 
$ as
\begin{equation}
\extD{1,\ldots,k}{1,\ldots,k}F[f](\bm{x},\bm{y})=F[\evaluateP{1,\ldots,k}{n+k}f(\bm{x})](\bm{y}), \qquad \forall f\in\SRR{n+k}, \,\forall \bm{x}\in\RR{k},\,\forall \bm{y}\in\RR{m}
\label{eq:Framework/extSimp}    
\end{equation}
More generally, take two sequences of pairwise different integers  $$i_1,\ldots,i_k\in\{1,\ldots,n+k\}, \qquad
i'_1,\ldots,i'_k\in\{1,\ldots,m+k\}.$$
Then there is $\extD{\arrs{i}{k}}{\arrs{i'}{k}}F\in\LL(\SRR{n+k},\SRR{m+k})$ defined by
\begin{equation}
\extD{\arrs{i}{k}}{\arrs{i'}{k}}F=\left(\evaluateP{\arrs{i'}{k}}{m+k}\right)^{-1}\left(\RR{k}\ni \bm{x}\mapsto F[\evaluateP{\arrs{i'}{k}}{n+k}f(\bm{x})]\right).
\label{eq:Framework/extGen}    
\end{equation}

\end{rmk}
\begin{proof}
Clearly, it is enough to consider $\extD{1,\ldots,k}{1,\ldots,k}F$. Equivalence of (\ref{eq:Framework/extGen}) and (\ref{eq:Framework/extSimp}) follows by construction of $\left(\evaluateP{\arrs{i'}{k}}{m+k}\right)^{-1}$ in Remark \ref{rmk:Framework/VectorTest}. 
\par 
By Remark \ref{rmk:Framework/VV=PD} we only have to show that the right-hand sides of (\ref{eq:Framework/extSimp}) is in $\SRR{m+k}$ as a function of $(\bm{x},\bm{y})$ and is continuous with respect to $f$ when $(\bm{x},\bm{y})$ is fixed. Both facts follow from Remark \ref{rmk:Framework/VectorTest} and continuity of $F$ (which in particular allows to commute it with derivatives with respect to $\bm{x}$ and estimate seminorms of $F[\evaluateP{1,\ldots,k}{n+k}f(\bm{x})]$ in terms of the ones of $f$). 
\end{proof}
\begin{rmk}\label{rmk:Framewor/ExtUniversal}
We use the same notation as in the previous remark.
Equivalently $\extD{1,\ldots,k}{1,\ldots,k}F$ can be characterised as a unique operator in $\LL(\SRR{n+k},\SRR{m+k})$ such that for any $f\in\SRR{k}$ and $g\in\SRR{n}$ holds
\begin{equation}\label{eq:Framework/extProd}
\extD{1,\ldots,k}{1,\ldots,k}F[f\otimes g]=g\otimes F[g].
\end{equation}
\end{rmk}
\begin{proof}
Clearly, (\ref{eq:Framework/extSimp}) implies (\ref{eq:Framework/extProd}). Uniqueness of the operator satisfying (\ref{eq:Framework/extProd}) follows by density of $\SRR{k}\otimes_{\mathrm{alg}}\SRR{n}$ in $\SRR{n+k}$.
\end{proof}
\begin{rmk}\label{rmk:Framework/ProdComm}
From the last remark follows that ``operators acting on different variables commute". In particular, for $F\in\LL(\SRR{n},\SRR{m})$ and $F'\in\LL\left(\SRR{n'},\SRR{m'}\right)$ we have $F\otimes F'\in \LL\left(\SRR{n+n'},\SRR{m+m'}\right)$ defined as
$$
F\otimes F'=\left(\extD{n+1,\ldots,n+m'}{n+1,\ldots,n+m'}F\right)\circ \left(\extD{1,\ldots,n}{1,\ldots,n}F'\right)=\left(\extD{1,\ldots,m}{1,\ldots,m}F'\right)\circ\left(\extD{n+1,\ldots,n+n'}{n+1,\ldots,n+n'}F\right).
$$
This operation coincides with the ordinary tensor product of test functions for $n=n'=0$ and of distributions for $m=m'=0$.
\end{rmk}
\begin{rmk}\label{rmk:Framework/Augment-as-tensor}
The augmentations can be expressed via the tensor product introduced above. In particular, let $n,m,k\in \NN_0$. Then
$$
\left(\extD{n+1,\ldots,n+k}{m+1,\ldots,m+k}F\right)=F\otimes \mathbb{1}_{\LL(\SRR{k})},
$$
$$
\left(\extD{1,\ldots,k}{1,\ldots,k}F\right)
= \mathbb{1}_{\LL(\SRR{k})}\otimes F.
$$
\end{rmk} 
\begin{rmk}\label{rmk:Framework/Augment/algebra}
The following two simple properties of augmentation are useful:
\begin{itemize}
    \item Consecutive augmentations can be combined, e.g.  for $n,m\in\NN_0$ and $F\in\LL(\SRR{n},\SRR{m})$
    $$
    \extD{1,\ldots,l}{1,\ldots,l}(\extD{1,\ldots,k}{1,\ldots,k}F)=\extD{1,\ldots,l+k}{1,\ldots,l+k}F,\,\forall k,l\in\NN_0; 
    $$
    \item Augmentation is distributive with respect to compositions, e.g. for $n,m\in\NN_0$ and $F\in\LL(\SRR{n},\SRR{m})$, $G\in\LL(\SRR{k},\SRR{n})$
    $$
    \extD{1,\ldots,k}{1,\ldots,k}(F\circ G)=(\extD{1,\ldots,k}{1,\ldots,k}F)\circ (\extD{1,\ldots,k}{1,\ldots,k}G), \forall k\in\NN_0.
    $$
    \item As a consequence of the above and Remark \ref{rmk:Framework/ProdComm}, if $n,n',n'',m,m',m''\in\NN_0$, 
    $$F\in\LL\left(\SRR{n'},\SRR{n''}\right),\,F'\in\LL\left(\SRR{n},\SRR{n'}\right),$$
    $$G\in\LL\left(\SRR{m'},\SRR{m''}\right),\,G'\in\LL\left(\SRR{m},\SRR{m'}\right),$$
    then
    $(F\otimes G)\circ (F'\otimes G')=(F\circ F')\otimes (G\circ G')$.
\end{itemize}
\end{rmk}

\subsection{Multipliers}\label{Framework/Mlt}
For now, we considered only parameter-dependent distributions fast decaying at the high values of the parameters. This condition is too restrictive. For example, the free quantum field is known to be an operator-valued distribution restrictable to the fixed values of time (see \cite{SimonReed} for standard treatment and Subsection \ref{DS/2Q} for construction in our formalism) but it does not decay in time. In this subsection we define a large enough space containing such objects.
\par
We recall that the space of multipliers \cite{NN}
$\Mlt{n}$ consists of functions having limited growth together with all their derivatives,
$$
\Mlt{n}=\{f\in C^{\infty}(\RR{n})| \forall k\in\mathbb{N}_0\, \exists  m\in \mathbb{N}_0: \, \sup_{x\in\RR{n},|a|\leq k} (1+||x||)^{-m}|\partial^{a}f(x)|\leq \infty \}.
$$
The name comes from the following alternative characterization of $\Mlt{n}$. For a function $f$ on $\RR{n}$ define pointwise multiplication by $f$ as

$$
(\MltS{f} g)(\bm{x})=f(\bm{x})g(\bm{x})\, (\forall g\in\SRR{n}, \forall \bm{x}\in\RR{n}). 
$$
Then\footnote{See \cite{NN} and \cite{Amman}, or generalization of this fact in Remark \ref{rmk:Framework/Mlt/Alt} below.} 
$$
\Mlt{n}=\{f:\RR{n}\rightarrow \mathbb{C}| \MltS{f}\in \LL(\SRR{n})\}.
$$
Note that by Remark \ref{rmk:Framework/VV=PD} it is enough to show that $\MltS{f}$ leaves $\SRR{n}$ invariant. We also define a complementary operator $\MltM{g}: \Mlt{n}\rightarrow \SRR{n}$ defined for each $g\in\SRR{n}$ and defined by
$$
\MltM{g}[f]=f\cdot g. 
$$
\par 
There is a natural topology on $\Mlt{n}$ generated by seminorms\footnote{We follow \cite{Amman}.}
$$
||f||_{\Mlt{n},g,m,k}=||\MltM{g} f||_{\SRR{n},m,k},\quad g\in\SRR{n},\, m,k\in\mathbb{N}_0
$$
making all $\MltM{g}$ for each $g\in\SRR{n}$ continuous.
Conversely, for some topological vector space $X$ an operator $F: X\rightarrow \Mlt{n}$ is continuous whenever 
$$
\MltM{g}\circ F\in\LL(X,\SRR{n}), \quad \forall g\in\SRR{n}.
$$
This way we can always replace $\Mlt{n}$-valued operators with $\SRR{n}$ ones. We postpone it to Remark \ref{rmk:Framework/Mlt/Strategy} to introduce first a generalization of the space of multipliers.
\par 
For $m,n\in\NN_0$ we introduce the mixed space $\SMlt{n}{m}$ as 
$$
\SMlt{n}{m}=\{f\in C^{\infty}(\RR{n+m})| \forall k,l\in\mathbb{N}_0\,$$
$$ \exists  r\in \mathbb{N}_0: \, \sup_{\substack{\bm{x}\in\RR{n},y\in\RR{m},\\|a|+|b|\leq k}} (1+||\bm{x}||)^{l}(1+||\bm{y}||)^{-r}|\partial_{\bm{x}}^{a}\partial_{\bm{y}}^{b}f(\bm{x},\bm{y})|\leq \infty \}.
$$
\begin{rmk}\label{rmk:Framework/Mlt/Alt}
There is an equivalent 
characterization of $\SMlt{n}{m}$, $n,m\in\NN{0}$ as a space of all functions $f$ that define a map $\MltS{f}:\SRR{m}\rightarrow \SRR{n+m}$ by pointwise multiplication
$$
\MltS{f}[g](\bm{x},\bm{y})=f(\bm{x},\bm{y})g(\bm{y}).
$$
Furthermore,
$$\MltS{f}\in\LL(\SRR{m},\SRR{n+m}), \forall f\in\SMlt{n}{m}.$$
For any $g\in\SRR{m}$ there is a map $\MltMn{g}{n}\in\LL(\SMlt{n}{m},\SRR{m+n})$ such that
$$
\MltMn{g}{n}[f]=\MltS{f}[g],
$$
where the topology on $\SMlt{n}{m}$ is defined by 
$$
||f||_{\SMlt{n}{m},g,l,k}=||\MltS{f}[g]||_{\SRR{n+m},l,k},\quad g\in\SRR{m},\, k,l\in\NN_0.
$$
Conversely, for an operator $F:X\rightarrow \SMlt{n}{m}$ with some topological vector space $X$ we have $F\in\LL(X,\SMlt{n}{m})$ if and only if
$$
\MltS{g}\circ F\in\LL(X,\SRR{n+m})\,\forall g\in\SRR{m}.
$$
\end{rmk}
\begin{proof}\footnote{The presented proof follows \cite{Amman} with minimal changes.}
If $f\in\SMlt{n}{m}$ and $g\in\SRR{m}$, then $\MltS{f}[g]\in\SRR{n+m}$ essentially due to (\ref{eq:xy<x}). Conversely, if a function $f$ on $\RR{n+m}$ induces a map $\MltS{f}: \SRR{m}\rightarrow \SRR{m}$, then by Remark \ref{rmk:Framework/VV=PD} this map is continuous. Let $g\in\SRR{m}$ be such that 
$$g(\bm{x})=1, \forall \bm{x}\in U$$ 
for some neighborhood of the origin $U$ and for $\bm{y}\in\RR{m}$ set $g_{\bm{y}}\in\SRR{m}$, $g_{\bm{y}}(\bm{x})=g(\bm{x}-\bm{y})$. Then 
$$
f(\bm{x})=\MltS{f}[g_y](\bm{x}), \quad \forall \bm{x}\in \bm{y}+U.
$$
The smoothness of $f$ is clear. By continuity of $\MltS{f}$ for any multi-index $\alpha\in\NN_0^{n+m}$ and any $r\in\NN{0}$ there are some $C>0$ and $l,b\in\NN_0$ such that for every $x\in\RR{n}$
$$
|\partial_{\alpha}f(\bm{x})|(1+||\bm{x}||)^r \leq C ||g_{\bm{y}}||_{\SRR{m},l,b}.
$$
Finally, for some $C'>0$
$$
||g_{\bm{y}}||_{\SRR{m},l,b}\leq C'(1+||\bm{y}||)^{l}
$$
by 
$$1+||\bm{x}-\bm{y}||\leq (1+||\bm{x}||)(1+||\bm{y}||).$$
Thus $f\in\SMlt{n}{m}$. The rest of the statement is trivial.
\end{proof}

\begin{rmk}\label{rmk:Framework/Mlt/VectorTest}
One can prove an analogue of Remark \ref{rmk:Framework/VectorTest}, identifying the $\SRR{n}$-valued vector functions of class $\SMlt{m}{k}$ with elements of $\SMlt{n+m}{k}$.
\end{rmk}
\begin{rmk}\label{rmk:Framework/Mlt/Strategy}
The key strategy of dealing with the operators like $$F: \SRR{n}\rightarrow \SMlt{m}{k}$$
for $n,m,k\in\NN_0$ is to always work with $\MltM{f}\circ F$ instead, where $f\in\SRR{n}$. For example, similarly to Subsection \ref{Framework/VV=PD} we may consider the space of parameter-dependent distributions $\mathcal{O_M}(\RR{m},\SpRR{n})$ defined along the same lines as $\mathcal{S}(\RR{m},\SpRR{n})$. Any  $F\in \mathcal{O_M}(\RR{m},\SpRR{n})$ defines a map $F:\SRR{n}\rightarrow \Mlt{m}$ which by Remark \ref{rmk:Framework/Mlt/Alt}
is continuous if $\MltM{f}\circ F\in\LL(\SRR{m},\SRR{n})$ for every $f\in\SRR{n}$, for which we use Remark \ref{rmk:Framework/VV=PD}. This way we get that Subsections \ref{Framework/VV=PD} - \ref{Framework/Restrict} can be generalized to the case then the target space is of the form $\SMlt{m}{k}$, $m,k\in\NN_0$.
\end{rmk}
We should separately take care of the augmentations of Subsection \ref{Framework/Ext}, because in general operators acting from the multiplier and mixed spaces may appear. We consider only one simple case, the generalizaiton is straightforward.
\begin{exmp}
 Let $m,n,k\in \NN_0$ and $F\in\LL(\SRR{m}, \SRR{n})$.  Then  $\extD{m+1,\ldots,m+k}{n+1,\ldots,n+k}F$ can be continued to an element of $\LL(\SMlt{m}{k},\SMlt{n}{k})$ such that
 \begin{equation}
\extD{m+1,\ldots,m+k}{n+1,\ldots,n+k}F[g](\bm{x},\bm{y})=F[\evaluateP{m+1,\ldots,m+k}{m+k}g(\bm{y})](\bm{x}),
\label{eq:Framework/Mlt/Ext-Expm}  
\end{equation}
$$\forall \bm{x}\in\RR{n},\forall \bm{y}\in\RR{k},\forall g\in\SMlt{m}{k}.
 $$
 Alternatively, it can be defined by the relation
 \begin{equation}
 \MltMn{f}{n}\circ\extD{m+1,\ldots,m+k}{n+1,\ldots,n+k}F=
 \extD{m+1,\ldots,m+k}{n+1,\ldots,n+k}F\circ \MltMn{f}{m}, \forall f\in\SRR{k},
 \label{eq:Framework/Mlt/Ext-Formal}  
 \end{equation}
 
 where in the right-hand sides  $\extD{m+1,\ldots,m+k}{n+1,\ldots,n+k}F\in\LL(\SRR{m+k},\SRR{n+k})$ is defined as in Subseciton \ref{Framework/Ext}.
\end{exmp}
\begin{proof}
First take some $f\in\SRR{k}$ evaluate the right-hand sides of (\ref{eq:Framework/Mlt/Ext-Formal}) on a function $g\in\SMlt{m}{k}$:
$$
\extD{m+1,\ldots,m+k}{n+1,\ldots,n+k}F\circ \MltMn{f}{m}[g](\bm{x},\bm{y})=F[\evaluateP{m+1,\ldots,m+k}{m+k}g(\bm{y})](\bm{x})f(\bm{y}).
$$
The left-hand side is a Schwartz function of $(\bm{x},\bm{y})$, so pointwise multiplication by a function $\extD{m+1,\ldots,m+k}{n+1,\ldots,n+k}F[g]$ defined by (\ref{eq:Framework/Mlt/Ext-Expm}) induces a map $\SRR{k}\rightarrow \SRR{k+n}$, so by Remark \ref{rmk:Framework/Mlt/Alt} we have $\extD{m+1,\ldots,m+k}{n+1,\ldots,n+k}F[g]\in\SMlt{n}{k}$ and it automatically is the unique solution of (\ref{eq:Framework/Mlt/Ext-Formal}). The right-hand side of (\ref{eq:Framework/Mlt/Ext-Formal}) belongs to $\LL(\SMlt{n}{k},\SRR{m+k})$, so applying once again Remark \ref{rmk:Framework/Mlt/Alt} we get $\extD{m+1,\ldots,m+k}{n+1,\ldots,n+k}F\in\LL(\SMlt{n}{k},\SMlt{m}{k})$.
\end{proof}

\subsection{Elementary examples}\label{Framework/exmp}
A lot of standard operations often arising in the distributions theory and QFT are operators in $\LL\left(\SMlt{n}{m},\SMlt{n'}{m'}\right)$ for some $n,m,n',m'\in\NN_0$. We have already constructed the pointwise multiplication operators. Here we list some other classes of operators which are convenient to use later in the text.
\paragraph{Partial derivatives} For $n\in\NN$ and a multi-index $\alpha\in\NN_0^{n}$ we have $\partial^{\alpha}\in\LL(\SRR{n},\SRR{n})$. Likewise, the derivatives can be defined for 
$\SMlt{n}{m}$, $n,m\in\NN_0$

\paragraph{Integration with moving boundary}
We can define $J_{+},J_{-}\in\LL(\SRR{})$ and $J\in\LL(\SRR{},\SRR{2})$ as
\begin{equation}\label{eq:Framework/exmp/J-}
J_{-}[f](t)=\int_{-\infty}^{t}f(t)dt, \, \forall t\in\RR{}, \forall f\in\SRR{};
\end{equation}
\begin{equation}
    J_{+}[f](t)=\int_{t}^{+\infty}f(t)dt, \, \forall t\in\RR{}, \forall f\in\SRR{};
\end{equation}
\begin{equation}\label{eq:Framework/exmp/J}
J_{}[f](t,t')=\int_{t}^{t'}f(t)dt, \, \forall t,t'\in\RR{}, \forall f\in\SRR{}.
\end{equation}

These operators interact with the derivatives in obvious ways, e.g. 
$$\partial\circ J_{+}=-J_{+}\circ \partial=\mathbb{1}_{\LL(\SRR{})},$$
\begin{equation}\label{eq:Framework/exmp/Int-dif}
\partial_1\circ J=-\partial_2\circ J=\mathbb{1}_{\LL(\SRR{})}.    
\end{equation}
\par 
It is worth noting that the trivial property $J_{+}[f](t)=\lim_{t'\rightarrow +\infty}J[f](t,t')$ survives the augmentation, i.e. for any $n\in\NN$ and $f\in\SRR{n}$ we have 
\begin{equation}\label{Framework/exmp/IntLim}
    \lim_{t'\rightarrow \infty} \evaluateP{2}{1}(\extD{2,\ldots,n}{3,\ldots,n+1}J[f])(t')=\extD{2,\ldots,n}{2,\ldots,n}J_{+}[f],
\end{equation}
and similarly for analogous relations.
\par 
For future use we introduce more complicated operators
$
\intSimplex{n}\in\LL(\SRR{n},\SRR{2})
$
and $\intUSimplex{n}{\pm}\in\LL(\SRR{n},\SRR{1})$
defined by
$$
\intSimplex{n}[f](t,t')=\int_{t}^{t'}dt_1\int_{t}^{t_2}dt_2\cdots \int_{t}^{t_{n-1}}dt_n f(\arrs{t}{n}),\, \forall f\in\SRR{n},\,\forall t,t'\in\RR{};
$$
$$
\intUSimplex{n}{+}[f](t)=\int_{t}^{+\infty}dt_1\int_{t}^{t_2}dt_2\cdots \int_{t}^{t_{n-1}}dt_n f(\arrs{t}{n}),\, \forall f\in\SRR{n},\,\forall t,t'\in\RR{}
$$
$$
\intUSimplex{n}{-}[f](t)=\int_{-\infty}^{t}dt_1\int_{t}^{t_2}dt_2\cdots \int_{t}^{t_{n-1}}dt_n f(\arrs{t}{n}),\, \forall f\in\SRR{n},\,\forall t,t'\in\RR{}.
$$
Analogues of  (\ref{Framework/exmp/IntLim}) can be written for this class of operators.

\begin{rmk}\label{rmk:Framework/exmp/Strong}
If we treat $\SRR{n+k}$ as $\mathcal{S}(\RR{k},\SRR{n})$ and think of $\SRR{n}$ as a subspace of $L^2(\RR{n})$, then the partial derivatives with respect to the parameters define partial derivatives of the vector-valued functions with respect to the parameters. Moreover, the integration with respect to the parameters\footnote{Here we suppose that the first $k$ arguments are parameters as in Subsection \ref{eq:Framework/extProd}} such as
$$
\extD{k+1,\ldots,k+n}{k+1,\ldots,k+n}J[f] \, (f\in\SRR{n+k})
$$
coincides with the Bochner integral of the vector-valued function $\evaluateP{1,\ldots,k}{1,\ldots,k}f$. This follows from approximation of $f\in\SRR{n+k}$ by
$$
f=\lim_{n\rightarrow \infty} g_n\otimes h_n, g_n\in\SRR{k}, h_n\in\SRR{n}
$$
and then an approximation of $h_n$ by sequences of the simple functions. 
\end{rmk}

\paragraph{Precompositions}
Clearly, if $A$ is an injective\footnote{If $A$ is not injective, then $A^*$ is valued in the multiplier or mixed space} linear operator\footnote{We limit ourselves to the linear functions for simplicity. In general a function should belong to the multipliers space (see Section \ref{Framework/Mlt}) together with its inverse, see Exercise 2.8 of \cite{NN}.} $\RR{n}\rightarrow\RR{m}$, then it defines an operator\footnote{This notation, of course, is in conflict with the one for adjoint operators. The later, however, is used here only episodically, since another concept of conjugated operator is more important for us. Anyway, we never introduce adjoint of a finite-dimensional operators, so no confusion is possible.} $A^*\in\LL(\SRR{m},\SRR{n}),$
$$
A^*[f](\bm{x})=f(A\bm{x}), \, \forall\in\SRR{n}, \, \forall\bm{x}\in\RR{m}.
$$
In particular, the diagonal map $D_n: \RR{n}\rightarrow\RR{2n}$, 
$$D_n(\bm{x})=(\bm{x},\bm{x}), \, \forall \bm{x}\in\RR{n}$$ defines $\DiagM{n}\in\LL(\SRR{2n},\SRR{n})$ such that $\DiagM{n}f(\bm{x})=f(\bm{x},\bm{x})$ for $f\in\SRR{n}$ and $\bm{x}\in\RR{n}$.
\par
Another important class of maps on comes from the permutation group. For $\sigma\in\symmgr{n}$ and $k\in\NN$ we set $(\sigma)_{(k)}: \RR{n\cdot k}\rightarrow \RR{n\cdot k}$ as $(\sigma)_{(k)}(\arrs{\bm{x}}{n})=(\arrsP{\bm{x}}{n}{\sigma}),\, \bm{x}_j\in\RR{k}, j=1,\ldots, n$. We define $\permK{k}{\sigma}=(\sigma^{-1})_{(k)}^*$ defining a representation of $\symmgr{n}$ on $\SRR{n\cdot k}$ and 
$$
\symmze{k}{n}=\frac{1}{n!}\sum_{\sigma\in\symmgr{n}}\permK{k}{\sigma}.
$$
%\paragraph{Shift}
%For $A\in\LL(\RR{n},\RR{m})$ we define $S_A\in \LL(\SRR{m},\SRR{n+m})$ as
%$$
%S_A[f](\bm{x},\bm{y})=f(\bm{x}+A\bm{y}),\,\fora%ll f\in\SRR{m},\forall x\in %%
%\SRR{m},\forall{y}\in\SRR{n+m}. 
%$$
\paragraph{Functions as operators, tensor and pointwise multiplication}
Recall that we identify $\SRR{n}$ with $\LL(\SRR{0},\SRR{n})$ by treating 
 $f\in\SRR{n}$ as a map  
$$\mathbb{C}\ni c\mapsto c f.$$
\par 
The augmentation define the tensor multiplication by $f$, e.g. $$\extD{1,\ldots,m}{1,\ldots,m}f[g]=g\otimes f,\forall f\in\SRR{m},\forall g\in\SRR{n}$$.
The following example illustrates usage of this operation together with the pointwise multiplication of Subsection \ref{Framework/Mlt}. The subsequent remark points out the issue of symbolic notation of the next subsection which one should be aware of.
\begin{exmp}\label{exmp:Framework/2products}
There is an important relation between the tensor multiplication above and the pointwise multiplication of Subsection \ref{Framework/Mlt}. For example, let $n,m,k\in\NN_0$, $f\in\SRR{n+k}$ and consider a map $F\in\LL(\SRR{m+k},\SRR{n+m+k})$
$$
F[g](\bm{x},\bm{y},\bm{z})=g(\bm{x},\bm{y})f(\bm{y},\bm{z}), \,\forall g\in\SRR{m+k}, \forall \bm{x}\in\RR{m},\forall \bm{y}\in\RR{k},\forall \bm{z}\in \RR{n}.
$$
On the one hand, we may first tensor with $f$
$$
\extD{1,\ldots,m+k}{1,\ldots,m+k}f[g](\bm{x},\bm{y},\bm{y'},\bm{z})=f(\bm{x},\bm{y})g(\bm{y}',\bm{z})\,\forall g\in\SRR{m+k}, \forall \bm{x}\in\RR{m},\forall \bm{y},\bm{y}'\in\RR{k},\forall \bm{z}\in \RR{n},
$$
and then collapse $\bm{y}$ and $\bm{y}'$ into one variable,
\begin{equation}
\label{eq:Framework/exmp/mult=tensor}
F=\extD{1,\ldots,m,m+2k+1,\ldots,m+2k+n}{1,\ldots,m,m+k+1,\ldots,m+k+n}\DiagM{k}\circ \extD{1,\ldots,m+k}{1,\ldots,m+k}f.
\end{equation}
On the other hand, we may start from tensoring with a constant function by $\extD{1,\ldots,m+k}{1,\ldots,m+k}\mathbb{1}_{\Mlt{n}}\in\LL(\SRR{m+k},\SRR{n+m+k}),$
$$
\extD{1,\ldots,m+k}{1,\ldots,m+k}\mathbb{1}_{\Mlt{n}}[g](\bm{x},\bm{y},\bm{z})=g(\bm{x},\bm{y})\,\forall g\in\SRR{m+k}, \forall \bm{x}\in\RR{m},\forall \bm{y}\in\RR{k},\forall \bm{z}\in \RR{n},
$$
and then multiply by $f$, ignoring first $n$ variables ($\bm{x}$ in the notation above), getting 
\begin{equation}
F=\extD{1,\ldots,n}{1,\ldots,n}\MltS{f}\circ \extD{1,\ldots,m+k}{1,\ldots,m+k}\mathbb{1}_{\Mlt{n}}.\label{eq:Framework/exmp/tensor=mult}    
\end{equation}
\end{exmp}
\begin{rmk}\label{rmk:exmp:Framework/2products}
In the setting of the previous example the forms (\ref{eq:Framework/exmp/mult=tensor}) and (\ref{eq:Framework/exmp/tensor=mult}). However, only the later generalizes easily to the case $f\in\Mlt{n}$. Indeed, the map $\DiagM{k}$ does maps $\SMlt{k}{k}$ into $\Mlt{k}$ rather than $\SRR{k}$, so in the form (\ref{eq:Framework/exmp/mult=tensor}) it takes a separate effort to prove that $F\in\LL(\SRR{m+k},\SRR{n+m+k})$. At the same time, $\MltS{f}$ can be replaced by  $\MltM{f}$ which extends to a continuous operator $\Mlt{n}\rightarrow\Mlt{n}$ for obvious reasons.
\end{rmk}

\paragraph{The Fourier transform}
$$
\Fourier{n}{\eta}[f](\bm{x})=\int{e^{\eta \ii \bm{k}\cdot \bm{x}}f(\bm{k})d^{k}\bm{x}}, \forall f\in \SRR{n},\forall \bm{k}\in\RR{n}
$$
is a well-known invertible operator in $\SRR{n}$. Here $\eta=\pm 1$.
Note that $\Fourier{n}{\eta}$ has a restriction $\restrD{1,\ldots,n}{n+1,\ldots,2n}\Fourier{n}{\eta}\in\Mlt{2n}$,
\begin{equation}
\restrD{1,\ldots,n}{n+1,\ldots,2n}\Fourier{n}{\eta}(\bm{x},\bm{k})=e^{i\eta\bm{k}\cdot \bm{x}}.
\label{eq:Framework/exmp/Fourier}    
\end{equation}
\subsection{Symbolic notation and calculus}\label{Framework/Symbolic}
In this subsection we generalize the standard symbolic integral notation from the distributions theory
\begin{equation}
F[f]=\int F(\bm{x})f(\bm{x})d^n\bm{x}, \, F\in\SpRR{n}, \,f\in\SRR{n}.
\label{Framework/symb/std}  
\end{equation}
We introduce all the notions and operations for the spaces $\SRR{n}$, $n\in\NN_0$, but use their generalization to the multiplier and mixed spaces as explained in Subsection \ref{Framework/Mlt}. Before giving it in the full generality let us first make one step from (\ref{Framework/symb/std})
\begin{notn}\label{notn:symb-dist-eval}
We write
\begin{equation}
F[f](\bm{y})=\int F\symbDistArgs{\bm{x}}{\bm{y}}f(\bm{x})d^n\bm{x}, 
\label{eq:Framework/Symb/F[f]}    
\end{equation}
where $F\in\LL(\SRR{n},\SRR{m})$, $f\in\SRR{n}$ and $y\in\RR{m}$.
We call $\bm{x}$ the \emph{distributional argument} in contrast to the \emph{parameters} $\bm{y}$ of $F$. The compose symbol $F\symbDistArgs{\bm{x}}{\bm{y}}$ we call the \emph{symbolic kernel} of $F$.
\end{notn}
Note that  (\ref{Framework/symb/std}) reduces to (\ref{eq:Framework/Symb/F[f]}) if $m=0$.
\begin{rmk}
By the Schwartz kernel theorem we can could treat the symbolic kernel of $F\in\LL(\SRR{n},\SRR{m})$ as an element of $\SpRR{n+m}$, keeping (\ref{eq:Framework/Symb/F[f]}). But this is not in the line with the goals we set in the beginning of the section: to use minimally special properties of the Schwartz spaces and to highlight that all the objects we deal with are, in one way or another, linear operators between the Schwartz spaces.
\end{rmk}
The notation (\ref{eq:Framework/Symb/F[f]}) resembles action of a matrix on a vector.
So, it is natural to use the matrix product-like notation for the symbols of compositions of operators. 
\begin{notn}\label{notn:symb-dist-compose}
Let $n\in\NN$. For each $j=1,\ldots,n$ take $F_j\in\LL(\SRR{m_j},\SRR{m_{j-1}})$, where $m_j\in\NN_0$, $j=0,\ldots,n$. Then we write
\begin{equation}\label{eq:notn:symb-dist-compose}
(F_1\circ F_2\circ\cdots\circ F_n)\symbDistArgs{\bm{x}_{n}}{\bm{x_0}}=    
\end{equation}
$$\int F_1\symbDistArgs{\bm{x}_{1}}{\bm{x_0}}F_1\symbDistArgs{\bm{x}_{2}}{\bm{x_1}}\cdots F_n\symbDistArgs{\bm{x}_{n-1}}{\bm{x_n}}\prod_{j=1}^{n-1}d^{m_j}\bm{x}_j,\,\bm{x}_0\in\RR{m_0},\,\bm{x}_n\in\RR{m_n}.
$$
\end{notn}
We wrote for one time the product of the symbolic kernels explicitly to underline that it is ordered in a particular way. As these expressions are formal, we have no right to permute the multipliers for now, eventhough essentially they are numerical distributions.  At the same time, we may safely use the ``Fubini's theorem" for symbolic integrals by construction (and associativity of the composition). As we make no distinguishment between $\SRR{n}$ and $\LL(\SRR{0},\SRR{n})$, we naturally identify $F\in \LL(\SRR{0},\SRR{n})$ with the function $F\in\SRR{n}$. Then Notation \ref{notn:symb-dist-compose} reduces to Notation \ref{notn:symb-dist-eval} when $n=2$ and $m_2=0$. We underline that the ordering of symbolic kernels in (\ref{eq:notn:symb-dist-compose}) is not arbitrary, instead, it reproduces the order of operators in the composition. 
\par 
The next step is to introduce separate notation for special classes of operators.
\begin{notn}\label{notn:Framework/symb/stdop}
Fix $n,m\in\NN_0$ and $F\in\LL(\SRR{n},\SRR{m})$. We set:
\begin{itemize}
    \item For $\alpha\in\NN_0^m$,
    $$\partial_{\bm{y}}^{\alpha}F\symbDistArgs{\bm{x}}{\bm{y}}=(\partial^{\alpha}\circ F )\symbDistArgs{\bm{x}}{\bm{y}} (\bm{x}\in\RR{n},\,\bm{y}\in\RR{m});$$
    \item For an injective map $A\in\LL(\RR{k},\RR{m})$ we set\footnote{As before, the non-injective maps are allowed, but lead to multiplier (or mixed) space-valued operators}
    $$
    F\symbDistArgs{A\bm{x}}{\bm{y}}=({A}^*\circ F)\symbDistArgs{\bm{x}}{\bm{y}}\, (\bm{x}\in\RR{n},\,\bm{y}\in\RR{m});
    $$
    \item For $f\in\Mlt{m+k}$ ($k\in\NN_0$) we set\footnote{Clearly, $\MltM{f}\circ \extD{1,\ldots,m}{1,\ldots,m}\mathbb{1}_{\Mlt{k}}\circ F\LL(\SRR{n},\SRR{n+m+k})$. We add this case despite our decision to not discuss the multiplier and mixed spaces, because it concerns the issues not appearing for the Schwartz spaces. We will need this operation to deal with quantum fields and Hamiltonian densities.}
    $$
    f(\bm{y},\bm{z})F\symbDistArgs{\bm{x}}{\bm{y}}=(\MltM{f}\circ \extD{1,\ldots,m}{1,\ldots,m}\mathbb{1}_{\Mlt{k}}\circ F)\symbDistArgs{\bm{x}}{\bm{y},\bm{z}} (\bm{x}\in\RR{n},\,\bm{y}\in\RR{m},\bm{z}\in\RR{k});
    $$
    \item For $n',m'\in\NN_0$ and $F'\in \LL\left(\SRR{n'},\SRR{m'}\right)$:
    $$
    (F\otimes F')\symbDistArgs{\bm{x},\bm{x'}}{\bm{y},\bm{y'}}=
    F\symbDistArgs{\bm{x}}{\bm{y}}F'\symbDistArgs{\bm{x}'}{\bm{y}'},\,\bm{x}\in\RR{n},\bm{y}\in\RR{m},\bm{x}'\in\RR{n'},\bm{y}'\in\RR{m'}.
    $$
\end{itemize}
\end{notn}
\begin{rmk}
Let us look on the consistency issues of the notation above. 
\begin{itemize}
    \item Notation \ref{notn:Framework/symb/stdop} becomes tautological if we put $m=0$;
    \item If $F$ and $A$ are as in Notation \ref{notn:Framework/symb/stdop} and $B\in\LL(\RR{l},\RR{k})$, then
    $$
    F\symbDistArgs{\bm{x}}{BA\bm{y}}
    $$
    has two interpretations
    $$
    ((BA)*\circ F)\symbDistArgs{\bm{x}}{\bm{y}} \,\rm{and}\, (B^*\circ A^*\circ F)\symbDistArgs{\bm{x}}{\bm{y}}.
    $$
    But clearly $(BA)^*=B^*\circ A^*$, so the interpetations coincide;
    \item for $F,F',A$ as in Notation \ref{notn:Framework/symb/stdop} and $A'\in\LL\left(\RR{k'},\RR{m'}\right)$ the expression 
    $$
    F\symbDistArgs{\bm{x}}{A\bm{y}}F'\symbDistArgs{\bm{x}'}{A'\bm{y}'}
    $$
    has two possible interpretations:
    $$
    (A\oplus A')^*\circ (F\otimes F') \quad \rm{and}\quad  (A^*\circ F)\otimes ({A'}^* \circ F')=(A^*\otimes {A'}^*)\circ (F\otimes F'),
    $$
    where we have used Remark \ref{rmk:Framework/Augment/algebra} to rewrite the second version. By a direct computation $(A\oplus A')^*=A^*\otimes {A'}^*$, so the interpretations are equivaent;
    \item For $F$ as in Notation \ref{notn:Framework/symb/stdop} and $f\in\SRR{m+k}$  expression
    $$
    f(\bm{y},\bm{z})F\symbDistArgs{\bm{x}}{\bm{y}}
    $$
    has two possible interpretations:
    $$
    (\MltM{f}\circ \extD{1,\ldots,m}{1,\ldots,m}\mathbb{1}_{\Mlt{k}}\circ F)\symbDistArgs{\bm{x}}{\bm{y},\bm{z}}  \,\rm{and}\, \extD{k+1,\ldots,k+m,2k+m+1,\ldots,2k+m+n}{k+1,\ldots,k+m+n}\DiagM{k}\circ(f\otimes F).
    $$
    Here in the first case we treat $f$ as an element of $\Mlt{m+k}\supset \SRR{m+k}$, while in the second case as an operator in $\LL(\SRR{0},\SRR{m+k})$.
    They coincide by Example  \ref{exmp:Framework/2products}. As explained in Remark \ref{rmk:exmp:Framework/2products}, the separate definition of pointwise multiplication for the multiplier spaces can not be discarded.
\end{itemize}
\end{rmk}
\begin{rmk}\label{rmk:Framework/symb/stdop-props}
The operations introduced above satisfy the following:
\begin{itemize}
    \item The partial derivatives commute among themselves;
    \item Derivatives interact with the linear changes of variables as usual, according to the chain rule;
    \item The Leibnitz rule for partial derivatives hods;
    \item  For $F,F'$ as in Notation \ref{notn:Framework/symb/stdop}
    \begin{equation}\label{eq:Framework/ProdComm-symb}
   F\symbDistArgs{\bm{x}}{\bm{y}}F'\symbDistArgs{\bm{x}'}{\bm{y}'}= F'\symbDistArgs{\bm{x}'}{\bm{y}'} F\symbDistArgs{\bm{x}}{\bm{y}}      
    \end{equation}
    \item Let $n,n'\in\NN$. Take
    $$
    F_j\in\LL\left(\SRR{m_j},\SRR{m_{j-1}}\right), j=1,\ldots,n
    $$
    $$
    F'_j\in\LL\left(\SRR{m'_j},\SRR{m'_{j-1}}\right), j=1,\ldots,n',
    $$
    where $\arrsM{m}{0}{n},\arrsM{m'}{0}{n'}\in\NN_0$. Then
    \begin{equation}
    \left(\int \left(\prod_{j=1}^{n}F_j\symbDistArgs{\bm{x}_{j}}{\bm{x}_{j-1}}\right) \prod_{j=1}^{n-1} d^{m_j}\bm{x}_j\right)
    \left(\int \left(\prod_{j=1}^{n'}F'_j\symbDistArgs{\bm{x}'_{j}}{\bm{x}'_{j-1}}\right) \prod_{j=1}^{n'-1} d^{m'_j}\bm{x}'_j\right)=
    \label{eq:symb-tensor-compos-distr}
    \end{equation}
    $$\int 
     \left(\prod_{j=1}^{n}F_j\symbDistArgs{\bm{x}_{j}}{\bm{x}_{j-1}}\right)  \left(\prod_{j=1}^{n'}F'_j\symbDistArgs{\bm{x}'_{j}}{\bm{x}'_{j-1}}\right) \left(\prod_{j=1}^{n}d^{m_j}\bm{x}_j\right) \left(\prod_{j=1}^{n'}d^{m'_j}\bm{x}'_j\right),
    $$
    $$
    \bm{x}_0\in \RR{m_0}, \quad \bm{x}_n\in \RR{m_n}
    $$
\end{itemize}
\end{rmk}
\begin{proof}
Most of the statements are trivial, so we only mark the way of the proof for some of the claims. The Leibnitz rule for the defining form of symbolic expression (\ref{eq:notn:symb-dist-compose}) is trivial by definition of the partial derivatives. To deal with tensor products use 
$$
\partial^{(a,b)}=\partial^{a}\otimes \partial^{b},
$$
where $a\in\NN_0^m$ and $b\in\NN_0^{m'}$ are two multiindices and $(a,b)\in\NN_0^{m+m'}$ is their concatenation ($m,m' \in\NN_0$) and the
Remark \ref{rmk:Framework/Augment/algebra}. 
\par 
 The commutativity (\ref{eq:Framework/ProdComm-symb}) is due to Remark \ref{rmk:Framework/ProdComm}. The relation (\ref{eq:symb-tensor-compos-distr}) holds By Remark \ref{rmk:Framework/Augment/algebra}.
\end{proof}

It is natural to write the symbolic kernel of $\mathbb{1}_{\LL(\SRR{n})}$ ($n\in\NN_0$) as\footnote{It can be also considered as a special case of Notation \ref{notn:Framework/symb/restrict}}
$$
\mathbb{1}_{\LL(\SRR{n})}\symbDistArgs{\bm{x}}{\bm{y}}=\delta^{(n)}(\bm{x}-\bm{y}),\,\bm{x},\bm{y}\in\RR{n}.
$$
Then we can introduce the folowing formal rules.
\begin{notn}\label{notn:symb-dist-1}
$$
\int \mathbb{1}_{\LL(\SRR{n})}\symbDistArgs{\bm{x}}{\bm{y}} F\symbDistArgs{\bm{z}}{\bm{x},\bm{u}} d^n\bm{x}= d^n\bm{x}=F\symbDistArgs{\bm{z}}{\bm{y},\bm{u}},
$$
and
$$
\int  F'\symbDistArgs{\bm{y},\bm{z}}{\bm{u}} \mathbb{1}_{\LL(\SRR{n})}\symbDistArgs{\bm{x}}{\bm{y}} d^n\bm{x}= d^n\bm{x}=F\symbDistArgs{\bm{z},\bm{x}}{\bm{u}},
$$
for appropriately chosen $F$ and $F'$. We assume that such an integral can be always separated from the long symbolic expression by ``Fubini theorem" and that the factors with independent variables commute as in Remark \ref{rmk:Framework/symb/stdop-props}.
\end{notn}
\begin{exmp}By Remark \ref{rmk:Framework/Augment-as-tensor},for $m,m',n',n,k\in\NN_0$, $F\in\LL(\SRR{n+k},\SRR{m})$, and $G\in\LL\left(\SRR{n'},\SRR{m'+k}\right)$ we have
$$\extD{{m+1,\ldots,m+m'}}{n+1,\ldots,n+m'}F\circ   \extD{1,\ldots,n}{1,\ldots,n}G=(F\otimes \mathbb{1}_{\LL(\SRR{m'})})\circ(\mathbb{1}_{\LL(\SRR{n})}\otimes G),$$
so
\begin{equation}\label{eq:Framework/Symbolic/ExtProd}
    (\extD{{m+1,\ldots,m+m'}}{n+1,\ldots,n+m'}F\circ   \extD{1,\ldots,n}{1,\ldots,n}G)\symbDistArgs{\bm{x},\bm{x}'}{\bm{y},\bm{y}'}=
   \end{equation}
    $$
    \int F\symbDistArgs{\bm{x}'',\bm{z}}{\bm{y}}\mathbb{1}_{\LL(\SRR{m'})}\symbDistArgs{\bm{y}''}{\bm{y}'} G\symbDistArgs{\bm{x}'}{\bm{y}'',\bm{z}}\mathbb{1}_{\LL(\SRR{n})}\symbDistArgs{\bm{x}}{\bm{x}''}d^{k}\bm{z}d^{m'}\bm{y}''d^{n}\bm{x}'=
    $$
    $$
    \int F\symbDistArgs{\bm{x}',\bm{z}}{\bm{y}} G\symbDistArgs{\bm{x}'}{\bm{y}',\bm{z}}d^{k}\bm{z},
    $$
$$\bm{x}\in\RR{n},\bm{x}'\in\RR{n'},\bm{y}\in\RR{m},\bm{y}'\in\RR{m'}.
$$
 \end{exmp}
 Finally, to restore all the power of the standard symbolic notation for distributions, we have to allow formal transformations under the integral sign.
\begin{notn}\label{notn-symb-dist}
Let $n,m,r,k,\arrs{m}{k},\arrs{m'}{k}\in\NN_0$, 
$$F_j\in\LL\left(\SRR{m_j},\SRR{m'_j}\right),\,
x_j\in\LL(\RR{n+m+r},\RR{m_j})\, y_j\in\LL(\RR{n+m+r},\RR{m'_j})
$$
and $F\in\LL(\SRR{n},\SRR{m})$. We write
$$
F\symbDistArgs{\bm{x}}{\bm{y}}=\int{\prod_{j=1}^k F_j\symbDistArgs{x_j(\bm{x},\bm{y},\bm{z})}{y_j(\bm{x},\bm{y},\bm{z})}d^k\bm{z}}
$$
if for general $f\in\SRR{n}$ and $\bm{y}\in\RR{m}$ the expression
$$
\int{\prod_{j=1}^k F_j\symbDistArgs{x_j(\bm{x},\bm{y},\bm{z})}{y_j(\bm{x},\bm{y},\bm{z})}f(\bm{x})d^k\bm{z} d^n\bm{x}}
$$
can be transformed to the one, translating to $F[f](\bm{y})$ by Notations \ref{notn:symb-dist-compose}, \ref{notn:Framework/symb/stdop}, and \ref{notn:symb-dist-1} using 
\begin{itemize}
    \item Formal linear invertible transformation of integration variables,
    \item Formal integration by parts.
\end{itemize}
\end{notn}
It is easy to see that the necessary change of variables is essentially unique (up to irrelevant renaming), so the construction above is well-defined. The ``Fubini theorem" and properties stated in Remark \ref{rmk:Framework/symb/stdop-props} are still valid and can be extended to new operations, linear transformation of the distributional arguments and partial derivatives with respect to them.
\begin{rmk}\label{rmk:symb-prod-cond}
Note that the order of formal kernels in the product is still not arbitrary: after all changes of variables it should be presented in the form (\ref{eq:notn:symb-dist-compose}), where, in particular, each integration variable appears first (if we are counting from left to right) as a distributional argument, then as a parameter. Notations \ref{notn:Framework/symb/stdop} and  \ref{notn:symb-dist-1} do not affect this rule.  One may show that if the order can be changed so that the conditions of Notation \ref{notn-symb-dist} are preserved, than the resulting operator does not change\footnote{Roughly speaking, such a reordering is possible only for tensor products, for which the commutativity was already found in Remark \ref{rmk:Framework/symb/stdop-props}}. In this light, we may postulate that the products of symbolic kernels is commutative, assuming that expression is defined if it can be reordered so that conditions of Notation \ref{notn-symb-dist} are satisfied. However, to avoid possible confusions, we will rarely use this possibility.
\end{rmk}
The restrictions can be naturally incorporated into this calculus.
\begin{notn}\label{notn:Framework/symb/restrict}
For $n,m,k$ let $F\in\LL(\SRR{n},\SRR{m+k})$ and $G=\intD{1,\ldots,k}{1,\ldots,k}F$ (or equivalently $F=\restrD{1,\ldots,k}{1,\ldots,k}G$). Then  
$$
G(\bm{x},\bm{y}|\bm{z})=F(\bm{y}|\bm{x},\bm{z}), \bm{x}\in\RR{k},\bm{y}\in\RR{n},\bm{z}\in\RR{m}.
$$
\end{notn}
 However, this identification makes the notation even more confusing, so we will use it only is selected situations.
In particular, Notation \ref{notn:Framework/symb/restrict} and (\ref{eq:Framework/exmp/Fourier})
$$
\Fourier{n}{\eta}(\bm{x},\bm{k})=e^{i\eta\bm{k}\cdot \bm{x}}.
$$
which makes the standard integral notation for the (partial) Fourier transform of distributions into a special case of (\ref{eq:Framework/extProd}).
\par 
\begin{rmk}\label{rmk:Framework/symb/disc-restriction}
In some cases it is convenient to go further and identify the formal symbol with a possibly discontinuous function making (\ref{eq:Framework/Symb/F[f]}) precise. For example, we write
$$
J_{+}\symbDistArgs{t'}{t}=\theta(t-t'), \, J_{-}\symbDistArgs{t'}{t}=\theta(t'-t),$$
$$J\symbDistArgs{t'}{t_1,t_2}=\theta(t-t_1)\theta(t_2-t)-\theta(t-t_2)\theta(t_1-t),
$$
and treat the Heaviside function just as $\theta\in L^{\infty}(\RR{})$. Note that we always are interested in $L^{\infty}$ class of such discontinuous kernels only. This identification, unlike the one in Notation \ref{notn:Framework/symb/restrict}, is ill-behaved if the derivatives are involved. We also point out that in this form it is absolutely not clear that $J_{\pm}$ and $J$ are smooth with respect to its parameters, although from direct definitions (\ref{eq:Framework/exmp/J-}-\ref{eq:Framework/exmp/J}) we know that they are.
\end{rmk}
Finally, we also use another intuitive notation for the integral operators of Subsection \ref{Framework/exmp}, writing, for example,
$
\int_{-\infty}^t (\ldots)dt'
$ instead of $\int J_+\symbDistArgs{t'}{t}(\ldots)dt'$ and similarly for $J$, $J_-$, $\intSimplex{n}{}$ and $\intUSimplex{n}{\pm}$ ($n\in\NN_0$).
\par 
To sum up, we see that all the operations on operators we have introduced have a simple intuitive presentation in the language of symbolic integrals. Most of the formal manipulations with such integrals are allowed by the facts we have established. However, overuse of this notation may hide under cover some non-trivial operations, so we often provide at least one precise translation for the formal notation used to avoid any confusion.

\section{The domain \texorpdfstring{$\DS$}{DS} and the secondary quantization}\label{DS}
\subsection{Construction of \texorpdfstring{$\DS$}{DS}}

To begin with, we define the space of $l$-particle states. 
We start from making the space $\SRR{3l}$ into a pre-Hilbert space with the scalar product
 \begin{equation}\label{eq:DS/SProd}
  (\Psi,\Psi')=\int
\overline{\Psi\left(\arrs{\vec{p}}{l}\right)} \Psi'\left(\arrs{\vec{p}}{l}\right)d^{3l}\arrs{\vec{p}}{l}, \, \forall \Psi,\Psi'\in\SRR{3l}.   
 \end{equation}
\begin{rmk}
In the Lorentz invariant theories a $\frac{1}{2\omega_0(\vec{k})}$ 
multiplier\footnote{Here we $\omega_0$ is the dispersion function defined below.} is often added to the measure of integration to make the measure Lorentz invariant (sometimes together with a  power of $2\pi$ for convenience). Here we do not assume the Lorentz invariance and chose the simplest possible measure of integration. Note, that as long as the dispersion function $\omega_0$ is as in Definition \ref{def:dispRel}, such a factor can be always absorbed into a redefinition of the wave function, without leaving the $\SRR{3l}$ class. 
\end{rmk}

\begin{rmk}\label{rmk:DS/tau-cont}
$$
|(\Psi,\Psi)|\leq C||\Psi||_{\SRR{3l}, k,0}^2
$$
for sufficiently large $C$ and $k$. Thus the tautological map from $\SRR{3l}$ (with the standard topology) to the homonymous pre-Hilbert space is continuous.
\end{rmk}
Now we take into account the bosonic statistics and define the 
$$
\DS^{l}=\symmze{3}{}\SRR{3l}
$$
and 
\begin{equation}
\label{eq:DS-def}
\DS={\bigoplus_{l=0}^{\infty}}_{\mathrm{alg}}\DS^{l}.
\end{equation}
For $\Psi\in\DS$ we denote by $\Psi_l\in\DS^{l}$ its $l$th component.
The scalar product on $\DS$ is given by
$$
(\Psi,\Psi')=\sum_{l=0}^{\infty}(\Psi_l,\Psi'_l), \, \forall \Psi,\Psi'\in\DS.
$$
The completion of $\DS$ with respect to that product is the Hilbert space of physical states $$\overline{\DS}=\mathcal{H}_{\mathrm{phys}}.$$ 
\begin{rmk}
The space $\mathcal{H}_{\mathrm{phys}}$  is important for the physical interpretation of the Hamiltonian perturbative QFT construction in Section \ref{HpQFT}. It is also relevant for the conjugated operator notion (Definition \ref{def:DS/conj}) and in particular for the unitarity condition in Subsection \ref{HpQFT/AdmTh}.  Besides these two contexts, we ignore the pre-Hilbert space structure and use instead the topology of the summands of (\ref{eq:DS-def}) inherited from $\SRR{3l}$.
\end{rmk}
It is convenient to introduce the truncated spaces
$$
\DS^{\leq l}={\bigoplus_{l'=0}^{l}}\DS^{l'}.
$$
We also define the vacuum vector as $$\Omega_0=1\in\mathbb{C}=\DS^{0}.$$
Clearly,
$$
(\Omega_0,\Omega_0)=1.
$$
\subsection{Operators on \texorpdfstring{$\DS$}{DS}}\label{DS/op}
As $\DS^{l}$, $l\in\NN$ are topological spaces, the space of operators $\LL(\DS^{l},\DS^{l'})$ can be defined as usual. The following characterization of it will be convenient for us later. 
\begin{rmk}\label{rmk:DS/op/sigma-proj}
There is a one-to-one correspondence between $\LL(\DS^{l},\DS^{l'})$ and the set
$$
\left\{A\in\LL(\SRR{3l},\SRR{3l'})| \symmze{3}{}[A]=A\right\}=\symmze{3}{}\left[\LL(\SRR{3l},\SRR{3l'})\right],
$$
where for $A\in\LL(\DS^{l},\DS^{l'})$ 
$$
\symmze{3}{}[A]=\symmze{3}{}\circ A \circ \symmze{3}{}.
$$
\end{rmk}
\begin{proof} The correspondence is almost tautological. 
If $A\in \LL(\SRR{3l},\SRR{3l'})$ satisfies $\symmze{3}{}[A]=A$, then it is valued in $\DS^{l'}$ and in particular its restriction to $\DS^{l'}$ belongs to $\LL\left(\DS^{l},\DS^{l'}\right)$. Conversely, if $A'\in \LL\left(\DS^{l},\DS^{l'}\right)$, then we set $A=A'\circ \symmze{3}{}\in\LL\left(\SRR{3l},\SRR{3l'}\right)$ and get $\symmze{3}{}[A]=A$. It is trivial to see that these two maps are inverse to each other.
\end{proof}

\begin{Def}\label{def:DS/LDS}
The algebra $\LDS$ consists of all operators $A: \DS\rightarrow \DS$ such that for any $l\in\NN_0$ there is $l'\in\NN_0$,
$$
A_{\DS^{l}}\in\LL\left(\DS^{l},\DS^{\leq l'}\right).
$$
Here $A_{\DS^{l}}$ is a restriction of $A$ to the space $\DS^{l}$.
\end{Def}
\begin{rmk}
One may show that the notation $\LDS$ becomes precise if we endow $\DS$ with the final topology induced by the inclusions $\DS^{\leq l}\rightarrow \DS$. Still, the direct characterization above is more convenient. 
\end{rmk}
\begin{rmk}\label{rmk:DS/LDS=unbounded}
All elements of $\LDS$ are at the same time densely-defined (unbounded) operators on $\mathcal{H}_{\mathrm{phys}}$. Their domain of definition is $\DS$, and they leave it invariant. 
\end{rmk}
\begin{Def}\label{def:DS/LDS1}
The space $\LLfirst(\DS)
$
is defined by 
$$
\LLfirst(\DS)={\bigoplus_{l,l'=0}^{\infty}}_{\mathrm{alg}}\LL(\DS^{l},\DS^{l'})
$$
\end{Def}
\begin{rmk}
There is a natural way to embed $\LLfirst(\DS)$ into $\LDS$ as an ideal, but we do not need that. It serves as a space of ``unquantized" operators on which the second quantization operation (see Section \ref{DS/2Q}) is defined.
\end{rmk}
In the context of Remark \ref{rmk:DS/LDS=unbounded}, it seems natural to look for the adjoint operators $\LDS^*$. For obvious reasons, $\LDS^*$ is not a subspace of $\LDS$, and furthermore, elements of $\LDS^*$ could be not defined on $\DS$ at all.
So, we introduce an alternative operation on $\LDS$.
\begin{Def}\label{def:DS/conj}
Fix an operator $A\in\LDS$. We say that it \emph{has a conjugate operator} if there is  $A^{\dag} \in\LDS$  such that
$$
(\Psi',A \Psi)=(A^{\dag}\Psi', \Psi), \quad \forall \Psi ,\Psi'\in\DS.
$$
\end{Def}
\begin{rmk}
If $A$ has a conjugate, then by definition $A^{\dag}\subset A^{*}$ is the restriction of the adjoint operator to the space $\DS$. In general.
The operators satisfying $A^{\dag}=A$ are exactly the symmetric operators defined on $\DS$.
\end{rmk}
\begin{rmk}\label{rmk:DS/LDS/conj-product}
If $A,B\in\LL(\DS)$ have conjugates, then so does $A\circ B$, and, moreover, $(A\otimes B)^{\dag}=B^{\dag}\otimes A^{\dag}$
\end{rmk}
\begin{rmk}
With almost no changes Definition can be applied for $A\in\LL(\DS^{l},\DS^{l'})$ and $A^{\dag}\in \LL(\DS^{l'},\DS^{l})$ for some $l,l'\in\mathbb{N}$ or to the case then both belong to $\LLfirst(\DS)$. We use this generalization.
\end{rmk}

\subsection{\texorpdfstring{$\DS$, $\LDS$ and $\LLfirst(\DS)$}{DS, L(DS) and Lv(DS)}-valued functions and distributions}
Lead by the ideas of Section \ref{Framework}, we introduce the following notions.
\begin{Def}
Fix $n,l,m\in\NN_0$. 
The space of $\DS^{(l)}$-valued functions of class $\SMlt{n}{m}$ is\footnote{Recall that in notation introduced in Section \ref{Framework} the operator $\extD{3l+1\ldots 3l+m+n}{3l+1\ldots 3l+m+n}\symmze{3}{}$ symmetrizes the function with respect to permutations of the first $3l$ arguments grouped by three and ignores the rest} 
$$
\DSSMlt{l}{n}{m}=\left(\extD{3l+1\ldots 3l+m+n}{3l+1\ldots 3l+m+n}\symmze{3}{}\right)\SMlt{3l+n}{m}.
$$
%We also set 
%$$
%\DSSMlt{\leq l}{n}{m}=\bigoplus_{l=0}^{l'}  \DSSMlt{l'}{n}{m}
%$$
%$$
%\DSSMlt{\leq l}{n}{m}={\bigoplus_{l=0}^{\infty}}_{\mathrm{alg}}  \DSSMlt{l}{n}{m}.
%$$
We use abbreviate notation 
$$
\DSMlt{l}{m}=\DSSMlt{l}{0}{m}, \, \DSS{l}{n}=\DSSMlt{l}{n}{0}.
$$
%and likewise for $\DSSMlt{\leq l}{n}{m}$, $\DSSMlt{}{n}{m}$.
\end{Def}
\begin{rmk}
The spaces $\DSSMlt{l}{n}{m}$ inherit the topology of $\SMlt{3l+n}{m}$, so $\LL\left(\DSS{l}{n},\DSSMlt{l'}{n'}{m}\right)$ and $\LL\left(\SRR{n},\DSSMlt{l'}{n'}{m}\right)$ are well-defined for any $l,l',m,n,n'\in\NN_{0}$. In the sense explained in Section \ref{Framework} they play role of the parameter-dependent distributions on $\RR{n}$ of class $\SMlt{n'}{m'}$ valued in $\LL(\DS^{l},\DS^{l'})$ and $\DS^{l'}$ respectively.  
\end{rmk}
\begin{rmk}
Similarly to Remark \ref{rmk:DS/op/sigma-proj} for each $\textbf{}$ we identify
$$
\LL\left(\DSS{l}{n},\DSSMlt{l'}{n'}{m'}\right)=$$
$$\symmze{3}{}\left[\LL(\SRR{3l+n},\SMlt{3l'+n'}{n'}{m'})\right],
$$
and
$$
\LL\left(\SRR{n},\DSSMlt{l'}{n'}{m'}\right)=$$
$$\left(\extD{3l'+1\ldots 3l'+m+n}{3l'+1\ldots 3l'+m+n}\symmze{3}{}\right)\circ \LL(\SMlt{3l+n}{m},\SMlt{3l'+n'}{m'}),
$$
where for $A\in \LL(\SMlt{3l+n}{m},\SMlt{3l'+n'}{n'}{m'})$ we set
$$
\symmze{3}{}[A]=
\left(\extD{3l'+1\ldots 3l'+m'+n'}{3l'+1\ldots 3l'+m'+n'}\symmze{3}{}\right)\circ A\circ \left(\extD{3l+1\ldots 3l+m+n}{3l+1\ldots 3l+m+n}\symmze{3}{}\right).
$$
\end{rmk}
Generalization of Definitions \ref{def:DS/LDS} and \ref{def:DS/LDS1} is straightforward.
\begin{Def} The space $\LL\left(\DSS{}{n},\DSSMlt{}{n'}{m'}\right),$ consists of all operators $A: \DSS{}{n}\rightarrow \DSSMlt{}{n'}{m'}$ such that for any $l\in\NN_0$ there is $l'\in\NN_0$, for which
$$
A_{\DSS{l}{n}}\in\LL\left(\DSS{ l}{n},\DSSMlt{\leq l'}{n'}{m'}\right),
$$
where $A_{\DSS{l}{n}}$ is the restriction of $A$ on the space $\DSS{l}{n}$.
\end{Def}
\begin{Def}
For $n,n',m'\in\NN$ we define
$$
\LLfirst\left(\DSS{}{n},\DSSMlt{}{n'}{m'}\right)=$$
$${\bigoplus_{l,l'=0}^{\infty}}_{\mathrm{alg}}\LL\left(\DSS{l}{n},\DSSMlt{l'}{n'}{m'}\right)
$$
and 
$$
\LL\left(\SRR{n},\DSSMlt{}{n'}{m'}\right)={\bigoplus_{l'=0}^{\infty}}_{\mathrm{alg}}\LL\left(\SRR{n},\DSSMlt{l'}{n'}{m'}\right).
$$
\end{Def}

\begin{notn}\label{notn:DS/opval-symb}
For the sake of readability, we eventually use the symbolic notation generalizing the on of Subsection \ref{Framework/Symbolic}. 
\begin{itemize}
    \item If $\Psi\in\DSSMlt{l}{n}{m}$, ($l,n,m\in\mathbb{N}_0$) we set 
    $$\underline{\Psi}=\evaluateP{3l+1\ldots 3l+n+m}{3l+n+m}\Psi;$$
    \item If $A\in\LL\left(\SRR{k},\DSSMlt{l}{n}{m}\right)$ ($l,n,m,k\in\NN_0$), then
    $$
    \underline{A}=\evaluateP{3l+1\ldots 3l+n+m}{3l+n+m}\circ A;
    $$
    \item If $A\in\LL(\DSS{l}{k},\DSSMlt{l'}{n}{m})$, ($l,l',,n,m,k\in\NN_0$), then $\underline{A}$ is the operator-valued parameter-dependent distributions defined by 
    $$
    \underline{A}=    \evaluateP{3l'+1,\ldots,3l'+n+m}{3l'+n+m}\circ (\evaluatePD{1,\ldots,k}{A})$$
\end{itemize}
We assume that these maps are defined for, respectively, $\Psi\in\DSSMlt{}{n}{m}$, $LL\left(\SRR{k},\DSSMlt{}{n}{m}\right)$ and $\LL(\DSS{}{k},\DSSMlt{}{n}{m})$ by linearity.
\end{notn}
\begin{rmk}\label{rmk:DS/opval-symb}
\begin{itemize}
    \item For $\Psi\in\DSSMlt{}{n}{m}$, ($n,m\in\mathbb{N}_0$) we have $\underline{\Psi}\in\mathcal{SO_M}(\RR{n},\RR{m};\mathcal{H}_{\rm{phys}})$\footnote{$\mathcal{SO_M}(\RR{n},\RR{m};\mathcal{H}_{\rm{phys}})$ can be defined similar to Remark \ref{rmk:Framework/VectorTest};.} ;
    \item For $A\in\LL\left(\SRR{k},\DS\right)$ ($k\in\NN_0$), is a vector-valued distribution in the sense of \cite{NN};
    \item For $A\in\LL(\DSS{}{k},\DS{})$, ($k\in\NN_0$),  $\underline{A}$ is an operator-valued distribution in the sense of \cite{NN}.
\end{itemize}
More general versions define parameter-dependent vector-valued distributions.
\end{rmk}
\begin{proof}
The key ingredient is Remark \ref{rmk:DS/tau-cont} combined with the results of Section \ref{Framework}. 
\end{proof}

Now we want to generalize the operations of Section \ref{Framework} to the operator-valued distributions. In the light of Notation \ref{notn:DS/opval-symb} we just have to ignore the first $3l$ arguments.
\begin{notn}\label{notn:DS/hat-op}
\begin{itemize}
    \item For $n,m\in\NN_{0}$, two tuples of pairwise distinct numbers  $\arrs{i}{k}\in\{1,\ldots,n\}$, $\arrs{i'}{k}\in\{1,\ldots,m\}$,  and  $A\in\LL(\DSS{}{n},\DSS{}{m})$ we define 
    $\extDH{\arrs{i}{k}}{\arrs{i'}{k}}A \in\LL(\DSS{}{n+k},\DSS{}{m+k})$ by setting for each $l,l'\in\NN_{0}$
    $$
    (\extDH{\arrs{i}{k}}{\arrs{i'}{k}}A\Psi)_{l'}=
    (\extD{(3l+i_j)_{j=1,\ldots,k}}{(3l'+i_j)_{j=1,\ldots,k}}A\Psi)_{l'}, \,\forall \Psi \in\DSS{l}{n+k} 
    $$
    \item  For $n,m\in\NN_{0}$, two pairwise distinct tuples $\arrs{i}{k}\in\{1,\ldots,n\}$, $\arrs{i'}{k}\in\{1,\ldots,m\}$, and  $A\in\LL(\DSS{}{n},\DSS{}{m+k})$ we define \\
    $\intDH{\arrs{i}{k}}{\arrs{i'}{k}}A \in\LL(\DSS{}{n+k},\DSS{}{m})$ by setting for each $l,l'\in\NN_{0}$
    $$
    (\intDH{\arrs{i}{k}}{\arrs{i'}{k}}A\Psi)_{l'}=
    (\intD{(3l+i_j)_{j=1,\ldots,k}}{(3l'+i_j)_{j=1,\ldots,k}}A\Psi)_{l'}, \,\forall \Psi \in\DSS{l}{n+k}; 
    $$
    We set $\restrDH{\arrs{i}{k}}{\arrs{i'}{k}}=\left(\intDH{\arrs{i'}{k}}{\arrs{i}{k}}\right)^{-1}.$
    \item For $n,m,n',m,m'\in\NN_0$,  $F\in\LLfirst(\DSS{}{n},\DSS{}{m})$ and\\
    $F'\in\LLfirst\left(\DSS{}{n'},\DSS{}{m'}\right)$ we set 
    $$
    F\hotimes F'=\left(\extDH{n+1,\ldots,n+m'}{n+1,\ldots,n+m'}F\right)\circ \left(\extDH{1,\ldots,n}{1,\ldots,n}F'\right)=$$
    $$\left(\extDH{1,\ldots,m}{1,\ldots,m}F'\right)\circ\left(\extDH{n+1,\ldots,n+n'}{n+1,\ldots,n+n'}F\right)
\in \LLfirst\left(\DSS{}{n+n'},\DSS{}{m+m'}\right).
    $$
    \item If $A\in\LL(\SRR{n},\SRR{m})$, then we define\footnote{As we will see in Example \ref{exmp:DS/2Q/numdist} this notation is consistent with the secondary quantization formalism of the next Subsection.} $\widehat{A}\in\LL(\SRR{n},\SRR{m})$ by setting for each $l\in\NN$
    $$
    \widehat{A}[\Psi]=\extD{1,\ldots,3l}{1,\ldots,3l}A[\Psi]\, \forall \Psi\in\DS^{l}.
    $$
\end{itemize}
\end{notn}
\begin{rmk}
For $n,n'\in\NN_0$,  $F\in\LLfirst(\DSS{}{n},\DS{})$ and $F'\in\LLfirst\left(\DSS{}{n'},\DS{}\right)$ we have
$$
\underline{F\hotimes F'}=\underline{F}\otimes \underline{F'},
$$
where the tensor product of the operator-valued distributions in the right-hand sides is understood as in \cite{NN}.
\end{rmk}
\begin{proof}
Recall that in \cite{NN} the tensor product of distributions is characterised as a uninque distribution such that
$$(\underline{F}\underline{F'})[f\otimes f']=
\underline{F}[f] \circ \underline{F'}[f'], \forall f\in\SRR{n},\forall f'\in\SRR{n'}.
$$
The rest follows by Remarks \ref{rmk:Framewor/ExtUniversal} and \ref{rmk:Framework/Augment/algebra}.
\end{proof}
\begin{notn}\label{notn:DS/symb-opval}
We define the symbolic notation of Subsection \ref{Framework/Symbolic} for the operator-valued distributions defined in Notation \ref{notn:DS/opval-symb} by replacing all the operations ($\extD{}{}$, $\intD{}{}$, $\restrD{}{}$, $\otimes$, pre- and post-compositions with auxiliary operators) by their hatted versions of Notation \ref{notn:DS/hat-op}. 
\end{notn}
\begin{rmk}
As was established in Section \ref{Framework} the correspondences behind Notation \ref{notn:DS/opval-symb} and Remark \ref{rmk:DS/opval-symb} are one-to-one. In particular, it means that the ``underlining" is an invertible operation. 
This allows to define the operators via symbolic expressions with their underlined counterparts. For example, for $F\in\LL(\DSS{}{n},\DSS{}{m})$ and $\Psi\in\DSS{}{n}$ we may write
$$
\underline{\Phi}(\bm{y})=\int \underline{F}\symbDistArgs{\bm{x}}{\bm{y}} \underline{\Psi}(\bm{x})d^n \bm{x}, \, \bm{y}\in\RR{m}
$$
instead of $\Phi=F\circ \Psi$. We will see more complicated examples in which, unlike the trivial one above, the symbolic notation is more insightful.
\end{rmk}
In the light of the last remark, the following definition is well formulated.
\begin{Def}
For $A\in \LL(\DSS{}{n},\DSSMlt{}{m}{k})$ we set $A^{\dag}\in \LL(\DSS{}{n},\DSSMlt{}{m}{k})$ to be a unique operator, characterized by
$$
\underline{A^{\dag}}[f](\bm{x})=\underline{A}[\overline{f}](\bm{x})^{\dag},\,\forall f\in\RR{n}, \forall \bm{x}\in\RR{m+k}.
$$
 (if it exists).
\end{Def}
\begin{rmk}\label{rmk:DS/LDS/conj-product-gen}
Remark \ref{rmk:DS/LDS/conj-product} does not generalize to this setting directly, because the operators, in general, appear in incompatible order. However, the following is true:
\begin{itemize}
    \item If $A\in\LL(\DSS{}{n},\DSS{}{m})$, $B\in\LL(\DSS{}{k},\DSS{}{l})$ have conjugated operators ($n,m,k,l\in\NN_0$) and
    $$
    C=A\hotimes B,
    $$
    then $C$ has a conjugated operator, defined by
    $$
    \underline{C}\symbDistArgs{\bm{x},\bm{x}'}{\bm{y},\bm{y'}}=\underline{B}^{\dag}\symbDistArgs{\bm{x}'}{\bm{y'}}\underline{A}^{\dag}\symbDistArgs{\bm{x}}{\bm{y}},\,\bm{x}\in\RR{n},\bm{y}\in\RR{m},\bm{x}'\in\RR{k},\bm{y'}\in\RR{l}.
    $$
    \item If $F\in\LL(\SRR{n},\SRR{m})$ and $A\in\LL(\DSS{}{m},\DSS{}{k})$ $A$ has a conjugated operator, that so does $F\circ A$. Moreover, $(F\circ A)^{\dag}=\overline{F}\circ A^{\dag}$. Here $\overline{F}$ is a unique operator in $\LL(\SRR{n},\SRR{m})$ such that
    $$
    \overline{F}[f]=\overline{F[\overline{f}]}, \,\forall f\in\SRR{n}.
    $$
\end{itemize}
\end{rmk}
The proof is straightforward. It is convenient to use Remark \ref{rmk:Framewor/ExtUniversal}. 
\subsection{Secondary quantization}\label{DS/2Q}

In perturbative quantum field theory we deal with a very special kind of operators in $\LDS$ (and thus very special kinds of operator-valued functions and distributions). The relevant operators are constructed by smearing tensor (or Wick) products of the creation and annihilation operator-valued distributions with some test functions and, more generally, with distributions with particular types of singularities (see e.g. \cite{pAQFT}).  We instead first characterize the explicitly by their action on $\DS$ and then show correspondence with more traditional approaches. 
\par
We use the name second quantization for this procedure resembles the way how the many-particle operators (say the energy and the momentum) are constructed from the single-particle ones in the second quantization in physical literature. We first present the construction in details for the operators and then explain how it generalizes to operator-valued functions and distributions (including parameter-dependent ones).
\subsubsection{Basic construction}

\begin{Def}\label{def:DS/2Q/AhatDS}Let $l,l'\in\mathbb{N}_0$ and $ A\in \LL(\DS^{l},\DS^{l'})$. The \emph{second quantization} of $A$ is $\widehat{A}\in\LDS$ defined by
\begin{equation}\label{eq:def:DS/2Q/LDS}
\widehat{A}\Psi=\frac{\sqrt{n!(n-l+l')!}}{(n-l)!}\left(\symmze{3}{}\circ \left(\extD{3l+1,\ldots ,3n}{3l'+1,\ldots, 3n-3l+3l'}A\right)\right)\Psi, \forall\Psi\in \DS^{n}, \qquad n\geq l.
\end{equation}
$$
\widehat{A}\Psi=0,\qquad \forall\Psi\in \DS^{n}, \qquad n<l.
$$
\end{Def}
\begin{rmk}
Recall that
$$
\extD{3l+1,\ldots ,3n}{3l'+1,\ldots, 3n-3l+3l'}A=A\otimes \mathbb{1}_{\LL(\SRR{3(n-l)})}.
$$
This form is less direct but more intuitive.
\end{rmk}
\begin{rmk}\label{rmq:DS/2Q/SmoothKernel}

To understand the meaning of the rather technical Definition let us introduce, a bit in advance, the standard creation and annihilation operator-valued distributions which we denote with $\underline{\widehat{a}_+}$ and $\underline{\widehat{a}_-^{R}}$ as they will appear in Examples \ref{exmp:DS/2Q/apm} and \ref{exmp:DS/2Q/amR}   \footnote{We use a non-Lorentz-invariant normalization leading to the commutation relation (\ref{eq:DS/CCR}).}. Then Definition \ref{def:DS/2Q/AhatDS} is a formalization of\footnote{See Example \ref{rmk:DS/2Q/symbolic} for the rigorous meaning of (\ref{eq:DS/2Q-symbolic})}
\begin{equation}
\widehat{A}=\int\left(\prod_{j=1}^{l'}\underline{\widehat{a}_+}(\vec{p}_j')\right)
A\symbDistArgs{\arrs{\vec{p}}{l}}{\arrs{\vec{p}'}{l'}}
\left(\prod_{j=1}^{l}\underline{\widehat{a}_-^{R}}(\vec{p}_j)\right)d^{3l}\arrs{\vec{p}}{l}d^{3l'}\arrs{\vec{p'}}{l'}
\label{eq:DS/2Q-symbolic}
\end{equation}

Here we took into account that the annihilation operator is in fact a smooth function of its parameter (in our terminology it means that it has a restriction which we denoted with $a_{-}^R$). The correspondence becomes straightforward in the case then $A$ has a smooth kernel (in our terminology, A has an  $\SRR{3(l+l')}$-class restriction with respect to all its distributional arguments). Then coincidence of (\ref{eq:DS/2Q-symbolic}) and (\ref{eq:def:DS/2Q/LDS}) follows by comparison of their value on generic $\Psi\in\DS$. The operators we consider in this paper are of this type only, but for the operator-valued distributions defined in the next subsection the general construction of Definition (\ref{eq:def:DS/2Q/LDS}) is more convenient. It is worth noting that operators with a smooth kernel could never realize momentum or energy conservation law, so in the context of the strong adiabatic limit the full construction will be necessary.
\end{rmk}
\begin{rmk}
The map $A\mapsto\widehat{A}$ is obviously linear. It has a straightforward linear extension to $\LLfirst(\DS)$ which we assume to be defined from now on.%%%TBD not augmentation
\end{rmk}
\begin{rmk}\label{rmk:DS/2Q/injective-sym}
We can extend the secondary quantization from the space $\LL(\DS^{l},\DS^{l'})$ to $\LL\left(\SRR{3l},\SRR{3l'}\right)$ keeping (\ref{eq:def:DS/2Q/LDS}) untouched. It is easy that $\widehat{A}=\widehat{\symmze{3}{}[A]}$. Our choice is dictated by injectivity of the second quantization acting defined on $\LLfirst(\DS)$\footnote{To see injectivity of second quantization on $\LL\left(\DS^{l},\DS^{l'}\right)$ consider its action on $\DS^{l}$. For the whole space $\LLfirst(\DS)$, proceed by induction in $l$.}. 
\end{rmk}
The main advantage of this formalism is the fact that all operations on such operators can be done on the level of $\LLfirst(\DS,\DS)$ without keeping track of the extra variables and combinatoric factors. In particular, the conjugation and the products can be computed that way.
\begin{rmk}\label{rmk:DS/2Q/2Q-conj}
For any $A\in\LLfirst(\DS,\DS)$ one has
$$
\widehat{A^{\dag}}=\widehat{A}^{\dag}.
$$
In other words, second quantization commutes with conjugation.
\end{rmk}
\begin{proof}
It is enough to consider $A\in\LL(\DS^{l},\DS^{l'})$ for all possible $l,l'\in\NN_0$ and show for any $n\in\mathbb{N}_0$, $n\geq l$, any $\Psi\in \DS^{n}$ and $\Psi'\in\DS^{n-l+l'}$ that 
$$
(\Psi',\extD{l+1,\ldots ,n}{l'+1,\ldots, n-l+l'}A\Psi)={(\extD{l'+1,\ldots, n-l+l'}{l+1,\ldots ,n}A^{\dag}\Psi,\Psi')}.
$$
Note that the combinatoric factors canceled out and we put away the symmetrization, for both $\Psi$ and $\Psi'$ being symmetric already. The last line is equivalent to
$$
\int
\left(\evaluateP{3l'+1,\ldots,3(n-l+l')}{n-l+l'}\Psi'(\arrs{\vec{p}}{n-l}),A\evaluateP{3l+1,\ldots,3n}{n}\Psi(\arrs{\vec{p}}{n-l})\right)d^{3(n-l)}\arrs{\vec{p}}{n-l}=
$$
$$\int \left(A^{\dag}\evaluateP{3l'+1,\ldots,3(n-l+l')}{n-l+l'}\Psi'(\arrs{\vec{p}}{n-l}),\evaluateP{3l+1,\ldots,3n}{n}\Psi(\arrs{\vec{p}}{n-l})
\right)d^{3(n-l)}\arrs{\vec{p}}{n-l},
$$
which holds by Definition \ref{def:DS/conj}.
\end{proof}
Let us now compute the product $\widehat{A}\widehat{B}$ for some $A,B\in\LLfirst(\DS)$. First of all, note that the creation and operators in (\ref{eq:DS/2Q-symbolic}) are normally ordered. Then to present $\widehat{A}\widehat{B}$ again in the form (\ref{eq:DS/2Q-symbolic}) we need some kind of the Wick theorem. So, we need a formalization of the Wick product with contractions.  Looking once again on  (\ref{eq:DS/2Q-symbolic}) we identify the tensor product in $\LLfirst(\DS)$ with the Wick product:
$$
:\widehat{A}\widehat{B}:=
\widehat{A\otimes B}.
$$
Inspired by that we introduce the following
\begin{Def}\label{def:DS/2Q/contractions}
Let
$A\in\LL(\DS^{l_A},DS^{l'_A})$, $B\in\LL(\DS^{l_B},DS^{l'_B})$ and let $r\leq \min(l_A,l'_B)$. Then the \emph{tensor product with $r$ contractions} 
$$A\otimes_r B\in \LL(\DS^{l_A+l_B-r},DS^{l'_A+l'_B-r})$$
is
$$
A\otimes_r B=\symmze{3}{}\left[(\extD{3(l_A-r)+1,\ldots,3(l_A'+l'_B-r)}{3(l'_A-r)+1,\ldots,3(l'_A-l_B+r)}A)\circ{}(\extD{1,\ldots,3(l_A-r)}{1,\ldots,3(l_A-r)}B)\right].
$$
\end{Def}
\begin{rmk}\label{rmk:DS/2Q/contractions}
For clarity, we provide more readable versions of Definition \ref{def:DS/2Q/contractions}. First of all, reading the proof of Proposition \ref{prop:DS/2Q/Wick} it may be convenient to keep in mind that  
$$
A\otimes_{r}B=
\symmze{3}{}\left[\left(A\otimes \mathbb{1}_{\LL\left(\SRR{3(l'_B-r)}\right)}\right)\left( 1_{\LL\left(\SRR{3(l_A-r)}\right)}\otimes B\right)\right].
$$
To see that this is indeed the Wick product with $r$ contractions, the symbolic form is more suitable. For the sake of clarity, we omit the symmetrization operator which is not relevant by Remark \ref{rmk:DS/2Q/injective-sym}.
$$
A\otimes_{r}B=\symmze{3}{}\left[A\otimes'_{r}B\right],
$$
$$
\left(A\otimes'_{r}B\right)\symbDistArgs{\arrs{\vec{p}}{l_A+l_B-r}}{\arrs{\vec{p}'}{l'_A+l'_B-r}}=
\int A\symbDistArgs{\arrs{\vec{p}}{l_A-r},\arrs{\vec{k}}{r}}{\arrs{\vec{p}'}{l'_A}}\times
$$
$$
B\symbDistArgs{\arrsM{\vec{p}}{l_A-r+1}{l_A+l_B-r}}{\arrs{\vec{k}}{r},\arrsM{\vec{p}'}{l'_A+1}{l'_A+l'_B-r}}d^{3r}\arrs{\vec{k}}{r}.
$$
Comparing it with (\ref{eq:DS/2Q-symbolic}) we see that our terminology is natural. Finally, such a product may be presented graphically as in Fig. \ref{fig:DS/2Q/contractions}. Here each dot denotes an operator and each line incident to it is an argument of the corresponding operator, distributional (for lines coming from the right), or parametric (for the lines going to the left). The arguments, connecting two dots are subject to integration. Note that the graph is totally ordered  \footnote{see Subsection \ref{Intro/prelim} for terminology related to the Feynman graphs}. Clearly, consequent products may be given such representation in terms of totally ordered graphs.
\end{rmk}
\begin{figure}
    \centering
    \begin{tikzpicture}
    \filldraw[black] (0,0) circle (2pt) node[anchor=south]{A};
        \filldraw[black] (2,-2) circle (2pt) node[anchor=north]{B};
        \draw (-1,1) -- (0,0);
        \draw (-1,0.5) -- (0,0);
        \node[] at (-0.8,0) {$\vdots$};
        \draw (-1,-0.8) -- (0,0);
        \draw [decorate,
    decoration = {brace, amplitude=5pt}] (-1.1,-0.8) --  (-1.1,1);
        \node[] at (-1.4,0) {$l'_A$};
        \draw (3,1) -- (0,0);
        \draw (3,0.5) -- (0,0);
         \node[] at (2.8,0) {$\vdots$};
         \draw (3,-0.8) -- (0,0);
         \draw [decorate,
    decoration = {brace, amplitude=5pt}] (3.1,1) --  (3.1,-0.8);
        \node[anchor=west] at (3.3,0) {$l_A-r$};
        \draw (0,0) to[bend left=30] (2,-2);
        \draw (0,0) to[bend right=45] (2,-2);
        \draw (0,0) to[bend left=15] (2,-2);
        \node[] at (1,-1) {$\iddots$};
        \draw [decorate,
    decoration = {brace, amplitude=5pt}] (1.6,-0.8) --  (0.6,-1.7);
        \node[] at (1.4,-1.4) {$r$};
        \draw (-1,-2) -- (2,-2);
        \draw (-1,-2.2) -- (2,-2);
         \node[] at (-0.8,-2.4) {$\vdots$};
         \draw (-1,-3) -- (2,-2);
         \draw [decorate,
    decoration = {brace, amplitude=5pt}] (-1.1,-3) --  (-1.1,-2);
        \node[anchor=east] at (-1.2,-2.5) {$l'_B-r$};
        \draw (3,-1.5) -- (2,-2);
        \draw (3,-1.7) -- (2,-2);
         \node[] at (2.8,-2) {$\vdots$};
         \draw (3,-2.7) -- (2,-2);
        \draw [decorate,
    decoration = {brace, amplitude=5pt}] (3.1,-1.5) --  (3.1,-2.7);
        \node[anchor=west] at (3.2,-2.2) {$l_B$};
        
    \end{tikzpicture}
    \caption{Graphical presentation of the contracted product $A\otimes_{r}B\in \LLfirst(\DS^{l_A+l_B-r},\DS^{l'_A+l'_B-r})$, $A\in\LLfirst(\DS^{l_A},\DS^{l'_A})$, $B\in\LLfirst(\DS^{l_B},\DS^{l'_B})$. }
    \label{fig:DS/2Q/contractions}
\end{figure}

\begin{prop}\label{prop:DS/2Q/Wick}
Let
$A\in\LL(\DS^{l_A},\DS^{l'_A})$, $B\in\LL(\DS^{l_B},\DS^{l'_B})$. Then
$\widehat{A}\widehat{B}=\widehat{A\otimes_{\bullet}B}$ with
$$
A\otimes_{\bullet}B=\sum_{r=0}^{\min(l_A,l'_B)}r! \binom{l_A}{r}\binom{l'_B}{r} A\otimes_r B\in\LLfirst(\DS).
$$
\end{prop}
\begin{rmk}
The binomial coefficients are nothing but the number of selecting $r$ pairs from the creation operators of $B$ and annihilation operators of $A$, so the above is indeed a form of the Wick theorem.
\end{rmk}
The proof of Proposition \ref{prop:DS/2Q/Wick} is straightforward, but technical. It is convenient to separate the following combinatoric facts.
\begin{lem}\label{lem:prop:DS/2Q/Wick/Comb/Sigmas}
Let $n,l_A,l_B,l_A',l_B'\in\NN_0$ and $n\geq \min(l_A,l_A+l_B-l_B')$. Introduce
$R: \symmgr{n-l_B+l'_B}\longrightarrow \NN_0$ as
$$
 R(\sigma)=\left|\{j=1,\ldots,l_B'|\sigma(j)\leq l_A\}\right|.
$$
For each $r\in\{\max(0,l_A+l_B-n),\max(0,l_A+l_B-n)+1,\ldots,\min(l_A,l_B')\}$ define $\sigma^{(r)}$ by setting
$\sigma^{(r)}\in \symmgr{n-l_B+l_B'}$ as
$$
\sigma^{(r)}(j)=j+l_A-r, \qquad 1\leq j \leq l_B',
$$
$$
\sigma^{(r)}(j)=j-l_B', \qquad l_B'< j \leq l_A+l_B'-r,
$$
$$
\sigma^{(r)}(j)=j, \qquad l_A+l_B'-r< j \leq n.
$$
Then
\begin{enumerate}
    \item $\max(0,l_A+l_B-n)\leq R(\sigma)\leq \min(l_A,l_B')$ $\forall \sigma\in\symmgr{n-l_B+l_B'}$;
    \item $R(\sigma^{(r)})=r$, $r\in\{\max(0,l_A+l_B-n),\ldots,\min(l_A,l_B')\} $
    \item For any $\sigma\in\symmgr{n-l_B+l_B'}$ there is a (not necessary unique) decomposition
    \begin{equation}
    \sigma=(\sigma_1\times \sigma_2)\circ \sigma^{(R(\sigma))}\circ (\sigma_3\times \sigma_4)
    \label{eq:lem:prop:DS/2Q/Wick/Comb/Sigma}  
    \end{equation}
    
    with $\sigma_1\in \symmgr{l_A}$, $\sigma_2\in \symmgr{n-l_B+l_B'-l_A}$, $\sigma_3\in \symmgr{l_B'}$, $\sigma_4\in\symmgr{n-l_B}$, 
    Here we assume the natural inclusion $\symmgr{k}\times \symmgr{k'}\subset \symmgr{k+k'}$ with the first and second factors acting, respectively, on the first $k$ the rest $k'$ elements; 
    \item \begin{equation}\label{eq:lem:prop:DS/2Q/Wick/Comb/count}
        \big|\{\sigma \in \mathfrak{S}_{n+l_B'-l_B}|R(\sigma)=r\}\big|=\binom{l_B'}{r}\binom{l_A}{r}r! \frac{(n+l_B'-l_B-l_A)!(n-l_B)!}{(n-l_B-l_A+r)!}.
    \end{equation}
\end{enumerate}
\end{lem}
\begin{proof}
\begin{enumerate}
    \item The upper bound $R(\sigma)\leq \min(l_A,l_B')$ follows directly from the definition. At the same time $$
    l_B'-R(\sigma)=\left|\{j=1,\ldots,l_B'|l_A+1\leq \sigma(j) \leq n-l_B+l_B'\}\right|,
    $$
    so $l_B'-R(\sigma)\leq n-l_B+l_B'-l_A$ giving the desired lower bound.
    \item Directly by construction of $\sigma^{(r)}$ 
    $$
    \{j=1,\ldots,l_B'|\sigma(j)\leq l_A\}=\{1,\ldots,r\},
    $$
    so $R(\sigma^{(r)})=r$.
    \item[3,4] Let us classify all permutations $\sigma$ with fixed $R(\sigma)=r$. To fix such a permutation we do the following:
    \begin{itemize}
    \item We select $r$ elements among the first $l_B'$ numbers, find for them places among the first $l_A$ numbers, and chose a bijection between the former and the latter. This gives $$\binom{l_B'}{r}\binom{l_A}{r}r!$$ 
    possibilities;
    \item The rest $l_B'-r$ numbers from $1$ to $l_B'$ should be placed somewhere from $l_A+1$ to $n+l_B'-l_B$ leading to   $$\frac{(n+l_B'-l_B-l_A)!}{(n-l_B-l_A+r)!}$$
    possibilities;
    \item Finally, we should fix a bijection between the numbers from $l_B'+1$ to $n+l_B'-l_B$ and the yet free places in $$(n-l_B)!$$ 
    ways. 
    \end{itemize}
    The product of these factor give (\ref{eq:lem:prop:DS/2Q/Wick/Comb/count}), and the described procedure gives a decomposition of type (\ref{eq:lem:prop:DS/2Q/Wick/Comb/Sigma}).
\end{enumerate}
\end{proof}
\begin{proof}[Proof of Proposition \ref{prop:DS/2Q/Wick}]
Take first $n< \max(l_B,l_A+l_B-l'_B)$ and $\Psi\in \DS^{n}$. We get $\widehat{A}
\widehat{B}\Psi=0$ immediately as well as $\widehat{A\otimes_{r}B}\Psi$ for any $r\in\{1,\ldots,\min(l_A,l'_B)\}$.
\par 
Now take $n\geq  \max(l_B,l_A+l_B-l'_B)$ and again $\Psi\in\DS^{n}$. 
We have
$$
\widehat{A}
\widehat{B}\Psi=\frac{\sqrt{n!(n-l_B+l'_B)!}}{(n-l_B)!}\frac{\sqrt{(n-l_B+l'_B)!(n-l_A+l'_A-l_B+l'_B)!}}{(n-l_B+l'_B-l_A)!}\times$$
$$\symmze{3}{}\left(\extD{3l_A+1,\ldots ,3n-3l_B+3l_B'}{3l'_A+1,\ldots, 3n_A-3(l_A+l_B)+3(l'_A+l'_B)}A\right)\symmze{3}{}\left(\extD{3l_B+1,\ldots ,3n}{3l_B'+1,\ldots, 3n-3l_B+3l_B'}B\right)\Psi.
$$
Now we present the operator $\symmze{3}{}$ in the middle explicitly:
$$
\widehat{A}
\widehat{B}\Psi=c\sum_{\sigma\in\mathfrak{S}_{n-l_B+l_B'}}\symmze{3}{}\left(\extD{3l_A+1,\ldots ,3n-3l_B+3l_B'}{3l'_A+1,\ldots, 3n_A-3(l_A+l_B)+3(l'_A+l'_B)}A\right)\circ$$
$$\permNKS{n-l_B+l'_B}{3}[\sigma]
\left(\extD{3l_B+1,\ldots ,3n}{3l_B'+1,\ldots, 3n-3l_B+3l_B'}B\right)\Psi,
$$
$$
c=\frac{\sqrt{n!(n-l_A+l'_A-l_B+l'_B)!}}{(n-l_B)!(n-l_B+l'_B-l_A)!}.
$$
We apply Lemma \ref{lem:prop:DS/2Q/Wick/Comb/Sigmas} to classify the permutations. By the decomposition \ref{eq:lem:prop:DS/2Q/Wick/Comb/Sigma}
$$
\permNKS{n-l_B+l'_B}{3}[\sigma]=(\permNKS{l_A}{3}[\sigma_1]\otimes \permNKS{n-l_B+l'_B-l_A}{3}[\sigma_2])\permNKS{n-l_B+l'_B}{3}[\sigma^{(R(\sigma))}](\permNKS{l_B'}{3}[\sigma_3]\otimes \permNKS{n-l_B}{3}[\sigma_4]).
$$
Since $\symmze{3}{}[B]=B$ and $\Psi=\symmze{3}{}\Psi$ we get \footnote{We use that the symmetrization operator absorbs any permutations.}, 
$$
(\permNKS{l_B'}{3}[\sigma_3]\otimes \permNKS{n-l_B}{3}[\sigma_4])\left(\extD{3l_B+1,\ldots ,3n}{3l_B'+1,\ldots, 3n-3l_B+3l_B'}B\right)\Psi=
$$
$$
\left(\extD{3l_B+1,\ldots ,3n}{3l_B'+1,\ldots, 3n-3l_B+3l_B'}B\right)(
\extD{1,\ldots ,3l_B}{1,\ldots ,3l_B} \permNKS{n-l_B}{3}[\sigma_4])\Psi=\left(\extD{3l_B+1,\ldots ,3n}{3l_B'+1,\ldots, 3n-3l_B+3l_B'}B\right)Psi.
$$
Similarly, using $\symmze{3}{}[A]=A$ we get
$$\symmze{3}{}\left(\extD{3l_A+1,\ldots ,3n-3l_B+3l_B'}{3l'_A+1,\ldots, 3n_A-3(l_A+l_B)+3(l'_A+l'_B)}A\right)
(\permNKS{l_A}{3}[\sigma_1]\otimes \permNKS{n-l_B+l'_B-l_A}{3}[\sigma_2])=
$$
$$
\symmze{3}{}
\left(\extD{3l_A+1,\ldots ,3n-3l_B+3l_B'}{3l'_A+1,\ldots, 3n_A-3(l_A+l_B)+3(l'_A+l'_B)}A\right)$$
Finally,
$$
\left(\extD{3l_A+1,\ldots ,3n-3l_B+3l_B'}{3l'_A+1,\ldots, 3n_A-3(l_A+l_B)+3(l'_A+l'_B)}A\right)\left(\extD{3l_B+1,\ldots ,3n}{3l_B'+1,\ldots, 3n-3l_B+3l_B'}B\right)=$$
$$
\extD{3l_A+3l_B+1,\ldots,3n}{3l'_A+3l_B'+1,\ldots, 3n-3l_A-3l_B+3l'_A+3l_B'}\left(A\otimes_r B\right).
$$
Putting this all back and using $\widehat{A\otimes_{r} B}\Psi=0$ whenever $n<l_A+l_B-r$ we get
$$
\widehat{A}
\widehat{B}\Psi=\sum_{r=0}^{\min(l_A,l_B')}C_r\widehat{A\otimes_{r} B}\Psi ,
$$
where 
$$C_r=c\big|\{\sigma \in \mathfrak{S}_{n+l_B'-l_B}|R(\sigma)=r\}\big|
\left(\frac{\sqrt{n!(n-l_A+l'_A-l_B+l'_B)!}}{(n-l_A-l_B+r)!}\right)^{-1}=
$$
$$
\big|\{\sigma \in \mathfrak{S}_{n+l_B'-l_B}|R(\sigma)=r\}\big| \frac{(n-l_A-l_B+r)!}{(n-l_B)!(n-l_B+l'_B-l_A)!}.
$$

Thus by the last assertion of Lemma \ref{lem:prop:DS/2Q/Wick/Comb/Sigmas}
$C_r=\binom{l_B'}{r}\binom{l_A}{r}r!$ as in the statement.
\end{proof}
\begin{rmk}\label{rmk:DS/2Q/formal-alg}
We can consider $\LLfirst(\DS)$ as a formal involutive algebra with the product $\otimes_{\bullet}$. Then by Proposition \ref{prop:DS/2Q/Wick} and Remark \ref{rmk:DS/2Q/2Q-conj} the second quantization is a homomorphism of such algebras. On the other hand, it means that the whole analysis (see also Remark \ref{rmk:DS/2Q/formal-alg-gen}) of this paper can be done completely within the formal algebra, without mentioning its particular representation by unbounded operators on the physical Hilbert space (similar to perturbative AQFT \cite{pAQFT}).
\end{rmk}

\subsubsection{Generalization to  operator-valued functions and (possibly parameter-dependent) distributions}
From this and the previous subsections it is clear that it is enough to define the secondary quantization for parameter-dependent operator-valued distribution as a map 
$$\LLfirst(\DSS{}{n},\DSSMlt{}{m}{k})$$
for $n,m,k\in\NN_0$.
Then the case of operator-valued functions follows by setting $n=0$ , while parameterless distributions correspond to $m=k=0$. 
\par 
Generalization of Definition \ref{def:DS/2Q/AhatDS} is straightforward:
\begin{Def}\label{def:DS/2Q/AhatGen}
Let $n,m,k,l,l'\in\mathbb{N}_0$ and $ A\in \LL(\DSS{l}{n},\DSSMlt{l'}{m}{k})$. The \emph{second quantization} of $A$ is $\widehat{A}\in\LL(\DSS{}{n},\DSSMlt{}{m}{k})$ defined by
$$
\widehat{A}\Psi=\frac{\sqrt{L!(L-l+l')!}}{(L-l)!}\left(\left(\extD{3L+1,\ldots,3L+m+k}{3L+1,\ldots,3L+m+k}\symmze{3}{}\right)\circ \left(\extD{3l+1,\ldots ,3L}{3l'+1,\ldots, 3L-3l+3l'}A\right)\right)\Psi,
$$
$$\forall\Psi\in \DSS{L}{n}, \qquad L\geq l;
$$
$$
\widehat{A}\Psi=0,\qquad \forall\Psi\in \DSS{L}{n}, \qquad L<l.
$$
\end{Def}
\begin{rmk}\label{rmk:AhatGenCommutes}
By Remark \ref{rmk:Framework/Augment/algebra}, \emph{the
secondary quantization commutes with evaluation}, i.e. for $A$ as in Definition \ref{def:DS/2Q/AhatGen}
$$
\widehat{\underline{A}[f](\bm{x})}=\underline{\widehat{A}}[f](\bm{x}), \,\forall f\in\SRR{n}, \forall \bm{x}\in\RR{m+k}. 
$$
As a consequence, 
$$\widehat{A}^{\dag}=\widehat{A^{\dag}}.$$
\end{rmk}
\begin{rmk}[Wick's theorem for functions and distributions]\label{rmk:DS/2Qgen/Wick}
Take $l_A,l'_A,l_B,l'_B,n,n',n''\in\NN_{0}$, $A\in\LL\left(\DSS{l_A}{n'},\DSS{l'_A}{n''}\right)$ and $B\in\LL\left(\DSS{l_B}{n},\DSS{l'_B}{n'}\right)$. Then Proposition \ref{prop:DS/2Q/Wick} holds in the same symbolic form with 
$$A\otimes_r B\in \LL\left(\DSS{l_A+l_B-r}{n},\DSS{l'_A+l'_B-r}{n''}\right)$$ 
defined precisely as in Definition \ref{def:DS/2Q/contractions}. The proof goes along the same lines. Generalization to the multiplier and mixed spaces goes along the lines of Subsection \ref{Framework/Mlt}.
\end{rmk}
%\TBD{contructed tensor product commutes with evaluation - no tensor product defined}
\begin{rmk}
Let $n,m,k,r\in\NN_0$, $A\in\LLfirst(\DSS{}{n},\DSSMlt{}{m}{k})$, and two tuples of pairwise distinct numbers $\arrs{i}{r}\in\{1,\ldots,n+r\}$ and $\arrs{i'}{r}\in\{1,\ldots,m+k+r\}$. Define
$
B=\extDH{\arrs{i}{r}}{\arrs{i'}{r}}A
$.
Then, by Remark \ref{rmk:Framework/Augment/algebra},
$$
\widehat{B}=\extDH{\arrs{i}{r}}{\arrs{i'}{r}}\widehat{A}.
$$
In other word, \emph{secondary quantization commutes with augmentations}.
\end{rmk}
Further properties of the secondary quantization and its generalization will be stated after some examples.
\subsubsection{Important examples}
We list a few applications of the formalism introduced in this subsection. On the one hand, they illustrate the framework we have constructed. On the other hand, they play important roles in what follows.
\begin{exmp}\label{exmp:DS/2Q/apm}
Define\footnote{Here we use $\DSS{0}{3}=\SRR{3}$. Simple identifications appear in other examples with no further comments.}
$a_{+}\in \LL(\SRR{3},\DSS{}{3})$
$$
a_{+}[f](\vec{p})=f(\vec{p})
$$
and 
$a_{-}\in \LL(\DSS{1}{3},\mathbb{C})$,
$$
a_{-}[f]=\int{f(\vec{p},\vec{p})d^3{\vec{p}}}
.
$$
Then $a_{-}=a_{+}^{\dag}$, thus $\widehat{a}_{-}=\widehat{a}_{+}^{\dag}$. The operator-valued distributions
$
\underline{\widehat{a}_{\pm}}
$
coincide with the the standard creation and annihilation distributions. By Wick's theorem of Remark \ref{rmk:DS/2Qgen/Wick} (or by direct computation),
\begin{equation}
\uwhat{a_{-}}(\vec{p})\uwhat{a_{+}}(\vec{p}')-\uwhat{a_{+}}(\vec{p}')\uwhat{a_{-}}(\vec{p})=\delta^{(3)}(\vec{p}-\vec{p}')\mathbb{1}_{\LDS},
\label{eq:DS/CCR}    
\end{equation}
Where we treat $\delta^{(3)}(\vec{p}-\vec{p}')$ as a (symbolic kernel of a) distribution on $\RR{6}$.
\end{exmp}
\begin{exmp}\label{exmp:DS/2Q/amR}
There is a restriction $a_{-}^R=\restrD{1,2,3}{1,2,3}a_{-}\in\LL(\DS^{3},\SRR{3})$,
$$
a_-^R[\psi](\vec{p})=\psi(\vec{p},\vec{p}),\,\forall \psi\in\DS^{1},\,\forall \vec{p}\in\RR{3}.
$$
It has the secondary quantization $\widehat{a}^{R}_{-}\in\LL(\DS,\DSS{}{3})$ and does not have the conjugate operator. This example is in agreement with the well-known fact that the annihilation (but not the creation) operator is in fact a smooth function of its parameter.
\end{exmp}
\begin{rmk}
The commutation relations for $a_{-}^R$ in place of $a_-$ have the same form, but the delta-function should be understood as a parameter-dependent distribution, namely $$\delta^{(3)}(\vec{p}-\vec{p}')=\mathbb{1}_{\LL(\SRR{3})}$$
(see Notation \ref{notn:Framework/symb/restrict}).
\end{rmk}
\begin{exmp}\label{exmp:DS/2Q/numdist}
Another class of examples comes from the identification
$$
\LL(\SRR{n},\SMlt{m}{k})=\LL(\DSS{0}{n},\DSSMlt{0}{m}{k}).
$$
By Definition \ref{def:DS/2Q/AhatGen}, for $F\in\LL(\SRR{n},\SMlt{m}{k})$,
$$
\widehat{F}=\mathbb{1}_{\DS}\otimes F\in\LL(\DSS{0}{},\DSSMlt{}{m}{k}).
$$
It is consistent with Notation \ref{notn:DS/hat-op}.
The Wick Theorem of Remark \ref{rmk:DS/2Qgen/Wick}, applied to products with one of the multipliers of this class leads to rather trivial result (as the number of contraction is always zero). In particular\footnote{As usual, we put $k=0$ for simplicity},
 for $F\in \LL(\SRR{n},\SRR{m})$, $A\in\LLfirst\left(\DSS{}{n'},\DSS{}{n}\right)$ and $B\in\LLfirst\left(\DSS{}{m},\DSS{}{m'}\right)$ we get $$F\otimes_{\bullet} A=F\otimes_{0} A=\widehat{F}\circ A\in \LLfirst\left(\DSS{}{n'},\DSS{}{m}\right)$$
 and 
 $$B\otimes_{\bullet}F=B\otimes_{0}F=B\circ \widehat{F}\in  \LLfirst\left(\DSS{}{n},\DSS{}{m'}\right).$$
In the symbolic notation we write
$$
\underline{\widehat{F\otimes_{\bullet} A}}\symbDistArgs{\bm{x}}{\bm{y}}=\int
F(\bm{z}|\bm{y})\underline{\widehat{A}}\symbDistArgs{\bm{x}}{\bm{z}}d^{n}\bm{z}, \quad \bm{x}\in\RR{n'}, \bm{y}\in\RR{m},
$$
and
$$
\underline{\widehat{B\otimes_{\bullet}F}}\symbDistArgs{\bm{x}}{\bm{y}}=\int\underline{\widehat{B}}\symbDistArgs{\bm{z}}{\bm{y}}
F\symbDistArgs{\bm{x}}{\bm{z}}d^{m}\bm{z}, \quad,\bm{x}\in\RR{n}, \,\bm{y}\in\bm{m'},
$$
where we omit the hat on the symbolic kernel of the numerical quantization $F$ (as in Notation \ref{notn:DS/symb-opval}). The main lesson we learned here is that composition with (hatted) numerical operator as above can be pulled into the second quantization.
\end{exmp}
We are now ready to define the quantum fields.
\begin{exmp}  \label{exmp:DS/2Q/fields}
Let $\omega_0$ be a massive dispersion function (Definition \ref{def:dispRel}). 
We define $\widetilde{\phi}^{R}_0\in\LLfirst(\DSS{}{3},\DSMlt{}{1})$ as
$$
\widetilde{\phi}^R_0=\widetilde{\phi}^R_{0+}+\widetilde{\phi}^R_{0-},
$$
where\footnote{We use properties of $\omega_0$ postulated in Definition \ref{def:dispRel}}
\begin{equation}\label{eq:DS/2Q/field-def}
\widetilde{\phi}^R_{0\pm}\symbDistArgs{\vec{p}}{t}=a_{\pm}(\pm \vec{p})\varphi_{\pm}(t,\vec{p}),
\end{equation}
$$
\varphi_{\pm} \in \Mlt{4},\quad \varphi_{\pm}(\vec{p},t)=\frac{e^{\mp \ii \omega_0(\vec{p})t}}{\sqrt{(2\pi)^3 2\omega_0(\vec{p})}}\, (\forall \vec{p}\in\RR{3},\,\forall t\in\RR{}).
$$
We also set $\widetilde{\phi}_0=\intDH{1}{1}\widetilde{\phi}_0^R$.
Then
$$
\uwhat{\phi}_{0,\pm }(t,\vec{k})=\uwhat{\phi}_{0,\pm }^R\symbDistArgs{\vec{k}}{t}=\uwhat{a}_{\pm}(\vec{k})\frac{e^{\mp \ii \omega_0(\vec{k})t}}{\sqrt{(2\pi)^3 2\omega_0(\vec{k})}}.
$$
which is nothing but the positive and the negative frequency parts of the partial Fourier transform of the real scalar quantum field.
\end{exmp}
\begin{rmk}
For the completeness we present a translation of (\ref{eq:DS/2Q/field-def}) from the symbolic language. 
We take arbitrary $f\in\DSS{}{3}$ and write
$$
\underline{\widetilde{\phi}^R_{0\pm}\circ f}=\int \underline{a}_{\pm}(\pm \vec{p})\varphi_{\pm}(t,\vec{p})\underline{f}(t,\vec{p})\psi dt d^3{\vec{p}}=\widetilde{\phi}_{0\pm}[f]=\int \underline{a}_{\pm}( \vec{p})\varphi_{\pm}(t,\pm\vec{p})\underline{f}(\pm\vec{p})dt d^3{\vec{p}}.
$$
As $\varphi_{\pm}\in\Mlt{4}$, we present multiplication by it as $\MltM{\varphi_{\pm}}\extD{1,2,3}{2,3,4}{\mathbb{1}_{\Mlt{1}}}$ and add a hat according to Notation \ref{notn:DS/symb-opval}. Treating the linear transform once again by Notations \ref{notn:Framework/symb/stdop} and \ref{notn:DS/symb-opval} we arrive to
$$
\widetilde{\phi}^R_{0\pm}=
 a_{\pm} \circ \extDH{1}{1}\widehat{\left(\pm\mathbb{1}_{\LL(\RR{3})}\right)}  \circ \widehat{\MltM{\varphi_{\pm}}}\circ \extDH{1,2,3}{2,3,4}{\widehat{\mathbb{1}_{\Mlt{1}}}}.
$$
\end{rmk}
We also set the position space presentation for the quantum fields as $\phi_0^R=\widetilde{\phi_0^R}\circ \Fourier{3}{+}$ and $\phi_0=\intDH{1}{1}\phi_0^R$. Symbolically we have, for example,
$$
\phi_0^R\symbDistArgs{t}{\vec{x}}=\int e^{\ii \vec{k}\cdot\vec{x}} \widetilde{\phi}_0^R(\vec{k})d^3\vec{k}.
$$

\begin{rmk}\label{rmk:DS/2Q/symbolic}
With Examples \ref{exmp:DS/2Q/apm}-\ref{exmp:DS/2Q/numdist}  the expression (\ref{eq:DS/2Q-symbolic}) gets a precise sense. Its translation from the symbolic language for $A\in\LL(\DS^{l},\DS^{l'})$  ($l,l'\in\NN_0$) is 
$$
\widehat{A}=\widehat{a}_{+}^{\hotimes l'}\circ \widehat{A\circ \symmze{3}{}}\circ \widehat{a}_{-}^{R\hotimes l},
$$
where $A\circ \symmze{3}{}\in{\LL\left(\SRR{3l},\SRR{3l'}\right)}$ is a numerical distribution, so for $\widehat{A\circ \symmze{3}{}}$ we can use Notation \ref{notn:DS/hat-op}.
Similar expression can be written for general Definition \ref{def:DS/2Q/AhatGen}.
\end{rmk}
\begin{rmk}\label{rmk:DS/2Q/formal-alg-gen}
The Wick's theorem of Remark \ref{rmk:DS/2Qgen/Wick} may be derived from (\ref{eq:DS/CCR}) by means of the symbolic calculus. 
In this sense the spaces $\LLfirst\left(\DSS{}{n},\DSS{}{m}\right)$ can be considered as formal spaces generated by $a_{+}$ and $a_{-}^R$ with the product $\otimes_{\bullet}$, the partially defined involution $\dag$ and operations of Notation \ref{notn:DS/hat-op}. As was anticipated in Remark \ref{rmk:DS/2Q/formal-alg}, this allows to forget about the particular realization of the elements of $\LLfirst\left(\DSS{}{n},\DSS{}{m}\right)$ as parameter-dependent (unbounded) operator-valued distributions acting on $\mathcal{H}_{\rm{phys}}$ and define Hamiltonian perturbative QFT (Section \ref{HpQFT}) on the formal algebra language.
\end{rmk}
For future use we also define the Wick products. Note that the singular products never appear in the non-local quantum field theory, so we can survive with a very simple definition.
\begin{exmp}\label{exmp:DS/2Q/exmp-norm}
For $n\in\NN_0$ and $\arrs{\alpha}{n}\in\{+,-\}$ we set 
\begin{equation}\label{eq:DS/2Q/exmp/normProd}
:\prod_{j=1}^n \uwhat{a}_{\alpha_j}(\vec{p}_j): =\uwhat{a_{:\arrs{\alpha}{n}}:}(\arrs{\vec{p}}{n}),
\end{equation}
where $a_{:\arrs{\alpha}{n}:}\in\LLfirst(\DSS{}{3n},\DS{})$ is uniquely defined by
\begin{equation}
a_{\arrs{\alpha}{n}}\left[\bigotimes_{j=1}^n f_j \right]=\symmze{3}{}\left[\bigotimes_{j=1}^n \underline{a_{\alpha_j}}[f_j]\right], \,\forall \arrs{f}{n}\in\RR{3}.
\label{eq:DS/2Q/exmp/normProd-first}    
\end{equation}
To see that (\ref{eq:DS/2Q/exmp/normProd}) indeed defines the Wick product, note that the tensor product is the same as product with no contructions. In particular, if $\alpha_j$ are normally ordered (i.e. if $\alpha_i=+$ and $\alpha_j=-$, then $i<j$), then 
$$
:\prod_{j=1}^n \uwhat{a}_{\alpha_j}(\vec{p}_j):=\left(\widehat{\bigotimes}_{j=1}^n\uwhat{a}_{\alpha_j} \right)(\arrs{\vec{p}}{n}),
$$
and the other cases follow by obvious symmetry of (\ref{eq:DS/2Q/exmp/normProd-first}). It is, of course, possible to present explicitly for each $l\in\NN_0$ action of such operator on $\DSS{l}{3n}$, but we omit it to avoid unecessary complicated combinatorics.
\end{exmp}
\begin{exmp}\label{exmp:DS/2Q/exmp-norm-fields}
Continuing the previous example, we set 
$$
:\prod_{j=1}^n\uwhat{\widetilde{\phi}}^R_{0\alpha_j}\symbDistArgs{\vec{p}_j}{t_j}:=\uwhat{a_{:\arrs{\alpha}{n}}:}(\arrs{\vec{p}}{n})\prod_{j=1}^n \varphi_{\alpha_j}(t_j,\vec{p}_j),
$$
and
$$:\prod_{j=1}^n\uwhat{\widetilde{\phi}}^R_{0\alpha_j}\symbDistArgs{\vec{p}_j}{t_j}:=\sum_{\arrs{\alpha}{n}=\pm}:\prod_{j=1}^n\uwhat{\phi}^R_{0\alpha_j}\symbDistArgs{\vec{p}_j}{t_j},
$$
and
$$:\prod_{j=1}^n\uwhat{\phi}^R_{0\alpha_j}\symbDistArgs{\vec{x}_j}{t_j}:=\int
:\prod_{j=1}^n\uwhat{\widetilde{\phi}}^R_{0\alpha_j}\symbDistArgs{\vec{p}_j}{t_j} e^{\ii\sum_{j=1}^n\vec{p}_j\cdot \vec{x}_j}
d^{3n}\arrs{\vec{p}}{n}.
$$
Clearly, all these operators can be presented as second quantization of some objects in $\LLfirst(\DSS{}{3n},\DSMlt{}{1})$.  
The unrestricted form and the can be defined similarly.
\end{exmp}
%%%%%%%%%%%%%%HPQFT%%%%%%%%%%%%%%%%%%%%
%%%%%%%%%%%%%%%%%%%%%%%%%%%%%%%%%%%

%%%%%%%%%%%%%%%%%%%%%%%HpQFT

\section{Non-local Hamiltonian Perturbation Quantum Field Theory}\label{HpQFT}
\subsection{Motivation: Hamiltonian Perturbation Quantum Field Theory}
 In order to motivate the technical constructions of this section and fix the terminology we briefly (and rather informally) recall how the Hamiltonian perturbation theory is usually constructed. A rigorous version of this construction for a class of non-local quantum field theories is presented in the next subsections.
 \par
 The goal  is to construct the operator-valued distributions $\phi$, $\pi$ satisfying the commutation relations
 \begin{equation}\label{eq:HpQFT/AdmTh/canonicalQuant}
[\phi(t,\vec{x}),\pi(t,\vec{x}')]=-i\delta(\vec{x}-\vec{x}').    
\end{equation}
and the equations of motion
\begin{equation}
    \label{eq:HpQFT/AdmTh/HamEv}
\partial_t \phi(x)=-i [H(t),\phi(t,\vec{x})],
\end{equation}
\begin{equation}\label{eq:HpQFT/AdmTh/HamEvPi}
\partial_t \pi(x)=-i [H(t),\pi(t,\vec{x})].
\end{equation}
Here the Hamiltonian operator $H(t)$
$$
H(t)=H_0(t)+H_{int}(t). 
$$
The first term, the free Hamiltonian $H_0$ is chosen so that the solution of  (\ref{eq:HpQFT/AdmTh/canonicalQuant}-\ref{eq:HpQFT/AdmTh/HamEvPi}) for $H_0$ in place of $H$ is known. More precisely, we are given $\phi_0$, $\pi_0$ such, that 
 \begin{equation}\label{eq:HpQFT/AdmTh/canonicalQuant0}
[\phi_0(t,\vec{x}),\pi_0(t,\vec{x}')]=-i\delta(\vec{x}-\vec{x}').    
\end{equation}
\begin{equation}
    \label{eq:HpQFT/AdmTh/HamEv0}
\partial_t \phi_0(x)=-i [H(t),\phi_0(t,\vec{x})],
\end{equation}
\begin{equation}\label{eq:HpQFT/AdmTh/HamEvPi0}
\partial_t \pi_0(x)=-i [H(t),\pi_0(t,\vec{x})].
\end{equation}
The interaction part of the Hamiltonian $H_{int}(t)$ has the form
$$
H_{int}(t)=g\int{h_{int}(t,\vec{x})\lambda(t,\vec{x})}d^3\vec{x},
$$
where $h_int(x)$
is a fixed translationally-invariant polynomial functional of the field $\phi(t,\vec{x})$ and its conjugated momentum $\pi(t,\vec{x})$. The function $\lambda(t,\vec{x})$ is the adiabatic cut-off. As mentioned in Introduction, its presence is necessary due to the Haag theorem, forbidding the existence of the unitary equivalence between the free and the interacting theories (\ref{eq:HpQFT/AdmTh/IntPhi}) below for translationally-invariant interaction. At this moment we need $\lambda(t,\vec{x})$ to decay then $t\rightarrow \infty$  in order to make the integral in (\ref{eq:HpQFT/AdmTh/UT}) below converge.  In the next subsection we will see that $\lambda$ should also decay whenever $\vec{x}\rightarrow \infty$ for $H_{int}$ to be a well-defined unbounded operator. The function $\lambda$ should also be smooth enough to make the switching adiabatic. These conditions are guaranteed if we assume $\lambda \in \SRR{4}$. To get physically relevant information we should pass to the limit $\lambda(t,\vec{x})\rightarrow 1$ in the appropriate sense. 
Finally, $g$ is a formal parameter introduced for convenience. 
\par 
Then we construct a formal solution of (\ref{eq:HpQFT/AdmTh/canonicalQuant}-\ref{eq:HpQFT/AdmTh/HamEvPi}) in the form
\begin{equation}\label{eq:HpQFT/AdmTh/IntPhi}
\phi(\vec{x},t)=U(t)^{-1}\phi_0(\vec{x},t)U(t),
\end{equation}
\begin{equation}
\label{eq:HpQFT/AdmTh/IntPi}
\pi(\vec{x},t)=U(t)^{-1}\pi_0(\vec{x},t)U(t),
\end{equation}
where $U(t)$ is given by

\begin{equation}\label{eq:HpQFT/AdmTh/UT}
U(t)=\sum_{n} 
(-i)^n \int_{-\infty}^t dt_1 \int_{-\infty}^{t_1} dt_2\ldots \int_{-\infty}^{t_{n-1}} dt_n {\prod_{j=1}^n H_{I}(t_j)d^{n}\arrs{t}{n}}=
\end{equation}
$$
\timeorder{e^{-i \int_{\tau<t'}{H_{I}(\tau)d\tau}}},
$$
where $\timeorder{\ldots}$ is the symbolic time-ordering operator and
\begin{equation}\label{eq:HpQFT/AdmTh/HI}
H_I(t)=U(t)H_{int}(t)U(t)^{-1}.
\end{equation}
Assuming that $H_{int}$ is symmetric, we may expect that $U(t)$ is (in some sense) unitary.
\par
We claim that (\ref{eq:HpQFT/AdmTh/IntPhi}-\ref{eq:HpQFT/AdmTh/IntPi}) is a solution of  (\ref{eq:HpQFT/AdmTh/canonicalQuant}-\ref{eq:HpQFT/AdmTh/HamEvPi}). First, (\ref{eq:HpQFT/AdmTh/IntPhi}-\ref{eq:HpQFT/AdmTh/IntPi}) together with  (\ref{eq:HpQFT/AdmTh/canonicalQuant0}) imply the commutation relations (\ref{eq:HpQFT/AdmTh/canonicalQuant}). To get the equations of motion we formally differentiate (\ref{eq:HpQFT/AdmTh/UT}) to get
$$
\frac{dU(t)}{dt}=-i H_I(t)U(t).
$$
Now differentiating (\ref{eq:HpQFT/AdmTh/HamEvPi0}-\ref{eq:HpQFT/AdmTh/HamEv0}) and substituting (\ref{eq:HpQFT/AdmTh/HI})
we arrive to (\ref{eq:HpQFT/AdmTh/HamEvPi}-\ref{eq:HpQFT/AdmTh/HamEv}). Now let us assume for simplicity that $H_{int}(t)$ is a polynomial functional of $\phi$ and $\pi$ at the time $t$. Then by (\ref{eq:HpQFT/AdmTh/IntPhi}-\ref{eq:HpQFT/AdmTh/IntPi}) definition of $H_I$ (\ref{eq:HpQFT/AdmTh/HI}) may be rewritten as
\begin{equation}\label{eq:HpQFT/AdmTh/Hloc}
H_I(t)=H_{int}(t)_{\phi\rightarrow \phi_0,\quad \pi\rightarrow \pi_0}.
\end{equation}
\par 
There are (at least) two ways to get physically relevant information from this construction. First of all, we may the scattering operator. For simplicity of the interpretation let us assume that $\lambda(y,\vec{x})$ vanishes for $t>|T|$. Then the same holds for $H_{int}(t)$ and thus $H_{I}(t)$. We have 
$$U(t)=\bm{1}, \quad t<-T;\qquad U(t)=S, \quad t>T $$
 for some unitary operator $S$. Then from (\ref{eq:HpQFT/AdmTh/IntPhi}-\ref{eq:HpQFT/AdmTh/IntPi}) we see that in the distant past the interacting field $\phi$ coincides with the free one, while in the distant future it belongs to a unitary equivalent realization of the same commutation relations. These two realizations induce two interpretations of the Hilbert space as a Fock space, giving rise the states of incoming and the outgoing particles respectively. Then we can recognize the operator $S$, intertwining these two realizations, as the scattering operator. This object has a direct physical interpretation, but its adiabatic limit (see next subsection or the introduction for a discussion of the adiabatic limit) is rather delicate \cite{EG76}. We postpone it for the next publication \cite{PAP1}.
\par 
Another type of objects of interest are the correlators. We consider two families of distributions\footnote{We use a non-standard partially Fourier transformed presentation of all functions for technical convenience.}:
\begin{itemize}
    \item The \emph{Wightman functions} $\mathcal{W}_n\in\SRR{4n}$
    \begin{equation}
    \mathcal{W}_n\left((t_j,\vec{p}_j)_{j=1\ldots n}\right)=\left(\Omega,\prod_{j=1}^{\infty}\widetilde{\phi}(\vec{t}_j,\vec{p}_j)\Omega\right)
    \label{eq:HpQFT/WightmanDef}    
    \end{equation}
    \item The \emph{Green functions} $\mathcal{G}_n\in\SRR{4n}$
    \begin{equation}
    \mathcal{G}_n\left((t_j,\vec{p}_j)_{j=1\ldots n}\right)=\left(\Omega,\timeorder{\prod_{j=1}^{\infty}\widetilde{\phi}(\vec{t}_j,\vec{p}_j)}\Omega\right).
    \label{eq:HpQFT/GreenDef}    
    \end{equation}
    \end{itemize}
Equations (\ref{eq:HpQFT/WightmanDef}-\ref{eq:HpQFT/GreenDef}) are written already in the adiabatic limit (or at least the temporal adiabatic limit $\lambda(\vec{x},t)\rightarrow\lambda_{\mathrm{S}}(\vec{x})$), and $\Omega$ is the vacuum state of the full Hamiltonian. To compute it in the frame of the perturbation theory one may use the Gell-Man and Low theorem \cite{GL, Molinari} giving
\begin{equation}\label{eq:HpQFT/AdmTh/WightDef}
\mathcal{W}_n=\lim_{\lambda(\vec{x},t)\rightarrow\lambda_{\mathrm{S}}(\vec{x})} \frac{W_n}{W_0},\,\forall n\in\NN_0
\end{equation}

and
\begin{equation}\label{eq:HpQFT/AdmTh/GreenDef}
\mathcal{G}_n=\lim_{\lambda(\vec{x},t)\rightarrow\lambda_{\mathrm{S}}(\vec{x})} \frac{G_n}{G_0},\,\forall n\in\NN_0
\end{equation}

where 
\begin{equation}
W_n\left((t_j,\vec{p}_j)_{j=1\ldots n}\right)=
\left(\Omega_0,S\prod_{j=1}^n \widetilde{\phi}_0(t_j,\vec{p}_j)\Omega_0\right)
=
\end{equation}
$$
\left(\Omega_0,U(+\infty,t_1)\prod_{j=1}^n \widetilde{\phi}_0(t_j,\vec{p}_j)U(t_j,t_{j+1})\Omega_0\right)$$
and
\begin{equation}\label{eq:HpQFT/GreenUnNormDef}
    G_n\left((t_j,\vec{p}_j)_{j=1\ldots n}\right)=\left(\Omega_0,S\timeorder{\prod_{j=1}^n \widetilde{\phi}_0(t_j,\vec{p}_j) }\Omega_0\right)=
\end{equation}
$$
\left(\Omega_0,\timeorder{\prod_{j=1}^n\widetilde{\phi}_0(t_j,\vec{p}_j)e^{-i \int_{-\infty}^{+\infty}{H_{I}(\tau)d\tau}}}\Omega_0\right)$$
with $\Omega_0$ being the vacuum vector of the free theory (in accordance with the one defined in Section \ref{DS}), assuming that the limits above exist. This limit is known as the weak adiabatic limit and is the central object of this paper. Comparing the weak and strong adiabatic limit existence theorems \cite{EG73,EG76} we may expect that the former will be much easier to prove than the latter.
\par 
The LSZ procedure\footnote{The LSZ reduction for non-local theories was considered in \cite{PhDThesis}.} gives a connection between the strong and weak adiabatic limit, allowing to retrieve the matrix elements of $S$ from the residues of the Fourier transform of the Green functions.
\par 
Before going further let us briefly summarize the unclarified issues of the exposition above.
\begin{rmk}\label{rmk:HpQFT/AdmTh/fields-dist}
The objects $\phi$ and $\pi$ (respectively $\phi_0$ and $\pi_0$) satisfying (\ref{eq:HpQFT/AdmTh/canonicalQuant}) (respectively (\ref{eq:HpQFT/AdmTh/canonicalQuant0})) may exist only as distributions valued in unbounded operators. We build them as elements of $\LL(\DSS{}{4},\DS{})$ which has restrictions\footnote{Otherwise the simultaneous commutation relations can not be defined.} and the mentioned relations must be rewritten more carefully, as was done in Example \ref{exmp:DS/2Q/fields}. In other words, they make sense then both hands sides are applied to a vector in $\DSS{}{n}$ with appropriate $n\in\mathbb{N}_0$. 
\end{rmk}
\begin{rmk}
We deal with right-hand sides of (\ref{eq:HpQFT/AdmTh/UT}) as with a formal power series in terms of the formal parameter $g$. Thus, in addition to the previous Remark, to make sense of all expressions containing $H$, $H_{int}$, $S$, $U$, $\phi$ and $\pi$ we have to truncate them at arbitrary finite order (see Subsection \ref{Intro/prelim}). 
\end{rmk}
\begin{rmk}
To make sense of  (\ref{eq:HpQFT/AdmTh/IntPhi}-\ref{eq:HpQFT/AdmTh/HI}) and dependent equations we have to treat $t$ as a parameter rather than a distributional argument. In Remark \ref{rmk:HpQFT/AdmTh/fields-dist} we already noted that by writing  (\ref{eq:HpQFT/AdmTh/canonicalQuant}-  \ref{eq:HpQFT/AdmTh/HamEvPi0}) we have implicitly assumed the existence of restrictions of operator-valued distributions $\phi$, $\pi$, $\phi_0$ and $\pi_0$ to fixed values of $t$. We also have to require that $U$ and $H_I$ are functions of $t$. In local QFT, due to the UV singularities, $H_I$ exists only as a distribution and thus all expressions including the time-ordered product require a regularization. This problem does not appear in the class of non-local quantum field theories we define in the rest of the section.
\end{rmk}

\subsection{Admissible Quantum Field Theories}\label{HpQFT/AdmTh}
Here we define a class of admissible quantum field theories for which the construction of the previous subsection becomes precise.
\par 
\subsubsection{Free Quantum Field and admissible dispersion relation}
 Instead of the fields themselves, in the non-local case it is more convenient to work with their partial Fourier transforms. The appropriate definition of such an object was already given in Example \ref{exmp:DS/2Q/fields}. We use it always assuming the mass gap. Analogously, within the same notation we can define the momentum $\widetilde{\pi}_0$ as
$$
\pi_0(t,\vec{k})=\partial_{t}\phi_0(t,\vec{k}).
$$
Using the Wick Theorem of Remark \ref{rmk:DS/2Qgen/Wick} we may show that 
$$
[\uwhat{\phi}(t,\vec{k}),\uwhat{\pi}(t,\vec{k}')]=-i\widehat{\delta}(\vec{k}+\vec{k'})\mathbb{1}_{\LL(\DS)}%\TBD{verify!}
$$
which is the precise form of the Fourier transform of (\ref{eq:HpQFT/AdmTh/canonicalQuant0}).
\subsubsection{Admissible interaction}
The next ingredient is the interaction $H_I$. Our approach mostly follows \cite{PhDThesis}. 
\begin{Def}\label{def:HpQFT/AdmInt}
We say  that $h_I\in\LLfirst(\DSS{}{3},\DSMlt{}{1})$ is an \emph{admissible interaction density} if for each  $l,l'\in\mathbb{N}_0$ 
there is$F_{(l,l')}\in\Mlt{3(l+l')}$, which we call \emph{the admissible interaction kernels}  such that
\begin{equation}\label{eq:def:HpQFT/AdmInt}
h_I[\psi]_{l'}(\arrs{\vec{p'}}{l'},t)=\sum_{l=0}^{\infty}\int{F'_{(l',l)}(\arrs{\vec{p'}}{l'},\arrs{\vec{p}}{l})\widehat{\Fourier{-}{3}}\psi\left(\arrs{\vec{p}}{l},\sum_{j=1}^{l'}\vec{p'}_j-\sum_{j=1}^{l}\vec{p}_{j}\right)}\times
\end{equation}
$$
\exp\left(\ii\left(\sum_{j=1}^{l'}\omega_0(\vec{p}'_j)-\sum_{j=1}^{l}\omega_0(\vec{p}_j)\right)t \right)
d^{3l}\arrs{\vec{p}}{l}, $$
$$
\forall \psi\in\DSS{l'}{3}, \quad \forall \arrs{\vec{p'}}{l'}\in \RR{3l'}, \, \forall t\in \RR{},
$$
\begin{equation}\label{eq:def:HpQFT/AdmInt/KernelGrowth}
\left((\arrs{\vec{p}}{l},\arrs{\vec{p}'}{l'})\mapsto 
 F_{(l,l')} s\left(\sum_{j=1}^{l'}\vec{p'}_j-\sum_{j=1}^{l}\vec{p}_{j}\right)\right) \in\SRR{3(l+l')},\, \forall s\in\SRR{3(l+l')}    
\end{equation}
and
$$
h_I^{\dag}=h_I.
$$
\end{Def}
The following remarks clarify the technical definition above.

\begin{rmk}
The condition (\ref{eq:def:HpQFT/AdmInt/KernelGrowth}) basically limits the directions in which $F_{(l,l')}$ can grow. We could instead change the variables of integration in (\ref{eq:def:HpQFT/AdmInt}) and ask the kernel to belong to the mixed space, but that would break the symmetry of Definition \ref{def:HpQFT/AdmInt}. 
\end{rmk}

\begin{rmk} An admissible interaction is characterized by its kernels. These kernels are constrained by the following.
By definition of $\LLfirst$ there can be only finitely many non-zero kernels among $F_{(l',l)}$. Each such kernel should be symmetric with respect to its first $l'$ arguments for $h_I$ to value in $\DSS{}{1}$. Symmetry in the next $l$ arguments can be always assumed as we do from now on. Finally, the ''hermicity'' condition translates to
$$
F'_{(l',l)}(\arrs{\vec{p'}}{l'},\arrs{\vec{p}}{l})=\overline{F'_{(l,l')}(\arrs{\vec{p}}{l},\arrs{\vec{p'}}{l'})}, \quad \forall l,l'\in \mathbb{N}_0, \,\forall \arrs{\vec{p}}{l}\RR{3l},\,\arrs{\vec{p'}}{l'}\in\RR{3l'}. 
$$
It is easy to see that there is a one-to-one correspondence between the admissible interaction densities and the families of kernels satisfying all conditions listed here.
\end{rmk}

\begin{rmk}\label{rmk:HpQFT/AdmTh/forms-of-HI}
In the symbolic notation Definition \ref{def:HpQFT/AdmInt} reads as\footnote{to make it precise we have to use presentation of $F_{l,l'}$ via $F'_{l,l'}$ as in Definition \ref{eq:def:HpQFT/AdmInt}.}
$$
\uwhat{h_I}\symbDistArgs{\vec{x}}{t}=\sum_{l,l'=0}^{\infty}\frac{1}{l!l'!}\int
F_{(l',l)}(\arrs{\vec{p'}}{l'},\arrs{\vec{p}}{l})\prod_{j=1}^{l'}\uwhat{a}_+(\vec{p}'_j)\prod_{j=1}^{l}\uwhat{a}_-(\vec{p}_j)$$
$$
\exp\left(\ii\left(\sum_{j=1}^{l'}\omega_0(\vec{p}'_j)-\sum_{j=1}^{l}\omega_0(\vec{p}_j)\right) -\ii \left(\sum_{j=1}^{l'}\vec{p'}_j-\sum_{j=1}^{l}\vec{p}_{j}\right)\vec{x}\right)
d^{3l'}\arrs{\vec{p}'}{l'}d^{3l}\arrs{\vec{p}}{l}.
$$
Alternatively, one may write
\begin{equation}\label{eq:HpQFT/AdmTh/hI-timeloc}
\uwhat{h_I}\symbDistArgs{\vec{x}}{t}=\sum_{n=0}^{\infty}\frac{1}{n!}\sum_{\arrs{\alpha}{n}\in\{+,-\}}\int \mathcal{F}^{\arrs{\alpha}{n}}(\arrs{\vec{p}}{n})\prod_{j=1}^{n}{\mathpunct{:}}\uwhat{\widetilde{\phi}}_{0,\alpha_j}(t,-\alpha_j\vec{p}_j){\colon}
\end{equation}
$$
e^{\ii\left(\sum_{j=1}^{l}\vec{p}_{j}\right)\vec{x}}d^{3n}\arrs{\vec{p}}{n},
$$
with 
$$
\mathcal{F}^{\overbrace{+\ldots+}^{l'}\overbrace{-\ldots-}^{l}}(\arrs{\vec{p'}}{l'},\arrs{\vec{p}}{l})=\mathcal{F}^{(l',l)}(\arrs{\vec{p'}}{l'},\arrs{\vec{p}}{l})=$$
$$
 F_{(l',l)}(-\arrs{\vec{p'}}{l'},\arrs{\vec{p}}{l})\left(\prod_{j=1}^l\sqrt{(2\pi)^32\omega_0(\vec{p}_k)}\right)\left( \prod_{j=1}^{l'}\sqrt{(2\pi)^32\omega_0(\vec{p}_k)}\right),
$$
$$
\mathcal{F}_{\arrsP{\alpha}{n}{\sigma}}(\arrsP{\vec{k}}{n}{\sigma})=\mathcal{F}_{\arrs{\alpha}{n}}(\arrs{\vec{k}}{n}), \, \forall \arrs{\vec{k}}{n}\in\RR{3n},\, \forall \arrs{\alpha}{n}\in\{+,-\}^{n},  \, \forall \sigma\in\symmgr{n}.
$$
 We use this form in the formulation of the Feynman rules in Appendix \ref{appFR}.
 \end{rmk}
 \begin{rmk} \label{rmk:HpQFT/AdmTh/forms-of-HI-pos}
 If $F_{(l,l')}\in\SRR{3(l+l')}$, we may further write
\begin{equation}\label{eq:HpQFT/AdmTh/hI-timeloc-pos}
    \uwhat{h_I}\symbDistArgs{\vec{x}}{t}=\sum_{n=0}^{\infty}\frac{1}{n!}\sum_{\arrs{\alpha}{n}\in\{+,-\}}\int \mathcal{K}^{\arrs{\alpha}{n}}((\vec{x}_j-\vec{x})_{j=1,\ldots,n})\prod_{j=1}^{n}{\mathpunct{:}}\uwhat{{\phi}}_{0,\alpha_j}(t,\vec{x}_j){\colon}d^{3n}\arrs{\vec{x}}{n},
\end{equation}

where $\mathcal{K}^{\arrs{\alpha}{n}}\in\SRR{3n}$ is given by the relation
$$
\mathcal{K}^{\arrs{\alpha}{n}}(\arrs{\vec{x}}{n})=\int{\mathcal{F}^{\arrs{\alpha}{n}}(\arrs{\vec{x}}{n})
e^{-\ii \sum_{j=1}^n \alpha_j \vec{p}_j\vec{x}_j}}d^{3n}\arrs{\vec{x}}{n}.
$$
 The clear advantage of this form is the manifest translational invariance. Unfortunately,
in general case  $\mathcal{K}$ is only a distribution, so (\ref{eq:HpQFT/AdmTh/hI-timeloc-pos}) may have no sense. 
\end{rmk}
\begin{exmp}
Let $\kappa\in\SRR{4n}$. Then
\begin{equation}\label{eq:HpQFT/AdmInt/exmp-smooth}
\uwhat{h_I}\symbDistArgs{\vec{x}}{t}=
\end{equation}
\begin{equation}
    \int
:\prod_{j=1}^n \uwhat{\widetilde{\phi}}_0(t_j,\vec{x}_j):\kappa\left((\vec{x}_j-\vec{x},{t}_j-t)_{j=1\ldots n}\right)d^{3n}\arrs{\vec{x}}{n}d^n\arrs{t}{n}
\label{eq:HpQFT/AdmInt/hI-kappa}
\end{equation}
defines an admissible interaction density.
\end{exmp}
\begin{proof}
Note, that (\ref{eq:HpQFT/AdmInt/exmp-smooth}) is clearly well-defined, as it is essentially evaluation of a distribution on a test-function. It involves a non-injective linear transform $(t,\arrs{t}{n})\mapsto ((t_j-t)_{j=1\ldots n})$, but this is not a problem since we allow $h_I$ to  be a multiplier.  
\par 
Now we use Example \ref{exmp:DS/2Q/exmp-norm-fields} to write
$$
h_I\symbDistArgs{\vec{x}}{t}= \int
a_{:\arrs{\alpha}{n}:}(\arrs{\vec{p}}{n})\left(\prod_{j=1}^n\varphi_{\alpha_j}(t_j,\alpha_j\vec{p}_j)e^{\alpha_j\ii \vec{p}_j\vec{x}_j}\right)\times$$
$$
\kappa\left((\vec{x}_j-\vec{x},{t}_j-t)_{j=1\ldots n}\right)d^{3n}\arrs{\vec{x}}{n}d^n\arrs{t}{n}d^{3n}\arrs{\vec{p}}{n}=
$$
$$
\int
a_{:\arrs{\alpha}{n}:}(\arrs{\vec{p}}{n})e^{\ii\left(\sum_{j=1}^n\alpha_j\vec{p}_j\right)\cdot \vec{x}-\ii\left(\sum_{j=1}^n\alpha_j\omega_0(\vec{p}_j)\right)t}F_{\arrs{\alpha}{n}}(\arrs{\vec{p}}{n})d^{3n}\arrs{\vec{p}}{n},
$$
where
$$
F_{\arrs{\alpha}{n}}(\arrs{\vec{p}}{n})=\int
\left(\prod_{j=1}^n\varphi_{\alpha_j}(t_j,\alpha_j\vec{p}_j)e^{\alpha_j\ii \vec{p}_j\vec{x}_j}\right)\times$$
$$
\kappa((\vec{x}_j,{t}_j)_{j=1\ldots n})d^{3n}\arrs{\vec{x}}{n}d^n\arrs{t}{n}.
$$
Then it is easy to see that $F_{\arrs{\alpha}{n}}\in\SRR{3n}$ and we arrive to the form of Remark \ref{rmk:HpQFT/AdmTh/forms-of-HI}.
\end{proof}
Informally, one may say that the admissible interactions are the ones of the form (\ref{eq:HpQFT/AdmInt/hI-kappa}), but with slightly singular $\kappa$. There are two allowed types of singularities: first, we can use the fact that the quantum field restricts with respect to time variable, second we can use the fact that $h_I$ is a distribution with respect to $\vec{x}$. The first type we have already seen in Remarks \ref{rmk:HpQFT/AdmTh/forms-of-HI} and \ref{rmk:HpQFT/AdmTh/forms-of-HI-pos}. As for the second one, we consider the following.
\begin{exmp}[Quantum Wick Product]
The Quantum Wick Product is a combination of the Quantum Diagonal Map \cite{BahnsEtAl,BahnsPhD}, a well-behaved replacement of the pointwise product in the DFR quantum spacetimes \cite{AreaDistanceVolume}, with the Wick product of quantum field theory. It defines an admissible interaction 
$$
\uwhat{h_I}\symbDistArgs{\vec{x}}{t}=\int
\delta\left(t-\frac{1}{n}\sum_{j=1}^n t_j\right)\delta^{(3)}\left(\vec{x}-\frac{1}{n}\sum_{j=1}^n \vec{x}_j\right)\exp \left( -\frac{1}{2l_P^2} \sum_{j=1}^n ((t_j-t)^2) \right)$$
$$:\prod_{j=1}^n \uwhat{\widetilde{\phi}}^R_0(t_j|\vec{x}_j): 
\exp \left( -\frac{1}{2l_P^2} \sum_{j=1}^n ((\vec{x}_j-\vec{x})^2) \right)d^{3n}\arrs{\vec{x}}{n}d^n\arrs{t}{n},
$$
where $l_{P}$ is the Planck scale.
\end{exmp}
\begin{proof}
To see well-definedness let us multiply with $s\in\SRR{1}$ (since the result is in the multipliers space) and act on $f\in\DSS{}{3}$. After change of variables we get:
$$
\uwhat{h_I}[f](t)s(t)=\int s(t)
\delta\left(t-\frac{1}{n}\sum_{j=1}^n t_j\right)\delta^{(3)}\left(\vec{y}\right)\times$$
$$\exp \left( -\frac{1}{2l_P^2} \sum_{j=1}^n ((t_j-t)^2) \right)s(t):\prod_{j=1}^n \uwhat{\widetilde{\phi}}^R_0(t_j|\vec{x}_j): \times$$
$$
\exp \left( -\frac{1}{2l_P^2} \sum_{j=1}^n ((\vec{x}_j-\vec{y}-\frac{1}{n}\sum_{j=1}^n \vec{x}_j)^2) \right)\underline{f}\left(\vec{y}+\frac{1}{n}\sum_{j=1}^n \vec{x}_j\right)d^{3n}\arrs{\vec{x}}{n}d^n\arrs{t}{n}.
$$
Reading from the right to the left, we see multiplication by 
$$
\left((\vec{y},\arrs{\vec{x}}{n})\mapsto
\exp \left( -\frac{1}{2l_P^2} \sum_{j=1}^n ((\vec{x}_j-\vec{y}-\frac{1}{n}\sum_{j=1}^n \vec{x}_j)^2) \right)\right)\in\Mlt{3n+3},
$$
(augmented) action of $:\prod_{j=1}^n \uwhat{\widetilde{\phi}}^R_0:\in\LL(\DSS{}{3n},\DSMlt{}{3})$, multiplication by  $$
\left((t,\arrs{t}{n})\mapsto \exp \left( -\frac{1}{2l_P^2} \sum_{j=1}^n ((t_j-t)^2) \right)s(t)\right) \in\SRR{n+1},
$$
and application of a numerical distribution to an element of $\DSS{}{n+3}$.
\par 
To see that the interaction is admissible, proceed as in the previous example.
\end{proof}
\begin{exmp}[Not an example: the star product]
Very popular in physics literature (see e.g. \cite{NekrasovNCQFT}) and initially suggested for the DFR spacetime \cite{DFR}, interaction terms, based on the replacement of the pointwise product with a non-commutative product is \emph{not} admissible. Indeed, it would require
$$
\kappa_{*}(\arrs{\tau}{n}\arrs{\vec{r}}{n})\sim  \exp\left(-2\ii \sum_{i<j}(\tau_i,\vec{r}_i)^{\mu}Q_{\mu\nu}(\tau_j,\vec{r}_j)^{\nu} \right),
$$
where $(\tau_i,\vec{r}_i)^{\mu}$ stands for the $\mu$th component of the four-dimensional vector $(\tau_i,\vec{r}_i)$, $Q_{\mu\nu}$ is the commutator of the coordinates (we use the notation of \cite{DFR}) and the Einstein summation rule is assumed. A notable property of $\kappa_{*}$ is that it does not decay at large displacements $\vec{r}_j$ and for this reason can not lead to smooth kernels $F_{l,l'}$ (note that in the examples above, the interaction kernels are essentialy Fourier transforms of $\kappa$). Physically it means that the interaction does not become negligible at large distances. For this reason the adiabatic cut-off in Definition \ref{def:HpQFT/HI} is not enough to regularize the infrared divergences of the theory. In fact, in \cite{BahnsWick} it was already noted that the infrared divergences are present already before the adiabatic limit is taken.  For rigorous treatment one should introduce further $\vec{r}_i$ adiabatic cut-off functions already in $\kappa$. The adiabatic limit of such theories is more involved and is not considered here.
\end{exmp}
To conclude, we see that Definition \ref{def:HpQFT/AdmInt} allows translational-invariant interaction of type (\ref{eq:HpQFT/AdmInt/hI-kappa}) with $\kappa$ fast decaying and having restricted singularities. 

\subsubsection{Admissible Hamiltonian}
\begin{Def}\label{def:HpQFT/HI}
Let $h_{I}$ be a fixed admissible interaction. Then for each $\tlambda\in\SRR{4}$ we define the \emph{interaction Hamiltonian} $H_I^{[g\tlambda]}[\Psi]\in\LLfirst(\DS,\DSS{}{1})$ by setting 
$$ H_I^{[g\tlambda]}=g\widehat{\DiagM{1}}\widehat{\circ{\extD{4}{1}}}(h_I\circ \widehat{\Fourier{3}{+}})[\Psi\otimes \tlambda],\forall \Psi\in\DS.$$

\end{Def}
\begin{rmk}
The symbolic form of (\ref{def:HpQFT/HI}) is
$$
\underline{H_I^{[g\tlambda]}}(t)=\int
\lambda(t,\vec{x})\underline{h_I}\symbDistArgs{\vec{x}}{t}d^3\vec{x}, t\in\RR{},
$$
where
$$
\lambda(t,\vec{x})=\int{\tlambda(\vec{k},t)e^{i\vec{k}\cdot\vec{x}}}d^{3}\vec{x}.
$$
By direct computation
$$
\underline{H_I^{[g\tlambda]\dag}}(t)=\underline{H_I^{[g\tlambda]}}(t)
$$
whenever
$$
\lambda(t,\vec{x})\in\RR{}, \quad \forall t\in\RR{}, \forall \vec{x}\in\RR{3}.
$$
Naturally, we interpret $H_I$ as the interaction representation of the interaction part of the Hamiltonian and $\lambda$ as the adiabatic cut-off function.  
\end{rmk}

\subsubsection{Evolution and scattering operators}
\begin{prop}\label{prop:HpQFT/US}
Let $h_I$ be fixed interaction density, $\tlambda\in\SRR{4}$ and $H_I^{[g\tlambda]}$ be the corresponding interaction Hamiltonian (see Definition \ref{def:HpQFT/HI}). Then There is unique $U^{[g\tlambda]}\in\LLfirst(\DS{},\DSMlt{}{2})\formalPS{g}$ such that
\begin{equation}
\uwhat{U}^{[g\tlambda]}(t,t)=\mathbb{1}_{\LDS},
\label{eq:prop:HpQFT/US/U(0)}
\end{equation}
\begin{equation}
    \label{eq:prop:HpQFT/US/Urec-dif}
\partial_{t_2}\uwhat{U}^{[g\tlambda]}(t_2,t_1)=-i\uwhat{H_{I}}^{[g\tlambda]}(t_2)\uwhat{U}^{[g\tlambda]}(t_2,t_1).    
\end{equation}
The operator $U^{[g\tlambda]}$ can be found recurrently from
\begin{equation}
\uwhat{U}^{[g\tlambda]}(t_2,t_1)=\mathbb{1}_{\LDS}-i\int_{t_1}^{t_2}\uwhat{H_I}^{[g\tlambda]}(t)\uwhat{U}^{[g\tlambda]}(t,t_1)dt,
\label{eq:prop:HpQFT/US/Urec}    
\end{equation}

or explicitly
\begin{equation}\label{eq:prop:HpQFT/US/Uexpl}
\uwhat{U}^{[g\tlambda]}(t,t_0)=    
\end{equation}
$$
\sum_{n=0}^{\infty}(-i)^n
\int_{t_0}^{t}dt_1\int_{t_0}^{t_1}dt_2\cdots\int_{t_0}^{t_{n-1}}dt_n
\prod_{j=1}^{n}\uwhat{H_I}^{[g\tlambda]}(t_i).
$$
\end{prop}
\begin{proof}
To prove the theorem we need only to decipher the symbolic notation of Subsection \ref{Framework/Symbolic} (generalized in Notation \ref{notn:DS/symb-opval}) and use the Wick Theorem of Remark \ref{rmk:DS/2Qgen/Wick} to present the result as a second quantization. For completeness we present all the translations from the symbolic language.
\par 
First, we rewrite (\ref{eq:prop:HpQFT/US/Urec}) explicitly\footnote{We use the integration operator $J$ defined in Subsection \ref{Framework/exmp}.}:
\begin{equation}
U^{[g\tlambda]}=\mathbb{1}_{\Mlt{2}}-i  \left(\extDH{3}{2}\left(\widehat{J}\circ\widehat{\DiagM{1}}\right)\right)\circ \extDH{1,2}{2,3}H_I^{[g\tlambda]}\otimes_{\bullet} U^{[g\tlambda]},
\label{eq:prop:proof:HpQFT/US/Urec-expl}
\end{equation}
where 
$$\mathbb{1}_{\Mlt{2}}\in\Mlt{2}=\LL(\DS^{0},\DSMlt{0}{2}),$$ 
$$
\mathbb{1}_{\Mlt{2}}(t,t')=1, \,\forall t,t'\in\RR{}.
$$
Since $H_I^{[g\tlambda]}$ is proportional to $g$,  (\ref{eq:prop:proof:HpQFT/US/Urec-expl}) gives a well-defined expression of higher orders of $U^{g\tlambda}$ via its lower orders,  and the zeroth order is given by
$$
U^{[g\tlambda]}_{g=0}=\mathbb{1}_{\Mlt{2}}.
$$
Thus (\ref{eq:prop:proof:HpQFT/US/Urec-expl}) can be solved order by order. 
By (\ref{eq:Framework/exmp/Int-dif}) we see that (\ref{eq:prop:HpQFT/US/U(0)}-\ref{eq:prop:HpQFT/US/Urec}) is equivalent to (\ref{eq:prop:HpQFT/US/Urec}).
\par 
Finally, we secondary quantize and formally iterate (\ref{eq:prop:proof:HpQFT/US/Urec-expl}) to get
$$
\widehat{U}^{g\tlambda}=\sum_{i=0}^n(-i)^n \left( \left(\extDH{3}{2}\left(J\circ\widehat{\DiagM{1}}\right)\right)\circ \extDH{1,2}{2,3}\widehat{H_I}^{[g\tlambda]}\right)^{\circ n}\circ \mathbb{1}_{\Mlt{2}},
$$
%TBD . verify
which is one of the possible forms of (\ref{eq:prop:HpQFT/US/Uexpl}). 
\end{proof}
\begin{=>}\label{=>:HpQFT/Uprop}
For $U^{[g\tlambda]}$ as in proposition above, assuming in addition that $\extD{1}{1}\Fourier{3}{+}\tlambda$ is real-valued,
\begin{enumerate}
    \item For any $t,t_0\in\RR{}$
    \begin{equation}\label{eq:=>:HpQFT/Uprop/Udag}
        \uwhat{U}^{[g\tlambda]}(t,t_0)^{\dag}=\uwhat{U}^{[g\tlambda]}(t_0,t);
    \end{equation}
    \item For any $t_0,t_1,t_2\in\RR{}$:
    \begin{equation}\label{eq:=>:HpQFT/Uprop/Ucompose}
        \uwhat{U}^{[g\tlambda]}(t_2,t_1)\uwhat{U}^{[g\tlambda]}(t_1,t_0)=\uwhat{U}^{[g\tlambda]}(t_2,t_0)
    \end{equation}
    \item There are well-defined limits $U_{\mathrm{R}},U_{\mathrm{A}}\in\LLfirst(\DS{},\DSS{}{1})$ and $S\in\LLfirst(\DS)$ defined by
    $$\uwhat{U_{\mathrm{R}}}^{[g\tlambda]}(t)=\lim_{t_0\rightarrow -\infty}\uwhat{U}^{[g\tlambda]}(t,t_0)$$
    $$\uwhat{U_{\mathrm{A}}}^{[g\tlambda]}(t)=\lim_{t_0\rightarrow +\infty}\uwhat{U}^{[g\tlambda]}(t_0,t)$$
    $$\uwhat{S}^{[g\tlambda]}=\lim_{t\rightarrow +\infty,t_0 \rightarrow -\infty }\uwhat{U}^{[g\tlambda]}(t_0,t);$$
    Moreover, relations (\ref{eq:=>:HpQFT/Uprop/Udag}-\ref{eq:=>:HpQFT/Uprop/Ucompose}) remain correct then one or more of the timestamps goes to $\pm\infty$.
\end{enumerate}
\end{=>}
\begin{proof}
For simplicity we use the symbolic notation. 
\par 
For the first statement follows from (\ref{eq:prop:HpQFT/US/Uexpl}) and Remark \ref{rmk:DS/LDS/conj-product-gen}.
\par 
For the second statement it is more convenient to use the recurrent relation (\ref{eq:prop:HpQFT/US/Urec}) and proceed inductively. The zeroth order is trivial. Assuming that (\ref{eq:=>:HpQFT/Uprop/Ucompose}) is valid up to $g^n$ we have
$$\left(
\uwhat{U}^{[g\tlambda]}(t_2,t_1)\uwhat{U}^{[g\tlambda]}(t_1,t_0)\right)_{g^n} =
$$
$$
\left(\mathbb{1}_{\LDS}-i\int_{t_1}^{t_2}\uwhat{H_I}^{[g\tlambda]}(t)\uwhat{U}^{[g\tlambda]}(t,t_1)\uwhat{U}^{[g\tlambda]}(t_1,t_0)dt\right)_{g^n}=
$$
$$
\left(
\mathbb{1}_{\LDS}-i\int_{t_1}^{t_2}\uwhat{H_I}^{[g\tlambda]}(t)\uwhat{U}^{[g\tlambda]}(t,t_0)\right)_{g^n}dt=\left(
\uwhat{U}^{[g\tlambda]}(t_2,t_0)\right)_{g^n},
$$
where we have used the fact that $H_I$ is proportional to $g$, so we can use the statement (\ref{eq:=>:HpQFT/Uprop/Ucompose}) for lower orders. 
\par 
Finally, the last statement follows from  (\ref{Framework/exmp/IntLim}).
\end{proof}
\begin{rmk}
In particular, we see that $\uwhat{U}^{[g\tlambda]}$,$\uwhat{U_{\mathrm{A}}}^{[g\tlambda]}$, $\uwhat{U_{\mathrm{R}}}^{[g\tlambda]}$ and $S^{[g\tlambda]}$ are valued in the "unitary" operators in the sense that
$$
\uwhat{U}^{[g\tlambda]}(t,t')\uwhat{U}^{[g\tlambda]}(t,t')^{\dag}=\uwhat{U}^{[g\tlambda]}(t,t')^{\dag}\uwhat{U}^{[g\tlambda]}(t,t')=\mathbb{1}_{\LDS}.
$$
\end{rmk}
From the above we see that $\uwhat{U}^{[g\tlambda]}_{\mathrm{R}}$ is an analogue of $U$ from the previous subsection. Similarly,  $\uwhat{S}$ may be interpreted as the scattering operator.
\subsubsection{Interacting fields}
To finish the perturbation quantum field theory construction we need to define the interacting quantum field and its momentum similarly to (\ref{eq:HpQFT/AdmTh/IntPhi}-\ref{eq:HpQFT/AdmTh/IntPi})
\begin{prop}\label{prop:HpQFT/QF}
There are unique operator-valued distributions $\phi^R,\pi^R\in\LLfirst(\DSS{}{3},\DSS{}{1})$ such that
$$
\uwhat{\phi}^{\mathrm{R}[g\tlambda]}\symbDistArgs{\vec{x}}{t}=\uwhat{U_{\mathrm{R}}}^{[g\tlambda]}(t)\uwhat{\phi^{\mathrm{R}}_0}\symbDistArgs{\vec{x}}{t}\uwhat{U_{\mathrm{R}}}^{[g\tlambda]}(t)^{\dag},
$$
$$
\uwhat{\pi}^{\mathrm{R}[g\tlambda]}\symbDistArgs{\vec{x}}{t}=\uwhat{U_{\mathrm{R}}}^{[g\tlambda]}(t)\uwhat{\pi^{\mathrm{R}}_0}\symbDistArgs{\vec{x}}{t}\uwhat{U_{\mathrm{R}}}^{[g\tlambda]}(t)^{\dag}.
$$
They satisfy
$$
[\uwhat{\phi}^{\mathrm{R}[g\tlambda]}\symbDistArgs{\vec{x}}{t},\uwhat{\phi}^{\mathrm{R}[g\tlambda]}\symbDistArgs{\vec{x}'}{t}]=0,
$$
$$
[\uwhat{\pi}^{\mathrm{R}[g\tlambda]}\symbDistArgs{\vec{x}}{t},\uwhat{\pi}^{[\mathrm{R}g\tlambda]}\symbDistArgs{\vec{x}'}{t}]=0,
$$
$$
[\uwhat{\pi}^{\mathrm{R}[g\tlambda]}\symbDistArgs{\vec{x}}{t},\uwhat{\phi}^{\mathrm{R}[g\tlambda]}(\vec{x}'|t)]=\ii(2\pi)^3\delta^{(3)}(\vec{x}-\vec{x}')\mathbb{1}_{\LDS},
$$
$$
\partial_{t}\uwhat{\phi}^{\mathrm{R}[g\tlambda]}\symbDistArgs{\vec{x}}{t}=-\ii[\uwhat{H_I}^{[g\tlambda]}(t),\uwhat{\phi}^{\mathrm{R}[g\tlambda]}\symbDistArgs{\vec{x}}{t}]
$$
$$
\partial_{t}\uwhat{\pi}^{\mathrm{R}[g\tlambda]}\symbDistArgs{\vec{x}}{t}=-\ii[\uwhat{H_I}^{[g\tlambda]}(t),\uwhat{\pi}^{\mathrm{R}[g\tlambda]}\symbDistArgs{\vec{x}}{t}].
$$
\end{prop}
\begin{proof}
Direct application of Proposition \ref{prop:HpQFT/US}, Corollary \ref{=>:HpQFT/Uprop} and translations of the symbolic expressions.% We note that
%$$\phi^{\mathrm{R}[g\tlambda]}=\widehat{\DiagM{1}}\circ \extDH{1,2,3}{1,2,3}
%%U_{\mathrm{R}}
%^{[g\tlambda]}
%\otimes_{\bullet}
%\widehat{\DiagM{1}}\circ\widehat{\extD{4}{2}}\phi^{\mathrm{R}[g\tlambda]}\otimes_{\bullet} \extDH{1,2,3}{1,2,3}U_{\mathrm{R}}^{\dag[g\tlambda]}.$$
\end{proof}
\par 
For future use, we introduce several alternative forms of the quantum field and its conjugated momentum.
\par 
First we define the unrestricted form,
$$
\phi^{[g\tlambda]}=\intD{1}{1}\phi^{\mathrm{R}[g\tlambda]}
$$
$$
\pi^{[g\tlambda]}=\intD{1}{1}\pi^{\mathrm{R}[g\tlambda]}
$$
to define the unrestricted versions. 
\par 
Instead of $\phi$ and $\pi$ it is more convenient to work with
$$
\uwhat{\phi_{\pm}}^{\mathrm{R}[g\tlambda]}\symbDistArgs{\vec{x}}{t}=\uwhat{U_{\mathrm{R}}}^{[g\tlambda]}(t)\uwhat{\phi^{\mathrm{R}}_{0,\pm}}\symbDistArgs{\vec{x}}{t}\uwhat{U_{\mathrm{R}}}^{[g\tlambda]}(t)^{\dag}.
$$
Since $\phi^{\mathrm{R}}_{0,\pm}$ is a linear combination of $\phi^{\mathrm{R}}_{0}$ and $\pi^{\mathrm{R}}_{0}$, these objects are well-defined already by Proposition \ref{prop:HpQFT/QF}.
\begin{rmk}
The decomposition of the interacting field into 
$$\widetilde{\phi}=\widetilde{\phi}_{+}+\widetilde{\phi}_{-}$$
have no physical sense. Neither do they represent positive and negative energy parts, no they annihilate the vacuum acting from the right and left. It is introduced for the convenience of formulations only.
\end{rmk}
\subsubsection{Correlators}
Our next goal is to give a precise form to (\ref{eq:HpQFT/WightmanDef}-\ref{eq:HpQFT/GreenUnNormDef}). We are mostly interested in the Green functions due to their role in physical computations of the scattering amplitudes, but for technical reasons it is more convenient to work with the Wightman function because they can be restricted to fixed values of the timestamps. It is also convenient to work with the correlators of the partial fields $\widetilde{\phi}_{\alpha}$ rather than with the field itself.
So, we start by defining
$$
\PWightmanRnag{n}{\alpha}{g\tlambda}
\in \LL(\SRR{3n},\Mlt{n})\formalPS{g},
$$
$$
\PWightmanRnag{n}{\alpha}{g\tlambda}\symbDistArgs{\arrs{\vec{p}}{n}}{\arrs{t}{n}}=\left(\Omega_0,S\prod_{j=1}^n \widetilde{\phi}^{R}_{\alpha}\symbDistArgs{\vec{p}_j} {t_j}\Omega_0\right).
$$
Then we normalize it by introducing
\begin{equation}
\WightmanRnag{n}{\alpha}{g\tlambda}\symbDistArgs{\arrs{\vec{p}}{n}}{\arrs{t}{n}}=\frac{\PWightmanRnag{n}{\alpha}{g\tlambda}\symbDistArgs{\arrs{\vec{p}}{n}  }{\arrs{t}{n}}}{\PWightmanRng{0}{g\tlambda}}.
\label{eq:HpQFT/Wightman-norm}    
\end{equation}
Finally, we define the unrestricted
$$
\WightmanUnag{n}{\alpha}{g\tlambda}=\intD{1,\ldots,n}{1,5,\ldots,4n-3}\WightmanRnag{n}{\alpha}{g\tlambda}
$$
and time-ordered 
\begin{equation}\label{eq:HpQFT/GreenUnag}
\GreenUnag{n}{\alpha}{g\tlambda}((t_j,\vec{p}_j)_{j=1,\ldots,n})=
\end{equation}
$$\sum_{\sigma\in\symmgr{n}}\left(\prod_{j=1}^{n-1}\theta(t_{\sigma_{j}}-t_{\sigma_{j+1}})\right)\WightmanRnagP{n}{\alpha}{g\tlambda}{\sigma}\symbDistArgs{\arrsP{\vec{p}}{n}{\sigma}}{\arrsP{t}{n}{\sigma}}.
$$
versions.

\begin{rmk}
For completeness, we present a form of (\ref{eq:HpQFT/GreenUnag}) not involving the symbolic notation. 
$$
\GreenUnag{n}{\alpha}{g\tlambda}((t_j,\vec{p}_j)_{j=1,\ldots,n})=$$
$$\sum_{\sigma\in\symmze{4}{}}L_{\permK{\sigma}{4}}\circ 
I_{\RR{}}\circ\extD{n}{1}\theta^{\otimes n}
\circ L_{K_n^{-1}} \circ {\DiagM{n}} \circ\extD{1,5,\ldots,4n-3}{n+1,n+2,\ldots,2n}\WightmanRnagP{n}{\alpha}{g\tlambda}{\sigma},
$$
where
$$
K_n: \RR{n}\rightarrow \RR{n}
$$
is
$$
K_n(\arrs{t}{n})=\left((t_j-t_{j+1)_{j=1\ldots n-1}},t_n\right).
$$
\end{rmk}

Putting it all back together one may verify that 
$$
\GreenUnag{n}{\alpha}{g\tlambda}((t_j,\vec{p}_j)_{j=1\ldots n})=\sum_{\sigma\in\symmgr{n}}\prod_{j=1}^{n-1}\theta(t_{\sigma_{j}}-t_{\sigma_{j+1}}) \frac{(\Omega_0,S\prod_{j=1}^n \widetilde{\phi}^{R}_{\alpha_{\sigma_j}}(t_{\sigma_j},\vec{p}_{\sigma_j})\Omega_0) }{(\Omega_0,S\Omega_0)}=
$$
$$
\frac{\left(\Omega_0,S\timeorder{\prod_{j=1}^n \widetilde{\phi}^{R}_{\alpha_{\sigma_j}}(t_{\sigma_j},\vec{p}_{\sigma_j})}\Omega_0\right) }{(\Omega_0,S\Omega_0)}.
$$
The original Wightman and Green functions (\ref{eq:HpQFT/AdmTh/WightDef}-\ref{eq:HpQFT/AdmTh/GreenDef}) are
$$
\WightmanUng{n}{g\tlambda}((t_j,\vec{p}_j)_{j=1\ldots n})=\sum_{\alpha_j\in\{+,-\},j=1,\ldots,n}\WightmanUnag{n}{\alpha}{g\tlambda}((t_j,\vec{p}_j)_{j=1\ldots n}).
$$

$$
\GreenUng{n}{g\tlambda}((t_j,\vec{p}_j)_{j=1\ldots n})=\sum_{\alpha_j\in\{+,-\},j=1,\ldots,n}\GreenUnag{n}{\alpha}{g\tlambda}((t_j,\vec{p}_j)_{j=1\ldots n}).
$$
\subsubsection{Summary and outline of the subsection}

Comparing Proposition \ref{prop:HpQFT/US}, Corollary \ref{=>:HpQFT/Uprop} and Proposition \ref{prop:HpQFT/QF} with the exposition in Subsection \ref{HpQFT/AdmTh} we conclude that we have constructed the non-local Hamiltonian quantum field theory. So it is natural to fix the following.
\begin{Def}
By a massive \emph{admissible Hamiltonian perturbation Quantum Field} (\emph{admissible HpQFT}) Theory we mean a pair $(\omega_0,h_I)$, where $\omega_0\in\SRR{3}$ is a massive dispersion relation\footnote{See Definition \ref{def:dispRel}} and $h_I\in\LLfirst(\SRR{3},\Mlt{1})$ is an admissible interaction density. 
\end{Def}
All admissible HpQFT are non-local in the physical sense (see Remark \ref{rmk:Intro/TwoLocalities} by construction.
The algebraic locality is broken in general. Still, some residual locality persists as explained in the following three remarks. 
\begin{rmk}\label{rmk:HpQFT/loc-phys}
The Hamiltonian can always be presented in the time-localized form (\ref{eq:HpQFT/AdmTh/hI-timeloc-pos}). So, from the physical point of view, the interaction is always "local in time". In view of this, the main advantage of the Hamiltonian approach with respect to the Lagrangian approaches\footnote{Recall that the Lagrangian approaches fail to construct a unitary scattering operator when theory is not local in time \cite{NekrasovNCQFT}} seems to be lost. In fact, as was shown in \cite{PhDThesis}, HpQFT is equivalent to a Lagrangian non-local perturbation theory. The corresponding Lagrangian is the Legendre transform of the Hamiltonian with the time-localized interaction (\ref{eq:HpQFT/AdmTh/hI-timeloc-pos})
\footnote{Note that (\ref{eq:HpQFT/AdmTh/hI-timeloc-pos}) can be rewritten in terms of $\phi_0$ and $\pi_0$ instead of $\phi_{0\pm}$.}. The Lagrangian theories of this form, however, are more complicated and, in particular, require regularization of the UV  divergences despite the spatial non-locality. 
\end{rmk}
\begin{rmk}\label{rmk:HpQFT/loc-alg}
From the algebraic (AQFT) point of view, the commutation relations (\ref{eq:HpQFT/AdmTh/canonicalQuant}) make the theory look like a local one if the time is fixed. Unlike the physical locality in the previous remark, this property is inherent to the theory: if the quantum fields $\phi$ are related to local measurements in any sense, (\ref{eq:HpQFT/AdmTh/canonicalQuant}) has measurable consequences. In particular, it selects a special reference frame, drastically breaking the Lorentz invariance. It is worth noting, that this issue can be bypassed by giving up the direct interpretation of the interacting field. For example, in \cite{HamQFT} it was shown that the Lorentz-covariant approach based on the Yang-Feldman equation can be formulated in the Hamiltonian style by appropriate reference frame-dependent redefinition of the fields. 
\end{rmk}
One may note that these two residual localities are in some sense complementary.
\begin{rmk}
Looking with attention at (\ref{eq:prop:HpQFT/US/Uexpl}) and (\ref{eq:=>:HpQFT/Uprop/Ucompose}), one may conclude that all the admissible HpQFT possess some residual causality property, similar to the causality of local theories \cite{EG73,pAQFT}. Formalization of this property is postponed to \cite{PAP1} where it plays an important role. For now we underline that this property also holds in one particular reference frame only. 
\end{rmk}
\begin{rmk}\label{rmk:HpQFT/higher-g}
We have assumed that the interaction Hamiltonian has precisely the first order in terms of the interaction constant $g$. However, up to some technical change, all construction remains valid for more general Hamiltonians, presented by formal power series in terms of $g$, provided that the zeroth order vanishes.  This is necessary to allow the (finite) counterterms for renormalization. We ignore this detail for simplicity.
\end{rmk}
\begin{rmk}\label{rmk:HpQFT/Hint}
It may look unnatural to assume that the interacting Hamiltonian is fixed already in the interaction picture. In particular, this is the cause of the residual locality observed in Remark \ref{rmk:HpQFT/loc-phys}. In \cite{PhDThesis} an alternative way to construct non-local Hamiltonian theories, based on setting by analogy with (\ref{eq:HpQFT/AdmTh/HI})
$$
\uwhat{h_I}\symbDistArgs{\vec{x}}{t}=\uwhat{U_{\rm{R}}}^{[g\tlambda]}(t)\uwhat{h_{int}}\symbDistArgs{\vec{x}}{t}\uwhat{U_{\rm{R}}}^{\dag[g\tlambda]}(t),
$$
where
$$
\uwhat{h_{int}}\symbDistArgs{\vec{x}}{t}=
$$
$$\int A\left(\arrs{\tau}{n},\frac{\sum_{j=1}^n\vec{x}_j}{n}-\vec{x}\right){\mathpunct{:} }\prod_{j=1}^n\underline{\widehat{\phi}}(t+\tau_j,\vec{x}_j){\mathpunct{:}} \times
$$
$$
B\left(\left(\vec{x}_j-\vec{x}\right)_{j=1,\ldots,n},\arrs{\tau}{n}\right) d^3\vec{x}d^{3n}\arrs{\vec{r}}{n}d^{n}\arrs{\tau}{n},
$$
and $U^{[g\tlambda]}$, $\phi$ are defined as in Proposition \ref{prop:HpQFT/QF}. It is easy to see that this operation is well-defined in the sense of formal series\footnote{Note that in the leading zeroth with respect to $g$ the interacting quantum field can be replaced by the free one, and the higher orders are always expressed via lower ones.}, Propositions \ref{prop:HpQFT/US}, \ref{prop:HpQFT/QF} and Corollary \ref{=>:HpQFT/Uprop} still hold, and
$$
\uwhat{h_{I}}   \symbDistArgs{\vec{x}}{t}=
$$
$$\int A\left(\arrs{\tau}{n},\frac{\sum_{j=1}^n\vec{x}_j}{n}-\vec{x}\right){\mathpunct{:} }\prod_{j=1}^n\uwhat{U}^{[g\tlambda]}(t,t+\tau_j)\underline{\widehat{\phi}_0}(t+\tau_j,\vec{x}_j)\uwhat{U}^{[g\tlambda]}(t+\tau_j,t){\mathpunct{:}} \times
$$
$$
B\left(\left(\vec{x}_j-\vec{x}\right)_{j=1,\ldots,n},\arrs{\tau}{n}\right) d^3\vec{x}d^{3n}\arrs{\vec{r}}{n}d^{n}\arrs{\tau}{n}.
$$
With minimal technical effort, one may show that the correctly written analogue of the equations of motion (\ref{eq:HpQFT/AdmTh/HamEv}) holds. At the same time one may note that the class of effective theories defined in this way is essentially the same.  
\end{rmk}
\subsection{Feynman rules with adiabatic cut-off}
\label{HpQFT/FR}

In this subsection we present a technical version of the Feynman rules for computation of the restricted Wightman function in the presence of the adiabatic cut-off which we use in the proof of Theorem \ref{thm:MainW}. More practical Feynman rules are presented in Appendix. We assume a massive admissible HpQFT to be fixed.
\par 
\subsubsection{Preparation}
We start by the following observation based on construction of $\PWightmanRnag{n}{\alpha}{g\tlambda}$ and $U^{[g\tlambda]}$.
\begin{rmk}\label{rmk:HpQFT/V-expansion}
For any $n\in\NN_0$, $\tlambda\in\SRR{4}$ and $f\in\SRR{3n}$
\begin{equation}\label{eq:HpQFT/FR/V-expansion}
\PWightmanRnag{n}{\alpha}{g\tlambda}[f]=
\sum_{V=0}^{\infty}g^V \PWightmanRnaV{n}{\alpha}{V}[f\otimes\tlambda^{\otimes V}],    
\end{equation}
where $\PWightmanRnaV{n}{\alpha}{V}[f\otimes\tlambda^{\otimes V}]\in\LL(\SRR{3n+4V},\SRR{n})$ is given by
$$
\PWightmanRnaV{n}{\alpha}{V}=\sum_{\substack{v_j\in\NN_0, j=0,\ldots, n\\ v_0+\ldots+v_n=V}}
\PWightmanRnaVarr{n}{\alpha}{v}
$$
and $\PWightmanRnaVarr{n}{\alpha}{v}\in \LL(\SRR{3n+4(\sum_{j=0}^{n}v_j)},\SRR{n})$ is defined by
$$
\PWightmanRnaVarr{n}{\alpha}{v}\symbDistArgs{\arrs{\vec{k}}{n},(\tau_{i,j},\vec{q}_{i,j})_{i=0\ldots n, j=1,\ldots, v_i}}{\arrs{t}{n}}=
$$
$$
\intUSimplex{+}{v_0}\symbDistArgs{(\tau_{0,j})_{j=1,\ldots,v_0}}{t_1}\left(\prod_{i=1}^{n-1}\intSimplex{v_i}\symbDistArgs{(\tau_{i,j})_{j=1,\ldots,v_i}}{t_{i+1},t_i}\right)\intUSimplex{-}{v_n}\symbDistArgs{(\tau_{n,j})_{j=1,\ldots,v_n}}{t_n}
$$
$$
(\Omega_0,\prod_{j=1}^{v_0}\uwhat{\widetilde{h}}_{I}\symbDistArgs{\vec{k}_{0,j}}{\tau_{0,j}} \prod_{i=1}^{n}\uwhat{\phi^{R}_{0\alpha_i}}\symbDistArgs{\vec{k}_i}{\vec{t}_i}\prod_{j=1}^{v_i}\uwhat{\widetilde{h}}_{I}\symbDistArgs{\vec{q}_{i,j}}{\tau_{i,j}}  \Omega_0).
$$
\end{rmk}
\begin{rmk}\label{rmk:HpQFT/FR/combFactor}
Here we use the symbolic presentation of the integration operators. For eventual convenience we present the corresponding symbols as discontinuous functions\footnote{Here we identify the Heaviside function with the corresponding discontinuous function. This is safe as long as no derivatives are involved (see also Remark \ref{rmk:Framework/symb/disc-restriction}).}
$$
\intSimplex{v}\symbDistArgs{\arrs{\tau}{v}}{t,t'}=(-1)^{v\sign (t'-t)}\theta((\tau_n-t)\sign(t'-t))\theta((t'-\tau_1)\sign(t'-t))\times
$$
$$
\prod_{j=1}^{v-1}\theta(\sign (t'-t)(\tau_j-\tau_{j+1}))=
T_{1,\ldots,v;<}^{\arrs{\tau}{v};-\sign (t'-t)} \prod_{j=1}^{v} J\symbDistArgs{\tau_j}{t,t'},
$$
where $T_{1,\ldots,v;<}^{\arrs{\tau}{v};+}\in\{0,1\}$ is defined below.  In the same way, we have
$$
\intUSimplex{\pm 1}{v}\symbDistArgs{\arrs{\tau}{v}}{t}=T_{1,\ldots,v;<}^{\arrs{\tau}{v};-1}J_{\pm}\symbDistArgs{\tau_j}{t}
$$
\end{rmk}
\begin{notn}
For any partially ordered set $(A,\prec)$, $v\in\NN$, any choice of $a_1,\ldots,a_v\in A$ and $\eta=\pm 1$ we set the \emph{ordering indicator function} $\RR{v}\rightarrow \{0,1\}$
$$
(\arrs{\tau}{v})\mapsto T_{a_1,\ldots,a_v;\prec}^{\arrs{\tau}{v};+1}.
$$
so that $T_{a_1,\ldots,a_v;\prec}^{\arrs{\tau}{v};+1}$ if
\begin{equation}\label{eq:HpQFT/FR/ord-agr}
    a_i\prec a_j \Rightarrow \eta(\tau_i-\tau_j)<0, \, i,j=1,\ldots, v.
\end{equation}
\end{notn}
The combinatoric interpretation of the factors above leads to the following results:
\begin{lem}\label{lem:HpQFT/FR/independent}
Let $(A',\prec')$ and $(A'',\prec'')$ be partially ordered sets and define
$$(A,\prec)=(A',\prec')\sqcup (A'',\prec'').$$
Then for any $n',n''\in\NN_0$, any $\arrs{a'}{n'}\in A'$, $\arrs{a''}{n''}\in A''$, $\arrs{\tau'}{n'},\arrs{\tau''}{n''}\in\RR{}$ and $\eta=\pm$ one has  
$$
T_{\arrs{a'}{n'},\arrs{a''}{n''};\prec}^{\arrs{\tau'}{n'},\arrs{\tau''}{n''};\eta}=T_{\arrs{a'}{n'};\prec'}^{\arrs{\tau'}{n'};\eta}T_{\arrs{a''}{n''};\prec''}^{\arrs{\tau''}{n''};\eta}
$$
\end{lem}

\begin{lem}\label{lem:HpQFT/FR/incompatible}
Let $(A,\prec)$ be a partially ordered set. Let $\overline{\prec}$ be the set of all linear extensions of $\prec$. Take r any , $v\in\NN$,  $a_1,\ldots,a_v\in A$ , $\eta=\pm 1$ and $\arrs{\tau}{v}$ such that $$i\neq j\Rightarrow \tau_{i}\neq \tau_j, \, i,j=1,\ldots,v.$$%%TBD not augmentation!
Then
$$
T_{\arrs{a}{v};\prec}^{\arrs{\tau}{v};\eta}=\sum_{\prec'\in\overline{\prec}}T_{\arrs{a}{v};\prec'}^{\arrs{\tau}{v};\eta}.
$$
\end{lem}
\begin{proof}[Proof of Lemmas \ref{lem:HpQFT/FR/independent} and \ref{lem:HpQFT/FR/incompatible}]
The statements are similar to standard results in combinatorics and elementary probability. In the first case the ordering conditions (\ref{eq:HpQFT/FR/ord-agr}) for the two disjoint subsets are independent, hence the indicator function factors. 
\par 
In the second case, we note that if all $\arrs{\tau}{v}$ are different, then there is precisely one total $\prec$ order satisfying (\ref{eq:HpQFT/FR/ord-agr}). At the same time, if $\prec$ is a partial order, then it satisfies the condition (\ref{eq:HpQFT/FR/ord-agr}) if and only if one of its linear extensions does. Thus Lemma \ref{lem:HpQFT/FR/incompatible} follows by summing over the incompatible possibilities.%TBD not augmenatation
\end{proof}

\subsubsection{Feynman rules for unrenormalized correlators}
The next step is to apply Wick's theorem in the form of Proposition \ref{prop:DS/2Q/Wick} and Remark \ref{rmk:DS/2Qgen/Wick} to get the Feynman rules.
\begin{prop}\label{prop:HpQFT/FR/PW-total}
The distribution
$$\PWightmanRnaVarr{n}{\alpha}{v}\symbDistArgs{\arrs{\vec{k}}{n},(\tau_{i,j},\vec{q}_{i,j})_{i=0\ldots n, j=1,\ldots, v_i}}{\arrs{t}{n}}$$ can be computed by the following Feynman rules\footnote{See Subsection \ref{Intro/prelim} for the terminology.}. 
\begin{itemize}
    \item The relevant are  the Feynman graphs $\Gamma$ with $V=\sum_{j=0}^{n}$ enumerated internal vertices  $\bullet_{(i,j)}$ with $i=0,\ldots,n$ and $j=1,\ldots,v_i$, and $n$ enumerated external vertices $\circ_i$, $i=1,\ldots, n$; 
    \item For convenience we fix total order $\prec_{\Gamma}$ generated by the following relations:
    \begin{equation}\label{eq:HpQFT/FR/ord1}
        \circ_{i+1}\prec_{\Gamma} \circ_{i}, \quad i=1,\ldots n, j=1,\ldots,v_i 
    \end{equation}
    \begin{equation}\label{eq:HpQFT/FR/ord2}
        \bullet_{(i,j)}\prec_{\Gamma} \circ_{i}, \quad i=1,\ldots n, j=1,\ldots,v_i 
    \end{equation}
    \begin{equation}\label{eq:HpQFT/FR/ord3}
        \circ_{i+1}\prec_{\Gamma} \bullet_{(i,j)}, \quad i=1,\ldots n-1, j=1,\ldots,v_i
    \end{equation} 
    \begin{equation}\label{eq:HpQFT/FR/ord4}
        \bullet_{(i+1,j)}\prec_{\Gamma} \bullet_{(i,j)}, \quad i=0,\ldots n, j=1,\ldots,v_i-1.
    \end{equation} 
    \item To each line  assign a free momentum flux $\vec{p}\in\RR{3}$ directed from the earlier to the later end;
    \item
        The factors corresponding to each element are shown in Table \ref{tab:FRules-TM}; 
        \item 
        The overall factor is 
        $$
        C_{\Gamma}=
        \intUSimplex{+}{v_0}\symbDistArgs{(\tau_{0,j})_{j=1,\ldots,v_0}}{t_1}\left(\prod_{i=1}^{n-1}\intSimplex{v_i}\symbDistArgs{(\tau_{i,j})_{j=1,\ldots,v_i}}{t_{i+1},t_i}\right)\times$$
        $$\intUSimplex{-}{v_n}\symbDistArgs{(\tau_{n,j})_{j=1,\ldots,v_n}}{t_n}=
        $$
        $$
        \left(\prod_{j=1}^{v_0}J_{+}\symbDistArgs{\tau_{0,j}}{t_0}\right)
        \left(\prod_{i=1}^{n-1}\prod_{j=1}^{v_i}J\symbDistArgs{\tau_{i,j}}{t_{i+1},t_i}\right)
         \left(\prod_{j=1}^{v_n}J_{-}\symbDistArgs{\tau_{n,j}}{t_n}\right)\times
         $$
        $$
        T_{1,\ldots,v_0;<}^{(\tau_{0,j})_{j=1,\ldots,v_0};-}
         \left(\prod_{i=1}^{n-1}T_{1, \ldots, v_i;<}^{(\tau_{i,j})_{j=1,\ldots,v_i};-\sign (t_{i}-t_{i+1})}\right)
        T_{1,\ldots,v_n;<}^{(\tau_{n,j})_{j=1,\ldots,v_n};-}     %
        $$
\end{itemize}
\end{prop}
\begin{proof}
The proof goes as usual with using the Wick Theorem of Proposition \ref{prop:DS/2Q/Wick} and Remark \ref{rmk:DS/2Qgen/Wick}. Each external vertex is an insertion of $\phi_0$, each internal is an insertion of $h_I$  and each line depicts a contraction as explained in Remark \ref{rmk:DS/2Q/contractions} and the ordering is just the right-to-left ordering of operators in the composition. The internal vertices are labeled with two numbers $(i,j)$
\end{proof}
The ordering (\ref{eq:HpQFT/FR/ord1}-\ref{eq:HpQFT/FR/ord4}) completely fixed by labeling of the vertices is convenient for estimations in the next section. 

\begin{rmk}\label{rmk:HpQFT/FR/symmze}
More standard object would be to deal with symmetrization of $\PWightmanRnaV{n}{\alpha}{V}$, which can be introduced as a formal variational derivative
\begin{equation}
g^{-V}\frac{\delta^W\PWightmanRnag{n}{\alpha}{g\tlambda}{\arrs{\vec{k}}{n}}{\arrs{t}{n}}}{\delta \tlambda(\tau_1,\vec{q}_1)\cdots \delta \tlambda(\tau_V,\vec{q}_V)}=\left(\PWightmanRnaV{n}{\alpha}{V}\circ \extD{1,\ldots,3n}{1,\ldots,3n}\symmze{4}{}\right)\symbDistArgs{\arrs{\vec{k}}{n},(\tau_{i},\vec{q}_{i})_{i=1\ldots n}}{\arrs{t}{n}}.
\label{eq:HpQFT/FR/symmze}    
\end{equation}
As $\symmze{4}{}[\lambda^{\otimes V}]=\lambda^{\otimes V}$, the right-hand sides of (\ref{eq:HpQFT/FR/symmze}) can be safely placed instead of  $\PWightmanRnaV{n}{\alpha}{V}$ in (\ref{eq:HpQFT/FR/V-expansion}). Similarly we can perform partial symmetrization of $\PWightmanRnaVarr{n}{\alpha}{v}$ by introducing $\PWightmanRnaVarrS{n}{\alpha}{v}\in\LL(\SRR{3n+4V}),\SRR{n})$,
\begin{equation}
\PWightmanRnaVarrS{n}{\alpha}{v}=\left(\PWightmanRnaVarr{n}{\alpha}{v}\circ \extD{1,\ldots,3n}{1,\ldots,3n}\bigotimes_{i=0}^{n}\symmze{4}{v_i} \right).
\label{eq:HpQFT/FR/symmzeP}    
\end{equation}
\end{rmk}

\begin{prop}\label{prop:HpQFT/FR/PW-partial}
$\PWightmanRnaVarrS{n}{\alpha}{v}(\arrs{\vec{k}}{n},\symbDistArgs{(\tau_{i,j},\vec{q}_{i,j})_{i=0\ldots n, j=1,\ldots, v_i}}{\arrs{t}{n}}$ can be computed according to the following Feynman rules:
\begin{itemize}
    \item The relevant are all partially ordered Feynman graphs $(\Gamma,\prec_{\Gamma})$ with $V$ internal vertices $\bullet_{j}$, $j=1,\ldots, v$ and $n$ external vertices $\circ_j$, $j=1,\ldots,n$;
    \item   The order is constrained by relations (\ref{eq:HpQFT/FR/ord1}-\ref{eq:HpQFT/FR/ord3});
    \item To each line a momentum flux, directed from earlier to later vertex is assigned;
    \item Factors corresponding to elements of the diagrams are presented in Table \ref{tab:FRules-TM};
        \item The overall factor is
        $$
        C'_{\Gamma,\prec_{\Gamma}}=\left(\prod_{j=1}^{v_0}J_{+}\symbDistArgs{\tau_{0,j}}{t_0}\right)
        \left(\prod_{i=1}^{n-1}\prod_{j=1}^{v_i}J\symbDistArgs{\tau_{i,j}}{t_{i+1},t_i}\right)
         \left(\prod_{j=1}^{v_n}J_{-}\symbDistArgs{\tau_{n,j}}{t_n}\right)\times
         $$
        $$
        T_{(\bullet_{(0,j)
        })_{j=1,\ldots,v_0};\prec_\Gamma}^{(\tau_{(i,0),i=1,\ldots,dj})_{j=1,\ldots,v_0};+} 
         \left(\prod_{i=1}^{n-1}T_{(\bullet_{(i,j)})_{j=1,\ldots,v_i};\prec_{\Gamma}}^{(\tau_{i,j})_{j=1,\ldots,v_i};\sign (t_{i}-t_{i+1})}\right)
        T_{(\bullet_{(n,j)
        })_{j=1,\ldots,v_n};\prec_\Gamma}^{(\tau_{n,j})_{j=1,\ldots,v_n};-}.     %
        $$
\end{itemize}
\end{prop}

\begin{proof}
Directly symmetrizing the Feynman rules of Proposition \ref{prop:HpQFT/FR/PW-total} with respect to permutations of the internal vertices within sets $\{(\bullet_{(i,j)})|j=1,\ldots,v_i\}$ for each $i=0,\ldots,n$ leads to the Feynman rules as above, but with completely ordered graphs. Indeed, the complete ordering is just a way to describe the permutation restoring of the vertices restoring (\ref{eq:HpQFT/FR/ord4}). 
\par 
The factors assigned to elements of the diagrams depend on the relative order of the vertices connected by a line (since the vertex factor distinguishes the incoming and outgoing lines), so we may consider partially ordered graphs as in the statement, but the overall factor will be
$$
\sum_{\prec'\ii \overline{\prec_{\Gamma}}}C'_{\Gamma,\prec'}=C'_{\Gamma,\prec_{\Gamma}},
$$
where we used notation and statement of Lemma \ref{lem:HpQFT/FR/incompatible} for the last step.
\end{proof}

\begin{table}
    \centering
    \begin{tabular}{|m{60pt}|m{80pt}|m{160pt}|}
    \hline 
    \textbf{Element} 
    &
    \textbf{Figure}
    &
    \textbf{Factor}
    \\
    \hline
    Internal vertex
    &

       \begin{tikzpicture}
       \node[] at (-1.2,1.6) 
    {$\vec{p}'_1$};

       \draw[black, thick] (-1.2,1.2)--(0,0);
       \draw[black, thick,<-] (-1.2,1.4)--(-0.8,1.0);
       
           \node[] at (0.5,1.4) 
    {$\vec{p}_1$};

       \draw[black, thick] (1,1)--(0,0);
       \draw[black, thick,->] (1,1.4)--(0.1,0.6);
         \node[] at (0.6,-1.2) 
    {$\vec{p}_l$};

       \draw[black, thick] (1,-1)--(0,0);
       \draw[black, thick,->] (0.9,-1.2)--(0.3,-0.6);
\filldraw[black] (0,0) circle (2pt) node[anchor=west] {$(i,j)$};
 \node[] at (-1,0.2) 
    {\Huge $\vdots$};
    \node[] at (1,0.2) 
    {\Huge $\vdots$};

     \node[] at (-1,-1.2) 
    {$\vec{p}'_{l'}$};

       \draw[black, thick] (-1.5,-1)--(0,0);
       \draw[black, thick,<-] (-1.4,-1.2)--(-0.8,-0.8);
\end{tikzpicture}  & $$F_{(l',l)}(\arrs{\vec{p}'}{l'},\arrs{\vec{p}}{l})\cdot$$
$$
\delta^{(3)}\left(\sum_{k=1}^{l'}\vec{p}'_k-\sum_{k=1}^{l}\vec{p}_k-\vec{q}_{i,j}\right)\cdot
$$
$$
\exp\left(\ii \left(\sum_{j=1}^{l'}\omega_0(\vec{p}'_j)-\sum_{j=1}^{l}\omega_0(\vec{p}_j)\right)\tau_{i,j}\right)
$$
\\
         \hline 
         Internal line
         & 
         \begin{tikzpicture}
        \node[] at (1,0.8) 
    {$\vec{p}$};
     %   \node[] at (0,-0.2) 
    %{$t$};
    %\node[] at (2,-0.2) 
    % {$t'$};
         \filldraw[black] (0,0)  circle (2pt);
         \filldraw[black] (2,0)  circle (2pt);
         \draw[black, thick] (0,0)--(2,0);
         \draw[black, thick,->] (0.5,0.5)--(1.5,0.5);
         \end{tikzpicture}
         &
         $$ 1$$
         \\
         \hline
          External vertex 
         &
        \begin{tikzpicture}
        \node[] at (0.7,0.2) 
    {$\vec{p}$};
        \node[] at (0.8,-1.3) 
    {$i$};
         \draw[black, thick] (0,0)--(1,-1);
         \draw[black, thick,<-] (0.2,0.3)--(.8,-.2);
         \filldraw[black,thick, fill=white] (1,-1)  circle (2pt);
        \node[] at (1.7,-0.4) 
    {or};
          \begin{scope}[shift={(0.5,0)}]
         \node[] at (1.9,0.2)
    {$\vec{p}$};
        \node[] at (1.6,-1.3) 
    {$i$};
         \draw[black, thick] (1.5,-1)--(2.5,0);
         \draw[black, thick,<-] (1.7,-0.3)--(2.3,.2);
         \filldraw[black,thick, fill=white] (1.5,-1)  circle (2pt);
         \end{scope}
        \end{tikzpicture}
        &
        \begin{center}
            $\delta_{\alpha_i,+}\delta^{(3)}(\vec{k}_i-\vec{p})\frac{e^{\ii\omega_0(\vec{p})t_i}}{\sqrt{(2\pi)^32 \omega_0(\vec{p})}}$ or $\delta_{\alpha_i,-}\delta^{(3)}(\vec{k}_i+\vec{p})\frac{e^{-\ii\omega_0(\vec{p})t_i}}{\sqrt{(2\pi)^32 \omega_0(\vec{p})}}$
        \end{center}

         \\
         \hline
    \end{tabular}

    \caption{Ordered Feynman rules with adiabatic cut-off, time-momentum presentation}
    \label{tab:FRules-TM}
\end{table}

\subsubsection{Normalized correlator}

The next step is to identify the denominator of (\ref{eq:HpQFT/Wightman-norm}). For it we note that by Corollary \ref{=>:HpQFT/Uprop} we have
$$
\left(\Omega_0,\uwhat{S}^{[g\tlambda]}\Omega_0\right)=
$$
$$
\left(\Omega_0,\uwhat{U_{\mathrm{R}}}^{[g\tlambda]}(t_1)\prod_{j=1}^{n-1}\uwhat{U}^{[g\tlambda]}(t_j,t_j+1)\uwhat{U_{\mathrm{R}}}^{[g\tlambda]}(t_n)\Omega_0\right),\, \forall n\in\NN, \forall \arrs{t}{n}\in\RR{}.
$$
This leads to the following rather strange version of the Feynman rules.
\begin{prop}\label{prop:HpQFT/FR/denom}
For $n\in\NN_0$ and $\arrs{t}{n}\in\RR{}$ one has
$$
\left(\Omega_0,\uwhat{S}^{[g\tlambda]}\Omega_0\right)=\sum_{\arrsM{v'}{0}{n}\in\NN_0}
g^{\sum_{j=0}^n v_j}
N_{n;\arrs{v'}{n}}[\tlambda^{\otimes(\sum_{j=0}^n v_j)}](\arrs{t}{n}),
$$
where $N_{n;\arrs{v'}{n}}\in\LL\left(\SRR{4(\sum_{j=0}^n v_j)},\SRR{n}\right)$
can be computed according to the Following Feynman rules:
\begin{itemize}
    \item The relevant are all partially ordered Feynman graphs $(\Gamma,\prec_{\Gamma})$ with $V$ internal vertices $\bullet_{(i,j)}$, $i=0,\ldots,n$, $j=1,\ldots, v_i$;
    \item   The order is constrained by relations
    $$
    \bullet_{(i+1,j)}\prec \bullet_{(i,j')}, \quad, i=0,\ldots,n-1,\,j=1,\ldots,v_{i+1},\,j'=1,\ldots,v_i; 
    $$
    \item To each line a momentum flux, directed from earlier to later vertex is assigned;
    \item Factors corresponding to elements of the diagrams are presented in Table \ref{tab:FRules-TM};
        \item The overall factor is
        $$
        C'_{\Gamma,\prec_{\Gamma}}=\left(\prod_{j=1}^{v_0}J_{+}\symbDistArgs{\tau_{0,j}}{t_0}\right)
        \left(\prod_{i=1}^{n-1}\prod_{j=1}^{v_i}J\symbDistArgs{\tau_{i,j}}{t_{i+1},t_i}\right)
         \left(\prod_{j=1}^{v_n}J_{-}\symbDistArgs{\tau_{n,j}}{t_n})\right)\times
         $$
        $$
        T_{(\bullet_{(0,j)
        })_{j=1,\ldots,v_0};\prec_\Gamma}^{(\tau_{(i,0),i=1,\ldots,dj})_{j=1,\ldots,v_0};+} 
         \left(\prod_{i=1}^{n-1}T_{(\bullet_{(i,j)})_{j=1,\ldots,v_i};\prec_{\Gamma}}^{(\tau_{i,j})_{j=1,\ldots,v_i};\sign (t_{i}-t_{i+1})}\right)
        T_{(\bullet_{(n,j)
        })_{j=1,\ldots,v_n};\prec_\Gamma}^{(\tau_{n,j})_{j=1,\ldots,v_n};-}.     %
        $$
\end{itemize}
\end{prop}
\begin{proof}
Repeat all the constructions of Remark \ref{rmk:HpQFT/V-expansion} and Propositions \ref{prop:HpQFT/FR/PW-total}, \ref{prop:HpQFT/FR/PW-partial} without insertions of $\widetilde{\phi}_{0}$ (hence no external vertices in the Feynman graphs).
\end{proof}

\begin{=>}\label{=>:HpQFT/FR/W-total}
For any $n\in\NN_0$, $\tlambda\in\SRR{4}$, $\arrsM{v}{0}{n}$  and $f\in\SRR{3n}$ define $\WightmanRnaVarr{n}{\alpha}{v}\in \LL(\SRR{3n+4(\sum_{j=0}^{n}v_j)},\SRR{n})$ as the distributions computed by the Feynman rules of Proposition \ref{prop:HpQFT/FR/PW-total} but disregarding all graphs with vacuum components.
Then
\begin{equation}
\sum_{\arrsM{v'}{0}{n}}\sum_{\arrsM{v''}{0}{n}\in\NN_0}g^{\sum_{j=0}^{n}(v'_j+v''_j)}
\WightmanRnaVarr{n}{\alpha}{v'}[f\otimes \tlambda^{\otimes (\sum_{j=0}^{n}v'_j)}]\cdot N_{n;\arrsM{v''}{0}{n}}[\tlambda^{\otimes (\sum_{j=0}^{n}v''_j)}]=
\label{eq:=>:HpQFT/FR/W-total}
\end{equation}
$$\sum_{\arrsM{v}{0}{n}}g^{\sum_{j=0}^{n}(v_j)}\PWightmanRnaVarr{n}{\alpha}{v}[f\otimes \tlambda^{\otimes (\sum_{j=0}^{n}v_j)}],\,\forall f\in\SRR{3n}
$$
where $\cdot$ denotes the pointwise multiplication in $\SRR{n}$
\end{=>}
\begin{proof}
Repeating the argument of Remark \ref{rmk:HpQFT/FR/symmze}, we replace $\WightmanRnaVarr{n}{\alpha}{v'}$ with its symmetrization computed by Feynman rules of Proposition \ref{prop:HpQFT/FR/PW-partial} with the graphs with vacuum corrections omitted. Now substitute that symmetrization together with the diagrammatic expansions of Propositions \ref{prop:HpQFT/FR/PW-partial} and \ref{prop:HpQFT/FR/denom} into the left-hand sides of (\ref{eq:=>:HpQFT/FR/W-total}) and expand out the product. Each term of the resulting sum can be bijectively identified with a Feynman graph of Proposition \ref{prop:HpQFT/FR/PW-partial}. Indeed, the factor coming from $N_{n;\arrsM{v''}{0}{n}}$ corresponds to the vacuum components and  $\WightmanRnaVarr{n}{\alpha}{v'}$ to the non-vacuum components. The factors, corresponding to the elements of a graph in all the cases are given by Table \ref{tab:FRules-TM}, so, apart from the overall factors, the contribution of the whole graph is a product of the contributions of the vacuum components and the rest of the graph. To factorize the overall factor we use
Lemma \ref{lem:HpQFT/FR/independent} and the trivial observation that the vacuum components are always disjoint from the rest of the graph\footnote{Recall that the partial orders we consider are always generated by orientation of the edges, so disjoint graph components correspond to disjoint partially ordered subsets.}. Thus the formal sums are equal.
\end{proof}
\begin{=>}
For any $n\in\NN_0$, $\tlambda\in\SRR{4}$ and $f\in\SRR{3n}$
$$\WightmanRnag{n}{\alpha}{g\tlambda}[f]=
\sum_{V=0}^{\infty}g^V \WightmanRnaV{n}{\alpha}{V}[f\otimes\tlambda^{\otimes V}],
$$
where $\WightmanRnaV{n}{\alpha}{V}[f\otimes\tlambda^{\otimes V}]\in\LL(\SRR{3n+4V},\SRR{n})$ is given by
$$
\WightmanRnaV{n}{\alpha}{V}=\sum_{\substack{v_j\in\NN_0, j=0,\ldots, n\\ v_0+\ldots+v_n=V}}
\WightmanRnaVarr{n}{\alpha}{v},
$$
where  $\WightmanRnaVarr{n}{\alpha}{v}$ are as in Corollary \ref{=>:HpQFT/FR/W-total}.
\end{=>}

\begin{notn}
The set of all graphs relevant for computation of $\WightmanRnaVarr{n}{\alpha}{v}$ by Feynman rules of Corollary \ref{=>:HpQFT/FR/W-total} is denoted with 
$\mathfrak{G}_{\arrs{\alpha}{n}}^{\arrsM{v}{0}{n}}$. For $\Gamma\in\mathfrak{G}_{\arrs{\alpha}{n}}^{\arrsM{v}{0}{n}}$ we denote with $\WightmanRnaVarrG{n}{\alpha}{v}{\Gamma}$ the corresponding contribution.
\end{notn}

%%%

%%%

\section{Weak adiabatic limit}
\subsection{Formulation}
In this subsection we formulate the main result of the paper and show how it implies the weak adiabatic limit existence. The rest of the section is devoted to the proof of this statement, presented as a series of technical lemmas.
\par 
Basically, we want to claim that the Green function restricts to a smooth function of the momentum and (off-shell) energy defects, and thus can be evaluated when all the defects are zero. There is a small complication caused by singularities far from the origin in the space of energy-momentum defects. To avoid it, we temporally bound support of the adiabatic cut-off function $\widetilde{\lambda}$ to remove some unwanted singularities. As we work in the position representation for the time coordinate, it can be done by a convolution.  
For $h\in \SRR{1}$ we define
$
\mathcal{M}_{[h]}^{V,n}\in\LL(\SRR{4(V+n)}) 
$
$$
\mathcal{M}_{[h]}^{V,n}[f](\arrs{(\tau,\vec{q})}{V},\arrs{(t,\vec{p})}{n})=\int{\prod_{j=1}^{V}h(\arrs{\tau'}{V})f((\tau_j-\tau'_j,\vec{q}_j)_{j=1\ldots V},\arrs{(t,\vec{p})}{n})}d^V\arrs{\tau'}{V}.
$$
As explained below, in Corollary \ref{=>:MainW} and Remark \ref{rmk:MainW}, this trick does not affect the physical adiabatic limit.
We also set
$$
Q_{\Delta}=\{h\in\SRR{4}| \mathrm{supp} {\Fourier{1}{-}}^{-1}[h] \in [-\Delta,\Delta] \}.
$$
\begin{thm}\label{thm:MainW}
Fix a massive HpQFT with admissible interaction.
Take $n, V\in\mathbb{N}_0$,  $\arrs{\alpha}{n}\in \{+,-\}$, $\arrs{v}{n}\in\NN_0$, $\Gamma\in \mathfrak{G}_{\arrs{\alpha}{n}}^{\arrsM{v}{0}{n}}$ and $h\in Q_{\frac{M}{V+1}}$.
Then the distribution $\WightmanRnaVarrG{n}{\alpha}{v}{\Gamma} \circ \mathcal{M}_{[h]}^{V,n}$
restricts to $$\restrictH{3n+1,\ldots,3n+4V}{1,\ldots,4V}\left(\WightmanRnaVarrG{n}{\alpha}{v}{\Gamma}\circ \mathcal{M}_{[h]}^{V,n}\right)\in\LL(\SRR{3n},\SMlt{4V}{n}).$$
\end{thm}

\begin{=>}\label{=>:MainW}
Fix a massive admissible HpQFT.
Take $\tlambda\in\SRR{4}$ such that
$$
\int \tlambda(\vec{k},t)d^3{\vec{k}}=1  \,\forall t\in\RR{},
$$
and for any $L>0$ define $\tlambda_{L}\in \SRR{4}$ by
$$
\tlambda_L(\vec{k},t)=L^3\tlambda(\vec{k}L,L^{-1}t).
$$
Then for any $n\in\mathbb{N}$ and any $f\in\SRR{4n}$
\begin{equation}
\lim_{L\rightarrow \infty}\GreenUnag{n}{\alpha}{g\tlambda_L}[f]=\GreenUnagAd{n}{\alpha}{g}[f],
\label{eq:Main:Green}    
\end{equation}

\begin{equation}
    \lim_{L\rightarrow \infty}\WightmanRnag{n}{\alpha}{g\tlambda_L}[f]=\WightmanRnagAd{n}{\alpha}{\alpha}{g}[f],
    \label{eq:Main:WightR}
\end{equation}

\begin{equation}
\lim_{L\rightarrow \infty}\WightmanUnag{n}{\alpha}{g\tlambda_L}[f]=\WightmanUnagAd{n}{\alpha}{g}[f],
    \label{eq:Main:WightU}
\end{equation}

where $\GreenUnagAd{n}{\alpha}{g},\WightmanUnagAd{n}{\alpha}{g}\in\SpRR{4n}\formalPS{g}$ and $\WightmanRnagAd{n}{g}\in\LL(\SRR{3n},\MltT(n))\formalPS{g}$ do not depend on $\tlambda$.
\end{=>}
\begin{proof}
Fix $V\in \mathbb{N}_0$ and $f\in\SRR{3n}$.
Take $h\in\SRR{4}$ as in Theorem \ref{thm:MainW} and for each $L>0$ define 
$$\tlambda^{(1)}_L=\mathcal{M}_{[h]}^{1,0}[\tlambda_L],$$
$$\tlambda^{(2)}_L=\tlambda_L-\tlambda^{(1)}_L.$$
It is easy to see that $\tlambda^{(2)}_L\subset{L\rightarrow \infty}{\longrightarrow} 0$ in $\SRR{4}$, so
$$
\lim_{L\rightarrow \infty}\WightmanRnaV{n}{\alpha}{V}\left[f\otimes \left((\tlambda^{\otimes V}_L-\tlambda^{(1)\otimes V}_L\right)\right]=0.
$$
Now by Theorem \ref{thm:MainW} we have: 
$$
\WightmanRnaV{n}{\alpha}{V}[f\otimes\tlambda^{(1)\otimes V}_L](\arrs{t}{n})=
(\WightmanRnaV{n}{\alpha}{V}\circ \mathcal{M}_{[h]}^{1,0})[f\otimes\tlambda^{\otimes V}_L]\arrs{t}{n}=
$$
$$
\int{\left(\prod_{j=1}^{V} \tlambda_L(\vec{k}_j,\tau)\right)\restrictH{3n+1,\ldots,3n+4V}{1,\ldots,4V}(\WightmanRnaV{n}{\alpha}{V}\circ \mathcal{M}_{[h]}^{1,0})[f]\left(\left(\tau_j,\vec{k}_j\right)_{j=1\ldots V},\arrs{t}{n}\right)d^{3}\vec{k}_jd\tau_j. }$$
Considering now $\tlambda$ as a distribution and noting that by assumptions of the Corollary it weakly converges to $\delta(\vec{k})$, we get that the limit $L\rightarrow \infty$ exists and is equal to
$$
\int{\restrictH{n+1,\ldots,n+4V}{1,\ldots,4V}(\WightmanRnaV{n}{\alpha}{V}\circ \mathcal{M}_{[h]}^{1,0})[f]\left(\arrs{t}{n},\arrs{\left(\tau_i,\vec{0},\right)}{V}\right)d^V\arrs{\tau}{V}. }
$$
This proves (\ref{eq:Main:WightR}). Analogously, one can show (\ref{eq:Main:Green})  and (\ref{eq:Main:WightU}).
\end{proof}
\begin{rmk}\label{rmk:MainW}
Theorem \ref{thm:MainW} says more than Corollary \ref{=>:MainW}. There is no need to take the uniform scaling limit in all directions. Instead, we can take any sequence or family of function 
$$
\tlambda_L(t,\vec{k})\underset{L\rightarrow\infty}{\longrightarrow} \delta(\vec{k})\, \mathrm{in}\, \SpRR{4},
$$
provided that
$$
\tlambda_L-\mathcal{M}_{[h]}^{1,0}\tlambda_L \underset{L\rightarrow\infty}{\longrightarrow} 0\, \mathrm{in}\, \SRR{4}.
$$
Extending the proof above we may also conclude the spatial and temporal limits can be taken in any order. Finally, we can put to the adiabatic limit each vertex of the graph individually. As a result, the Feynman rules in the adiabatic limit can be derived from Subsection \ref{HpQFT/FR}  (some variants are listed in Appendix). 
This nice behavior differs drastically from the strong adiabatic limit \cite{EG73}.
\end{rmk}

\subsection{Applying Feynman rules}
$n, V\in\mathbb{N}_0$,  $\arrs{\alpha}{n}\in \{+,-\}$, $\arrs{v}{n}\in\NN_0$, $\Gamma\in \mathfrak{G}_{\arrs{\alpha}{n}}^{\arrsM{v}{0}{n}}$, $f\in \SRR{3V}$ and $\arrs{t}{n}\in\RR{}$ by Corollary \ref{=>:HpQFT/FR/W-total} we have:

\begin{equation}\label{eq:Main/FormalW}
\WightmanRnaVarrG{n}{\alpha}{v}{\Gamma}[f](\arrs{t}{n})=     \int  \intUSimplex{+}{v_0}\symbDistArgs{(\tau_{0,j})_{j=1,\ldots,v_0}}{t_1}\times
\end{equation}
        $$\left(\prod_{i=1}^{n-1}\intSimplex{v_i}\symbDistArgs{(\tau_{i,j})_{j=1,\ldots,v_i}}{t_{i+1},t_i}\right)
        \intUSimplex{-}{v_n}\symbDistArgs{(\tau_{n,j})_{j=1,\ldots,v_n}}{t_n}$$
$$ F_{\Gamma}(\arrs{\vec{k}}{I_{\Gamma}},\arrs{\vec{p}}{n})
f\left(\arrs{\vec{p}}{n},\left(\tau_{i,j},\vec{q}_{i,j}^{\Gamma}(\arrs{\vec{k}}{I_{\Gamma}},\arrs{\vec{p}}{n})\right)_{i=1\ldots n,j=1\ldots v_n}\right)\times
$$
$$
\exp\left(\ii\sum_{i=0}^{n}\sum_{j=1}^{v_i}\Delta_{i,j}^{\Gamma}(\arrs{\vec{k}}{I_{\Gamma}},\arrs{\vec{p}}{n})\tau_{i,j}\right)
\left(\prod_{i=0}^{n} \prod_{j=1}^{v_i} d\tau_{i,j}\right)d^{3I_{\Gamma}}\arrs{\vec{k}}{I_{\Gamma}}d^{3n}\arrs{\vec{p}}{n}.
$$
The rest of this subsection is devoted to the clarification of the notation used in (\ref{eq:Main/FormalW}).
\par 
For each graph $\Gamma\in \mathfrak{G}_{\arrs{\alpha}{n}}^{\arrsM{v}{0}{n}}$ we set $I_{\Gamma}$ for the total number of internal (i.e. connecting to internal vertices) lines. 
Then $\arrs{\vec{k}}{I_{\Gamma}}\in\RR{3I_{\Gamma}}$ denotes the corresponding internal momenta. $\arrs{\vec{p}}{n}\in\RR{3n}$ are the $n$ external momenta. 
\par 
For a vertex marked by $(i,j)$ of a graph  $\Gamma\in \mathfrak{G}_{\arrs{\alpha}{n};\arrs{t}{n}}$ we denote with $l_{i,j}^{\Gamma}$ and ${l'}_{i,j}^{\Gamma}$ the number of incoming and outgoing lines respectively.   
$$
\kappa^{\Gamma}_{i,j}(\arrs{\vec{k}}{I_{\Gamma}},\arrs{\vec{p}}{n})\in \RR{3(l_{i,j}^{\Gamma}+{l'}_{i,j}^{\Gamma})}
$$

denotes the collection of the corresponding incoming/outgoing momenta as a function of the introduced above internal and external momenta. By construction for $i=0,\ldots,n$ and $j=1,\ldots,v_i$, $\kappa^{\Gamma}_{i,j}$ is a linear map. For convenience, we assume that the first $l'$ and the rest $l$ components are the momenta of the outgoing and incoming particles respectively.  
\par 
Finally for every $\arrs{\vec{k}}{I_{\Gamma}}\in\RR{3V}$ and every vertex $(i,j)$ (where as always $i=0,\ldots,n$ and $j=1,\ldots,v$) we define the momentum and energy defects
$$
\vec{q}_{i,j}^{\Gamma}\left(\arrs{\vec{k}}{I_{\Gamma}},\arrs{\vec{p}}{n}\right)=
\sum_{r=1}^{{l'}_{i,j}^{\Gamma}} \left(
    \kappa^{\Gamma}_{i,j}(\arrs{\vec{k}}{I_{\Gamma}},\arrs{\vec{p}}{n})\right)_r
-\sum_{r={l'}_{i,j}^{\Gamma}+1}^{{l}_{i,j}^{\Gamma}+{l'}_{i,j}^{\Gamma}}\left( \kappa^{\Gamma}_{i,j}(\arrs{\vec{k}}{I_{\Gamma}},\arrs{\vec{p}}{n})\right)_r,
$$
$$\Delta_{i,j}^{\Gamma}\left(\arrs{\vec{k}}{I_{\Gamma}},\arrs{\vec{p}}{n})\right)=
$$
$$
\sum_{r=1}^{{l'}_{i,j}^{\Gamma}} \omega_0\left(\left(
    \kappa^{\Gamma}_{i,j}(\arrs{\vec{k}}{I_{\Gamma}},\arrs{\vec{p}}{n})\right)_r\right)
-\sum_{r={l'}_{i,j}^{\Gamma}+1}^{{l}_{i,j}^{\Gamma}+{l'}_{i,j}^{\Gamma}}\omega_0\left(\left( \kappa^{\Gamma}_{i,j}(\arrs{\vec{k}}{I_{\Gamma}},\arrs{\vec{p}}{n})\right)_r\right),$$
i.e. the sum of the (on-shell) energies (momenta) of all the outgoing particles minus the sum of the (on-shell) energies (momenta) of all the incoming particles. Here by $(\cdot)_r$ we mean projection on the $r$th factor of $\RR{3({l}_{i,j}^{\Gamma}+{l'}_{i,j}^{\Gamma})}$ treated as $\left(\RR{3}\right)^{\otimes({l}_{i,j}^{\Gamma}+{l'}_{i,j}^{\Gamma})}$.
\par
All vertex factors are collected into one function
$$
F_{\Gamma}(\arrs{\vec{k}}{I_{\Gamma}},\arrs{\vec{p}}{n})=
 \prod_{i=1}^{n+1} \prod_{j=1}^{v_i} F_{{l'}_{i,j}^{\Gamma},l_{i,j}^{\Gamma,i,j}}\left(\kappa^{\Gamma}_{i,j}(\arrs{\vec{k}}{I_{\Gamma}},\arrs{\vec{p}}{n})\right), $$
 $$
 \forall \arrs{\vec{k}}{I_{\Gamma}}\in\RR{3I_{\Gamma}}\,  \forall \arrs{\vec{p}}{n}\in\RR{3n}.
$$
\par

\subsection{Spatial adiabatic limit}
We analyze (\ref{eq:Main/FormalW}) part by part, starting from the integral over the momenta.
\par 
The UV convergence in the momenta space is controlled by the vertex factors as follows from the following lemma.  
\begin{lem}\label{lem:Main/SpatialUV}
Using the notation introduced above, for any graph $\Gamma\in \mathfrak{G}_{\arrs{\alpha}{n}}^{\arrsM{v}{0}{n}}$ one has  $F_{\Gamma}\in \SRR{3(I_{\Gamma}+n)}$:
\end{lem}
\begin{proof}
Since each $F_{l',l}\in\SRR{3(l+l')}$ and $\kappa^{\Gamma}_{i,j}$ is a linear function, it is clear that $F_{\Gamma}$ is smooth. For the same reason each its factor is at least bounded together with all its partial derivatives considered as a function on $\RR{3(I_{\Gamma}+n)}$. At the same time each of $\vec{k}_j$, $j=1,\ldots, I_{\Gamma}$ and each of $\vec{p}_{j=1,\ldots,n}$ appears as an argument of at least one of the factors. From this the standard estimations of type (\ref{eq:xy<x}) one concludes that $F_{\Gamma}$, as well as all its partial derivatives, decays at infinity faster than any polynomial.
\end{proof}
 So, we see that convergence of the integral over the internal momenta in (\ref{eq:Main/FormalW}) is completely controlled by the part independent of both the adiabatic cut-off (hidden in the test function $f$) and the timestamps $\tau_{i,j}$. 
 
 \par
For the infrared behavior we need also to effectively resolve all momentum conservation constraints expressing the internal momenta via external momenta, momentum defects, and independent loop momenta.
Formally this is expressed by the following lemma.
\begin{lem}\label{lem:Main/spatialIR}
For any choice of  $n,V\in\mathbb{N}$ and $\arrsM{v}{0}{n}\in \mathbb{N}_0^{n+1}$ with $\sum_{i=0}^n v_i=V$ and a graph $\Gamma\in \mathfrak{G}_{\arrs{\alpha}{n}}^{\arrsM{v}{0}{n}}$ the functionals 
$$\kappa_{i,j}^{\Gamma}\in{\RR{3V}}',\quad i=1,\ldots,n, \quad, j=1,\ldots,v_i$$
are linearly independent.
\end{lem}
This fact (although not in this form) is well-known in quantum physics, but we still formally prove it for completeness.
\begin{proof}
Assume that there are $\alpha_{i,j}\in\mathbb{R}$, $i=0,\ldots,n$, $j=1,\ldots,v_i$ such that
\begin{equation}\label{eq:Main/KappaIndependent}
\sum_{i=0}^n\sum_{j=1}^{v_i}\alpha_{i,j}\kappa_{i,j}^{\Gamma}=0.
\end{equation}
For each internal vertex $(i,j)$ of the graph define $d_{\Gamma}(i,j)$, the shortest length of a path in $\Gamma$, connecting $(i,j)$ with an external vertex. Recall that all connected components of the relevant graphs contain some external vertices, so this function is well-defined. Thus we can prove that $\alpha_{i,j}=0$ for any vertex $(i,j)$ by induction in $d_{\Gamma}(i,j)$. 
\begin{itemize}
    \item Base of induction: if $d_{\Gamma}(i,j)=1$ then there is an external vertex connected with the internal vertex $(i,j)$.Let it be the $r$th external vertex. We evaluate (\ref{eq:Main/KappaIndependent}) by setting all the internal momenta and all the external momenta except $\vec{p}_r$ to zero. Then only $\kappa_{i,j}^{\Gamma}$ survives and we get $\alpha_{i,j}=0$.
    \item If we already know that $\alpha_{i,j}=0$ for all vertices $(i,j)$ such that $d_{\Gamma}(i,j)\leq d_0$. Take a vertex $(i',j')$  such that $d_{\Gamma}(i',j')=d_0+1$. Then it is connected with a vertex, say, $(i,j)$ with $d_{\Gamma}(i,j)=d_0$. Assume that one of the connecting edges is marked by $r$. We evaluate (\ref{eq:Main/KappaIndependent}) setting all the external and internal momenta except $\vec{k}_r$ to zero. Then only $\kappa_{i,j}^{\Gamma}$ and $\kappa_{i',j'}^{\Gamma}$ survive, but we already know that $\alpha_{i,j}=0$. Thus $\alpha_{i',j'}=0$ too. 
\end{itemize}
\end{proof}

\begin{=>}\label{=>:Main/Spatial}
The expression (\ref{eq:Main/FormalW}) may be rewritten as
\begin{equation}\label{eq:Main/FormalWS}\WightmanRnaVarrG{n}{\alpha}{v}{\Gamma}[f](\arrs{t}{n})=
 \int  \intUSimplex{+}{v_0}\symbDistArgs{(\tau_{0,j})_{j=1,\ldots,v_0}}{t_1}\left(\prod_{i=1}^{n-1}\intSimplex{v_i}\symbDistArgs{(\tau_{i,j})_{j=1,\ldots,v_i}}{t_{i+1},t_i}\right)\times
\end{equation}
$$
\intUSimplex{-}{v_n}\symbDistArgs{(\tau_{n,j})_{j=1,\ldots,v_n}}{t_1}
$$
$$F^{\mathrm{S}}_{\Gamma}(\left(\vec{\kappa}_{i,j}\right)_{i=1\ldots n,j=1\ldots v_n},\arrs{\vec{q}}{I_{\Gamma}+n-V})\times$$
$$
f\left(\arrs{\vec{p}_{\Gamma}(\left(\vec{\kappa}_{i,j}\right)_{i=1\ldots n,j=1\ldots v_n},\arrs{\vec{q}}{I_{\Gamma}+n-V})}{n},\left(\tau_{i,j},\vec{\kappa}_{i,j}\right)_{i=1\ldots n,j=1\ldots v_n}\right)\times
$$

$$
\exp\left(\ii\sum_{i'=0}^{n}\sum_{j'=1}^{v_{i'}}\Delta_{i',j'}^{\mathrm{S}\Gamma}\left(\left(\vec{\kappa}_{i,j}\right)_{i=1\ldots n,j=1\ldots v_i},\arrs{\vec{q}}{I_{\Gamma}+n-V}\right)\tau_{i',j'}\right)\times
$$
$$
\left(\prod_{i=0}^{n} \prod_{j=1}^{v_i} d\tau_{i,j}d\kappa_{i,j}\right)d^{3(I_{\Gamma}+n-V)}\arrs{\vec{q}}{I_{\Gamma}+n-V}.
$$

Here $\vec{p}_{\Gamma}$ is a linear function and $F^{\mathrm{S}}_{\Gamma}$, $\Delta_{i,j}^{\mathrm{S}\Gamma}$ are precompositions of $F^{\mathrm{S}}_{\Gamma}$, $\Delta_{i,j}^{\mathrm{S},\Gamma}$ with an invertible linear transform.
In particular, $F^{\mathrm{S}}_{\Gamma}\in\SRR{3(n+I_{\Gamma}
)}$.

\end{=>}
\begin{proof}
By Lemma \ref{lem:Main/spatialIR} we can add coordinates $\vec{q}^{\Gamma}_i\in\RR{3(n+I_{\Gamma})*}$, $i=1,\ldots,n+I_{\Gamma}-V$ completing the family of  functionals $\kappa_{i,j}^{\Gamma}$ to the basis. Then the form (\ref{eq:Main/FormalWS}) exists. The last statement is a reformulation of Lemma \ref{lem:Main/SpatialUV}.
\end{proof}
\subsection{Temporal adiabatic limit}\label{Main/Temp}
Now we focus on integrals over the timestamp positions, starting from the inner vertices. 
From (\ref{eq:Main/FormalW}) we get
\begin{equation}\label{eq:Main/FormalWS*h}
\left(\WightmanRnaVarrG{n}{\alpha}{v}{\Gamma}\circ \mathcal{M}_{[h]}^{V,n}\right)[f](\arrs{t}{n})=     
\end{equation}
$$
\int  \intUSimplex{+}{v_0}\symbDistArgs{(\tau_{0,j})_{j=1,\ldots,v_0}}{t_1}\left(\prod_{i=1}^{n-1}\intSimplex{v_i}\symbDistArgs{(\tau_{i,j})_{j=1,\ldots,v_i}}{t_{i+1},t_i}\right)\times
$$
$$
\intUSimplex{-}{v_n}\symbDistArgs{(\tau_{n,j})_{j=1,\ldots,v_n}}{t_1}
$$
$$F^{\mathrm{S}}_{\Gamma}(\left(\vec{\kappa}_{i,j}\right)_{i=1\ldots n,j=1\ldots v_n},\arrs{\vec{q}}{I_{\Gamma}+n-V})\times\prod_{i=0}^{n}\prod_{j=1}^{v_i}h(\tau_{i,j}-\tau'_{i,j})$$
$$
f\left(\arrs{\vec{p}_{\Gamma}(\left(\vec{\kappa}_{i,j}\right)_{i=1\ldots n,j=1\ldots v_n},\arrs{\vec{q}}{I_{\Gamma}+n-V})}{n},\left(\tau'_{i,j},\vec{\kappa}_{i,j}\right)_{i=1\ldots n,j=1\ldots v_n}\right)\times
$$

$$
\exp\left(\ii\sum_{i'=0}^{n}\sum_{j'=1}^{v_{i'}}\Delta_{i',j'}^{\mathrm{S}\Gamma}\left(\left(\vec{\kappa}_{i,j}\right)_{i=1\ldots n,j=1\ldots v_i},\arrs{\vec{q}}{I_{\Gamma}+n-V}\right)\tau_{i',j'}\right)\times
$$
$$
\left(\prod_{i=0}^{n} \prod_{j=1}^{v_i} d\tau_{i,j}d\tau'_{i,j}d\kappa_{i,j}\right)d^{3(I_{\Gamma}+n-V)}\arrs{\vec{q}}{I_{\Gamma}+n-V}.
$$
Since all the integrals above are absolutely convergent, the integration can be performed in any order. This way we arrive to
\begin{equation}\label{eq:Main/Temp-spec}
\left(\WightTf{n}{\alpha}{\arrsM{v}{0}{n}}\circ \mathcal{M}_{[h]}^{V,n}\right)[f](\arrs{t}{n})=    \sum_{\Gamma\in\mathfrak{G}_{\arrs{\alpha}{n}}^{\arrsM{v}{0}{n}}}  \int 
\end{equation}

$$F^{\mathrm{S}}_{\Gamma}(\left(\vec{\kappa}_{i,j}\right)_{i=1\ldots n,j=1\ldots v_n},\arrs{\vec{q}}{I_{\Gamma}+n-V})\times$$
$$
O_{[h]\Gamma,0}\left(\left(\tau'_{0,j}\right)_{j=1\ldots v_{0}},\left(\vec{\kappa}_{i,j}\right)_{i=1\ldots n,j=1\ldots v_i},\arrs{\vec{q}}{I_{\Gamma}+n-V},t_{1}
\right)
$$
$$\prod_{i'=1}^{n-1}
O_{[h]\Gamma,i'}\left(\left(\tau'_{i',j}\right)_{j=1\ldots v_{i'}},\left(\vec{\kappa}_{i,j}\right)_{i=1\ldots n,j=1\ldots v_i},\arrs{\vec{q}}{I_{\Gamma}+n-V},t_{i},t_{i+1}
\right)\times
$$
$$
O_{[h]\Gamma,n}\left(\left(\tau'_{,n}\right)_{j=1\ldots v_{n}},\left(\vec{\kappa}_{i,j}\right)_{i=1\ldots n,j=1\ldots v_i},\arrs{\vec{q}}{I_{\Gamma}+n-V},t_{n}
\right)
$$
$$
f\left(\arrs{\vec{p}_{\Gamma}(\left(\vec{\kappa}_{i,j}\right)_{i=1\ldots n,j=1\ldots v_i},\arrs{\vec{q}}{I_{\Gamma}+n-V})}{n},\left(\tau'_{i,j},\vec{\kappa}_{i,j}\right)_{i=1\ldots n,j=1\ldots v_n}\right)\times
$$
$$
\left(\prod_{i=0}^{n} \prod_{j=1}^{v_i} d\tau'_{i,j}d\kappa_{i,j}\right)d^{3(I_{\Gamma}+n-V)}\arrs{\vec{q}}{I_{\Gamma}+n-V},
$$
where for $i'=1,\ldots,n-1$
\begin{equation}
    \label{eq:Main/OGamma-Inner}
O_{[h]\Gamma,i'}\left(\arrs{\tau'}{v_{i'}},\left(\vec{\kappa}_{i,j}\right)_{i=1\ldots n,j=1\ldots v_i},\arrs{\vec{q}}{I_{\Gamma}+n-V},t',t\right)=
\end{equation}
$$
\int_{ t}^{t'}d\tau_{1}\int_{ t}^{\tau_{1}}d\tau_{2}\cdots \int_{ t}^{\tau_{v_i-1}}d\tau_{v_i} 
\exp\left(\ii\sum_{j'=1}^{v_{i'}}\Delta_{i',j'}^{\mathrm{S}\Gamma}\left(\left(\vec{\kappa}_{i,j}\right)_{i=1\ldots n,j=1\ldots v_i},\arrs{\vec{q}}{I_{\Gamma}+n-V}\right)\tau_{i',j'}\right)\times
$$
$$
\prod_{j=1}^{v_{i'}}h(\tau_{j}-\tau'_{j})
\prod_{j=1}^{v_{i'}} d\tau_{j},
$$
and 
\begin{equation}
    \label{eq:Main/OGamma-Pre}
O_{[h]\Gamma,0}\left(\arrs{\tau'}{v_{0}},\left(\vec{\kappa}_{i,j}\right)_{i=1\ldots n,j=1\ldots v_i},\arrs{\vec{q}}{I_{\Gamma}+n-V},t',t\right)=
\end{equation}
$$
\int_{t}^{+\infty}d\tau_{1}\int_{ t}^{\tau_{1}}d\tau_{2}\cdots \int_{ t}^{\tau_{v_0-1}}d\tau_{v_0} 
\exp\left(\ii\sum_{j'=1}^{v_{0}}\Delta_{0,j'}^{\mathrm{S}\Gamma}\left(\left(\vec{\kappa}_{i,j}\right)_{i=1\ldots n,j=1\ldots v_i},\arrs{\vec{q}}{I_{\Gamma}+n-V}\right)\tau_{0,j'}\right)\times
$$
$$
\prod_{j=1}^{v_{0}}h(\tau_{j}-\tau'_{j})
\prod_{j=1}^{v_{0}} d\tau_{j},
$$

\begin{equation}
    \label{eq:Main/OGamma-Post}
O_{[h]\Gamma,n}\left(\arrs{\tau'}{v_{n}},\left(\vec{\kappa}_{i,j}\right)_{i=1\ldots n,j=1\ldots v_i},\arrs{\vec{q}}{I_{\Gamma}+n-V},t\right)=
\end{equation}
$$
\int_{-\infty}^{t}d\tau_{1}\int_{ -\infty}^{\tau_{1}}d\tau_{2}\cdots \int_{-\infty}^{\tau_{v_n-1}}d\tau_{v_n} 
\exp\left(\ii\sum_{j'=1}^{v_{n}}\Delta_{n,j'}^{\mathrm{S}\Gamma}\left(\left(\vec{\kappa}_{i,j}\right)_{i=1\ldots n,j=1\ldots v_i},\arrs{\vec{q}}{I_{\Gamma}+n-V}\right)\tau_{n,j'}\right)\times
$$
$$
\prod_{j=1}^{v_{n}}h(\tau_{j}-\tau'_{j})
\prod_{j=1}^{v_{n}} d\tau_{j}.
$$
Note that here we use the direct functional interpretation of distributions $\intSimplex{n}$ and $\intUSimplex{n}{\pm}$ rather than the combinatoric version of Remark \ref{rmk:HpQFT/FR/combFactor}. The main advantage of this approach is the manifest smoothness with respect to $\arrs{t}{n}$.
\par 
There are two very different kinds of contributions: the \emph{inner vertices} $\bullet_{(i,j)}$, $i=1,\ldots,n_1$,  and \emph{outer vertices}  $\bullet_{(0,j)}$ and $\bullet_{(n,j)}$. We treat them separately in two subsections. The terminology comes from the fact that the inner vertices correspond to operators standing between two source insertions (the external vertices), while the operators corresponding to the outer ones appear either before or after all the insertions. This distinction is crucial for the treatment of the corresponding factors.
\subsubsection{Inner vertices: compact region of integration}
In (\ref{eq:Main/OGamma-Inner}) the region of integration is compact and all the estimations are rather direct. We have to deal with the fact that the integration region depends on the arguments but they do not bring anything but a polynomial factor. It is convenient to first formalize this fact.
\begin{lem}\label{lem:Main/TInner}
Fix $v,m\in\mathbb{N}_0$ and $f\in\SRR{m+v+2}$. Define a function $H$ on $\RR{m+v+2}$ by
$$
H(\arrs{\tau'}{v},x,t',t)=
\int_{ t}^{t'}d\tau_{1}\int_{ t}^{\tau_{1}}d\tau_{2}\cdots \int_{ t}^{\tau_{v-1}}d\tau_v g((\tau_j-\tau'_j)_{j=1\ldots v},\bm{x},t',t)d^{v}\arrs{\tau}{v}, 
$$
$$
\forall t,t'\in \RR{}, \forall \arrs{\tau'}{v}\in\RR{v}, \, \forall \bm{x}\in\RR{m}
$$
Then $H\in\SRR{n+v+2}$.
\end{lem}
\begin{proof}
The smoothness is clear. In particular, $H$ is continuous together with all its derivatives, so it is enough to estimate it for $t\neq t'$. In this case we can change variables of integration to 
$$
\xi_i=\frac{\tau_i-t}{t'-t}, \quad, i=1,\ldots,v.
$$
This way for any allowed arguments we get 
$$
H(\arrs{\tau'}{v},\bm{x},t',t)=$$
$$(t'-t)^{v}
\int_{0}^1d\xi_1\int_{0}^{\xi_1}d\xi_2\ldots\int_{0}^{\xi_{n-1}}d\xi_n g((\xi_j(t'-t)+t-\tau'_j)_{j=1\ldots v},\bm{x},t',t)d^{v}\arrs{\xi}{v}.
$$
The integrand is clearly fast decaying with respect to $t$, $t'$ and $\bm{x}$. For $\tau'_j$ note, that
whenever $0\leq\xi_j\leq 1$, we have
$$
(1+|\xi_j(t'-t)+t-\tau'_j|)(1+|t|)(1+|t'|)\geq 1+|\tau'_j|,
$$
so the integrand is fast decaying with respect to $\tau'$. Then using standard estimations of type (\ref{eq:x<xy}-\ref{eq:xy<x}) we get that the integral is fast decaying, and the partial derivatives can be treated in the same way. 
\end{proof}
\begin{=>}\label{=>:Main/TempI}
For $n,V,\arrs{v}{n},\arrs{\alpha}{n},\Gamma$ as in Theorem \ref{thm:MainW}, $i\in\{1,\ldots,n\}$ and  $h\in Q_{\frac{M}{V+1}}$ one has

$$O_{[h]\Gamma,i}\in\SMlt{v_i}{3(I^{\Gamma}+n)+2}.$$
\end{=>}
\begin{proof}
By Remark \ref{rmk:Framework/Mlt/Alt} it is enough to prove that for any $\chi\in\SRR{3(I_{\Gamma}+n)+2}$ $\MltMn{\chi}{v}[O_{[h]\Gamma,i'}]\in\SRR{3(I_{\Gamma}+n)+2+v_i}$.
\par
For that use Lemma \ref{lem:Main/TInner} with $v=v_i$, $m=3(I_{\Gamma}+n)$ and
$$
g(\arrs{\tau}{v_i},\bm{x},t',t)=$$
$$
\exp\left(\ii\sum_{j'=1}^{v_{i'}}\Delta_{i',j'}^{\mathrm{S}\Gamma}\left(\bm{x}\right)\tau_{j'}\right)
\prod_{j=1}^{v_{i'}}h(\tau_{j})
\chi\left(\bm{x},t',t\right), \forall \bm{x}\in\RR{3(I_{\Gamma}+n)+2}.
$$
Then 
$$
\MltMn{\chi}{v}[O_{[h]\Gamma,i'}]\left(\arrs{\tau'}{v_i},\bm{x},t',t\right)=H(\arrs{\tau'}{v_i},\bm{x},t',t)R(\arrs{\tau'}{v_i},\bm{x})
$$
where $H\in\SRR{v_i+3(I_{\Gamma}+n)+2}$ is constructed in Lemma \ref{lem:Main/TInner} and 
$$
R(\arrs{\tau'}{v_i},\bm{x})=
\exp\left(\ii\sum_{j'=1}^{v_{i'}}\Delta_{i',j'}^{\mathrm{S}\Gamma}\left(\bm{x}\right)\tau'_{i',j'}\right).
$$
%\TBD{verify}
The statement then follows by $R\in\Mlt{v_i+3(I_{\Gamma}+n)}$.
\end{proof}

\subsubsection{Outer vertices: spectral properties}
The integration regions in the right-hand sides of (\ref{eq:Main/OGamma-Pre}-\ref{eq:Main/OGamma-Post}) are not compact, so at the first sight, the convergence is controlled by the adiabatic cut-off function only. The situation becomes more clear in the spectral representation. Then $h\in Q_{\frac{M}{V+1}}$ makes the integration region compact. The key observation is that physically these groups of vertices describe creation of several particles from the vacuum (if $i=n$) or completion annihilation of particles up to the vacuum (if $i=0$) as an on-shell process. 
\par 
We consider in details the case $i=0$.
\par 
It is convenient to shift the variables of integration to $\arrs{s^{\mathrm{f}}}{v}$ as follows
$$\tau_{j}=t+\sum_{j'=1}^{j}s^{\mathrm{f}}_{j'}.$$
    Then the integration region becomes $\RR{v_0}_{+}$ and
    \begin{equation}\label{eq:Main/O0}
    O_{[h]\Gamma,0}\left(\arrs{\tau'}{v_0},\left(\vec{\kappa}_{i,j}\right)_{i=1\ldots n,j=1\ldots v_i},\arrs{\vec{q}}{I_{\Gamma}+n-V},t\right)=
\end{equation}
$$
\exp\left(-\ii\Omega_{v_0}^{\mathrm{f},\Gamma}\left(\left(\vec{\kappa}_{i,j}\right)_{i=1\ldots n,j=1\ldots v_i},\arrs{\vec{q}}{I_{\Gamma}+n-V}\right)t\right)\times
$$
$$
    \int_{\RR{v_0}_+} 
    \exp\left(-\ii\sum_{j=1}^{v_{0}}\Omega_{j}^{\mathrm{0}\Gamma}\left(\left(\vec{\kappa}_{i,j}\right)_{i=1\ldots n,j=1\ldots v_i},\arrs{\vec{q}}{I_{\Gamma}+n-V}\right)s_j\right)\times
    $$
    $$
    \prod_{j=1}^{v_{0}}h\left(t_0+\sum_{j'=1}^{j}s_{j'}-\tau'_{j}\right) ds_j.
    $$
    Here we introduced the commulated energy deffects
    \begin{equation}\label{eq:Main/Omega-f}
    \Omega_{j}^{\mathrm{f}\Gamma}=\sum_{j'=j}^{v_0}    (-\Delta_{0,j'}^{\mathrm{S}\Gamma}).
    \end{equation}
    The following observation is important. 
    \begin{rmk}\label{rmk:Main/MassBound}
    For any values of the arguments \begin{equation}\label{eq:Main/OmegaBound}
        \Omega_{j}^{\mathrm{f}\Gamma}\left(\left(\vec{\kappa}_{i,j}\right)_{i=1\ldots n,j=1\ldots v_i},\arrs{\vec{q}}{I_{\Gamma}+n-V}\right)\geq M.
    \end{equation}
    \end{rmk}
    \begin{proof}
    It follows from the fact that all lines started at some of vertices $(0,j)$ with $j=1,\ldots,v'$ must end at some other vertex from the same set, because the outgoing state is the vacuum and all external vertices may appear only before $(0,j)$. Let us see some details.
    \par 
        Each $(-\Delta_{0,j'}^{\mathrm{S}\Gamma})$ contains both positive contributions of the incoming lines and negative contributions of the outgoing lines of the vertex $(0,j')$. 
        Each outgoing line should end at some later vertex $(0,j'')$, $j''>j$. In other words, all negative contributions to the right-hand sides of (\ref{eq:Main/Omega-f}]) cancel out. Finally, if for some vertex all positive contributions were also canceled it would mean that this vertex belongs to a disconnected pieces of the diagram which is forbidden. Then by the mass gap assumed we get (\ref{eq:Main/OmegaBound}).
    \end{proof}
   Let us see how it helps to deal with $O_{[h]\Gamma,0}$. As long as $h\in Q_{\frac{M}{V+1}}$, (\ref{eq:Main/O0}) is equivalent to
    \begin{equation}\label{eq:Main/O0-spec}
    O_{[h]\Gamma,0}\left(\arrs{\tau'}{v_0},\left(\vec{\kappa}_{i,j}\right)_{i=1\ldots n,j=1\ldots v_i},\arrs{\vec{q}}{I_{\Gamma}+n-V},t\right)=
    \end{equation}
    $$
    \int\exp\left(-\ii\left(\sum_{j=1}^{v_0}\Delta_j+\Omega_{v_0}^{\mathrm{f}\Gamma}\left(\left(\vec{\kappa}_{i,j}\right)_{i=1\ldots n,j=1\ldots v_i},\arrs{\vec{q}}{I_{\Gamma}+n-V}\right)\right)t\right)\times
    $$
    $$
     \prod_{j=1}^{v_0}
    \frac{\ii \exp \left(\ii\tau'_j\Delta_j\right)\widetilde{h}(\Delta_j) d\Delta_j}
    {-\sum_{j'=j}^{v_0}\Delta_{j'}-\Omega_{j}^{\mathrm{f}\Gamma}\left(\left(\vec{\kappa}_{i,j}\right)_{i=1\ldots n,j=1\ldots v_i},\arrs{\vec{q}}{I_{\Gamma}+n-V}\right)+\ii0}.
    $$
    \begin{rmk}\label{rmk:Main/SpectralInt}
    In principle, this expression involves symbolic integral which should be understood as an evaluation of a distribution 
    $$
    F(\arrs{\Delta}{v_0})= \prod_{j=1}^{v_0}
    \frac{\ii \exp \left(\ii\tau'_j\Delta_j\right)\widetilde{h}(\Delta_j)}
    {-\sum_{j'=j}^{v_0}\Delta_{j'}-\Omega_{j}^{\mathrm{f}\Gamma}\left(\left(\vec{\kappa}_{i,j}\right)_{i=1\ldots n,j=1\ldots v_i},\arrs{\vec{q}}{I_{\Gamma}+n-V}\right)+\ii0}
    $$
    on a test function. But by Remark \ref{rmk:Main/MassBound} and $h\in Q_{\frac{M}{V+1}}$ the denominators are bounded from below,
    \begin{equation}
        \label{eq:Main/TempOuterDenF0}
    \left|
    -\sum_{j'=j}^{v_0}\Delta_{j'}-\Omega_{j}^{\mathrm{0}\Gamma}\left(\left(\vec{\kappa}_{i,j}\right)_{i=1\ldots n,j=1\ldots v_i},\arrs{\vec{q}}{I_{\Gamma}+n-V}\right)\right|\geq
    \end{equation}
    $$ M-v_0\frac{M}{V+1}\geq \frac{M}{V+1}.
    $$
    So, the integral in (\ref{eq:Main/O0-spec})  may be treated literally and $+\ii0$ can be ignored.
\end{rmk}
In the same way we get 
\begin{equation}\label{eq:Main/On-spec}
    O_{[h]\Gamma,n}\left(\arrs{\tau'}{v_n},\left(\vec{\kappa}_{i,j}\right)_{i=1\ldots n,j=1\ldots v_i},\arrs{\vec{q}}{I_{\Gamma}+n-V},t'\right)=
    \end{equation}
    $$
    \int\exp\left(-\ii\left(\sum_{j=1}^{v_n}\Delta_j-\Omega_{v_n}^{\mathrm{f},\Gamma}\left(\left(\vec{\kappa}_{i,j}\right)_{i=1\ldots n,j=1\ldots v_i},\arrs{\vec{q}}{I_{\Gamma}+n-V}\right)\right)t'\right)\times
    $$
    $$
     \prod_{j=1}^{v_n}
    \frac{\ii \exp \left(\ii\tau'_j\Delta_j\right)\widetilde{h}(\Delta_j) d\Delta_j}
    {\sum_{j'=1}^{j'}\Delta_{j'}-\Omega_{j}^{\mathrm{i}\Gamma}\left(\left(\vec{\kappa}_{i,j}\right)_{i=1\ldots n,j=1\ldots v_i},\arrs{\vec{q}}{I_{\Gamma}+n-V}\right)+\ii0},
    $$
    where 
    \begin{equation}
     \label{eq:Main/Omega-i}
    \Omega_{i}^{\mathrm{i}\Gamma}=\sum_{j'=1}^{j}    (\Delta_{0,j'}^{\mathrm{S}\Gamma}). 
    \end{equation}
    
    Similarly to Remark \ref{rmk:Main/MassBound} we have
    \begin{equation}
    \label{eq:Main/TempOuterDenF}
     \left |\sum_{j'=1}^{j'}\Delta_{j'}-\Omega_{j}^{\mathrm{i}\Gamma}\left(\left(\vec{\kappa}_{i,j}\right)_{i=1\ldots n,j=1\ldots v_i},\arrs{\vec{q}}{I_{\Gamma}+n-V}\right)\right|\geq \frac{M}{V+1},  
    \end{equation}
    and in particular the $+\ii0$ symbol can be ignored.
    \begin{lem}\label{lem:Main/TempOuter}
    For For $n,V,\arrs{v}{n},\arrs{\alpha}{n},\Gamma$ as in Theorem \ref{thm:MainW},
    $$O_{[h]\Gamma,i}\in \SMlt{v_i}{3(I_{\Gamma}+n)+1}, \, i=0,n.$$
    \end{lem}
\begin{proof}
By (\ref{eq:Main/O0-spec}) and (\ref{eq:Main/On-spec})  for $i=0,n$ $O_{[h]\Gamma,i}$ is a Fourier transform of $\widetilde{O}_{[h]\Gamma,i}$, 
$$\widetilde{O}_{[h]\Gamma,0}\left(\arrs{\Delta}{v_0},\left(\vec{\kappa}_{i,j}\right)_{i=1\ldots n,j=1\ldots v_i},\arrs{\vec{q}}{I_{\Gamma}+n-V},t\right)=$$
$$
\exp\left(-\ii\left(\sum_{j=1}^{v_0}\Delta_j+\Omega_{v_0}^{\mathrm{f}\Gamma}\left(\left(\vec{\kappa}_{i,j}\right)_{i=1\ldots n,j=1\ldots v_i},\arrs{\vec{q}}{I_{\Gamma}+n-V}\right)\right)t\right)\times
    $$
    $$
     \prod_{j=1}^{v_0}
    \frac{\ii \widetilde{h}(\Delta_j)}
    {-\sum_{j'=j}^{v_0}\Delta_{j'}-\Omega_{j}^{\mathrm{f}\Gamma}\left(\left(\vec{\kappa}_{i,j}\right)_{i=1\ldots n,j=1\ldots v_i},\arrs{\vec{q}}{I_{\Gamma}+n-V}\right)+\ii0}
    $$
and 

$$\widetilde{O}_{[h]\Gamma,n}\left(\arrs{\Delta}{v_n},\left(\vec{\kappa}_{i,j}\right)_{i=1\ldots n,j=1\ldots v_i},\arrs{\vec{q}}{I_{\Gamma}+n-V},t'\right)=$$
$$
\exp\left(-\ii\left(\sum_{j=1}^{v_n}\Delta_j-\Omega_{v_n}^{\mathrm{f},\Gamma}\left(\left(\vec{\kappa}_{i,j}\right)_{i=1\ldots n,j=1\ldots v_i},\arrs{\vec{q}}{I_{\Gamma}+n-V}\right)\right)t'\right)\times
    $$
    $$
     \prod_{j=1}^{v_n}
    \frac{\ii\widetilde{h}(\Delta_j) }
    {\sum_{j'=1}^{j'}\Delta_{j'}-\Omega_{j}^{\mathrm{i}\Gamma}\left(\left(\vec{\kappa}_{i,j}\right)_{i=1\ldots n,j=1\ldots v_i},\arrs{\vec{q}}{I_{\Gamma}+n-V}\right)+\ii0},
    $$
 where we used the accumulated energy defects (\ref{eq:Main/Omega-f}) and  
(\ref{eq:Main/Omega-i}). Taking into account (\ref{eq:Main/TempOuterDenF0}) and (\ref{eq:Main/TempOuterDenF}) we get
    $$\widetilde{O}_{[h]\Gamma,i}\in \SMlt{v_i}{3(I_{\Gamma}+n)+1}, \, i=0,n.$$
    Thus the same holds for its Fourier transform.
\end{proof}
\begin{rmk}
It is worth noting that the argument above is applicable only because both initial and final states are the vacuum. In particular, it does not work for the adiabatic limit of scattering amplitudes.
\end{rmk}
\subsection{Proof of Theorem \ref{thm:MainW}}
\begin{proof}
We have only to sum up the results of the two previous subsections. 
Indeed, applying Corollary \ref{=>:Main/TempI} and Lemma \ref{lem:Main/TempOuter} to the form (\ref{eq:Main/Temp-spec}), we have 
\begin{equation*}
\left(\WightTf{n}{\alpha}{\arrsM{v}{0}{n}}\circ \mathcal{M}_{[h]}^{V,n}\right)[f](\arrs{t}{n})=    \sum_{\Gamma\in\mathfrak{G}_{\arrs{\alpha}{n}}^{\arrsM{v}{0}{n}}}  \int 
\end{equation*}

$$N^{[h]}_{\Gamma}(\left(\tau'_{i,j}\right)_{i=0\ldots n, j=1\ldots v_i},\left(\vec{\kappa}_{i,j}\right)_{i=1\ldots n,j=1\ldots v_n},\arrs{\vec{q}}{I_{\Gamma}+n-V},\arrs{t}{n})\times$$
$$
f\left(\arrs{\vec{p}_{\Gamma}(\left(\vec{\kappa}_{i,j}\right)_{i=1\ldots n,j=1\ldots v_i},\arrs{\vec{q}}{I_{\Gamma}+n-V})}{n},\left(\tau'_{i,j},\vec{\kappa}_{i,j}\right)_{i=1\ldots n,j=1\ldots v_n}\right)\times
$$
$$
\left(\prod_{i=0}^{n} \prod_{j=1}^{v_i} d\tau'_{i,j}d\kappa_{i,j}\right)d^{3(I_{\Gamma}+n-V)}\arrs{\vec{q}}{I_{\Gamma}+n-V},
$$
where 
$$
N^{[h]}_{\Gamma}\in\SMltT{V+3(I_{\Gamma}+n)}{n},
$$
and $\vec{p}_{\Gamma}$ is a linear function from Lemma \ref{lem:Main/spatialIR}.
Then the necessary restriction is defined by 
$$
\restrictH{n+1,\ldots,n+4V}{1,\ldots,4V}(\WightmanRnaV{n}{\alpha}{V}\circ \mathcal{M}_{[h]}^{1,0})[f]((\tau'_{i,j},\kappa_{i,j})_{i=1\ldots n, j=1\ldots v_i},\arrs{t}{n})=    \sum_{\Gamma\in\mathfrak{G}_{\arrs{\alpha}{n}}^{\arrsM{v}{0}{n}}}  \int 
$$
$$N^{[h]}_{\Gamma}(\left(\tau'_{i,j}\right)_{i=0\ldots n, j=1\ldots v_i},\left(\vec{\kappa}_{i,j}\right)_{i=1\ldots n,j=1\ldots v_n},\arrs{\vec{q}}{I_{\Gamma}+n-V},\arrs{t}{n})\times$$
$$
f\left(\arrs{\vec{p}_{\Gamma}(\left(\vec{\kappa}_{i,j}\right)_{i=1\ldots n,j=1\ldots v_i},\arrs{\vec{q}}{I_{\Gamma}+n-V})}{n},\left(\tau'_{i,j},\vec{\kappa}_{i,j}\right)_{i=1\ldots n,j=1\ldots v_n}\right)\times
$$
$$d^{3(I_{\Gamma}+n-V)}\arrs{\vec{q}}{I_{\Gamma}+n-V},
$$

\end{proof}

\appendix 
\section{Feynman rules for the weak adiabatic limit}\label{appFR}Here we list the different Feynman rules that can be obtained from Corollary \ref{=>:HpQFT/FR/W-total} in the weak adiabatic limit. We use more physically natural form-factors $\mathcal{F} $ (see \ref{eq:HpQFT/AdmTh/hI-timeloc}) instead of $F$ to make a connection with the standard Feynman rules more explicit. For example, the usual $\phi^n$ interaction would correspond to constant factors $\mathcal{F}$. 
\par 
The first two versions (Remarks \ref{rmk:FR/MT-ord}, \ref{rmk:FR/ME-ord}) are generalizations of the results of \cite{BahnsPhD}, the other two (Remarks \ref{rmk:FR/MT-unord},\ref{rmk:FR/ME-unord} ) appeared in \cite{PhDThesis}
\begin{table}
    \centering
       \begin{tabular}{|m{60pt}|m{90pt}|m{150pt}|}
    \hline 
    \textbf{Element} 
    &
\textbf{Figure}
    &
    \textbf{Factor}
    \\
    \hline
    Internal vertex
    &

       \begin{tikzpicture}
       \node[] at (-1.2,1.6) 
    {$\vec{p}'_1$};

       \draw[black, thick] (-1.2,1.2)--(0,0);
       \draw[black, thick,<-] (-1.2,1.4)--(-0.8,1.0);
       
           \node[] at (0.5,1.4) 
    {$\vec{p}_1$};

       \draw[black, thick] (1,1)--(0,0);
       \draw[black, thick,->] (1,1.4)--(0.1,0.6);
         \node[] at (0.6,-1.2) 
    {$\vec{p}_l$};

       \draw[black, thick] (1,-1)--(0,0);
       \draw[black, thick,->] (0.9,-1.2)--(0.3,-0.6);
\filldraw[black] (0,0) circle (2pt) node[anchor=west] {$\tau$};
 \node[] at (-1,0.2) 
    {\Huge $\vdots$};
    \node[] at (1,0.2) 
    {\Huge $\vdots$};

     \node[] at (-1,-1.2) 
    {$\vec{p}'_{l'}$};

       \draw[black, thick] (-1.5,-1)--(0,0);
       \draw[black, thick,<-] (-1.4,-1.2)--(-0.8,-0.8);
\end{tikzpicture}  & $$\mathcal{F}_{(l',l)}(\arrs{\vec{p}'}{l'},\arrs{\vec{p}}{l})\cdot$$
$$\delta^{(3)}\left(\sum_{j=1}^{l'}\vec{p}'_j-\sum_{j=1}^{l}\vec{p}_j\right)\cdot$$
$$
\exp\left(\ii \left(\sum_{j=1}^{l'}\omega_0(\vec{p}'_j)-\sum_{j=1}^{l}\omega_0(\vec{p}_j)\right)\tau\right)
$$
\\
         \hline 
         Internal line
         & 
         \begin{tikzpicture}
        \node[] at (1,0.8) 
    {$\vec{k}$};
     %   \node[] at (0,-0.2) 
    %{$t$};
    %\node[] at (2,-0.2) 
    % {$t'$};
         \filldraw[black] (0,0)  circle (2pt);
         \filldraw[black] (2,0)  circle (2pt);
         \draw[black, thick] (0,0)--(2,0);
         \draw[black, thick,->] (0.5,0.5)--(1.5,0.5);
         \end{tikzpicture}
         &
         $$ \frac{1}{(2\pi)^3 2\omega_0(\vec{k})}$$
         \\
         \hline
        External vertex
         &
        \begin{tikzpicture}
        \node[] at (0.7,0.2) 
    {$\vec{p}$};
        \node[] at (0.8,-1.3) 
    {$i$};
         \draw[black, thick] (0,0)--(1,-1);
         \draw[black, thick,<-] (0.2,0.3)--(.8,-.2);
         \filldraw[black,thick,fill=white] (1,-1)  circle (2pt);
        \node[] at (1.7,-0.4) 
    {or};
          \begin{scope}[shift={(0.5,0)}]
         \node[] at (1.9,0.2)
    {$\vec{p}$};
        \node[] at (1.8,-1.3) 
    {$i$};
         \draw[black, thick] (1.5,-1)--(2.5,0);
         \draw[black, thick,<-] (1.7,-0.3)--(2.3,.2);
         \filldraw[black,thick,fill=white] (1.5,-1)  circle (2pt);
         \end{scope}
        \end{tikzpicture}
        &
        \begin{center}
            $e^{\ii\omega_0(\vec{p})t_i}\delta_{\alpha_i,+}\delta^{(3)}(\vec{k}_j-\vec{p})$ or $e^{-\ii\omega_0(\vec{p})t_i}\delta_{\alpha_i,-}\delta^{(3)}(\vec{k}_j+\vec{p})$
        \end{center}

         \\
         \hline
    \end{tabular}

    \caption{Ordered Feynman rules in the adiabatic limit, time-momentum presentation}
    \label{tab:FRules-TMAd}
\end{table}

\begin{table}
    \centering
    \begin{tabular}{|m{60pt}|m{90pt}|m{150pt}|}
    \hline 
    \textbf{Element} 
    &
\textbf{Figure}
    &
    \textbf{Factor}
    \\
    \hline
    Internal vertex
    &

       \begin{tikzpicture}
       \node[] at (-1.2,1.6) 
    {$\vec{p}'_1$};

       \draw[black, thick] (-1.2,1.2)--(0,0);
       \draw[black, thick,<-] (-1.2,1.4)--(-0.8,1.0);
       
           \node[] at (0.5,1.4) 
    {$\vec{p}_1$};

       \draw[black, thick] (1,1)--(0,0);
       \draw[black, thick,->] (1,1.4)--(0.1,0.6);
         \node[] at (0.6,-1.2) 
    {$\vec{p}_l$};

       \draw[black, thick] (1,-1)--(0,0);
       \draw[black, thick,->] (0.9,-1.2)--(0.3,-0.6);
\filldraw[black] (0,0) circle (2pt) node[anchor=west] {$\tau$};
 \node[] at (-1,0.2) 
    {\Huge $\vdots$};
    \node[] at (1,0.2) 
    {\Huge $\vdots$};

     \node[] at (-1,-1.2) 
    {$\vec{p}'_{l'}$};

       \draw[black, thick] (-1.5,-1)--(0,0);
       \draw[black, thick,<-] (-1.4,-1.2)--(-0.8,-0.8);
\end{tikzpicture}  & $$\mathcal{F}_{(l',l)}(\arrs{\vec{p}'}{l'},\arrs{\vec{p}}{l})\cdot$$
$$\delta^{(3)}\left(\sum_{j=1}^{l'}\vec{p}'_j-\sum_{j=1}^{l}\vec{p}_j\right)\cdot$$
\\
         \hline 
         Internal line
         & 
         \begin{tikzpicture}
        \node[] at (1,0.8) 
    {$\vec{k}$};
     %   \node[] at (0,-0.2) 
    %{$t$};
    %\node[] at (2,-0.2) 
    % {$t'$};
         \filldraw[black] (0,0)  circle (2pt);
         \filldraw[black] (2,0)  circle (2pt);
         \draw[black, thick] (0,0)--(2,0);
         \draw[black, thick,->] (0.5,0.5)--(1.5,0.5);
         \end{tikzpicture}
         &
         $$ \frac{1}{(2\pi)^3 2\omega_0(\vec{k})}$$
         \\
         \hline
        External vertex
         &
        \begin{tikzpicture}
        \node[] at (0.7,0.2) 
    {$\vec{p}$};
        \node[] at (0.8,-1.3) 
    {$i$};
         \draw[black, thick] (0,0)--(1,-1);
         \draw[black, thick,<-] (0.2,0.3)--(.8,-.2);
         \filldraw[black,thick,fill=white] (1,-1)  circle (2pt);
        \node[] at (1.7,-0.4) 
    {or};
          \begin{scope}[shift={(0.5,0)}]
         \node[] at (1.9,0.2)
    {$\vec{p}$};
        \node[] at (1.8,-1.3) 
    {$i$};
         \draw[black, thick] (1.5,-1)--(2.5,0);
         \draw[black, thick,<-] (1.7,-0.3)--(2.3,.2);
         \filldraw[black,thick,fill=white] (1.5,-1)  circle (2pt);
         \end{scope}
        \end{tikzpicture}
        &
        \begin{center}
            $\delta_{\alpha_i,+}\delta^{(3)}(\vec{k}_j-\vec{p})$ or $\delta_{\alpha_i,+}\delta^{(3)}(\vec{k}_j+\vec{p})$
        \end{center}

         \\
         \hline
    \end{tabular}

    \caption{Ordered Feynman rules in the adiabatic limit, energy-momentum presentation}
    \label{tab:FRules-EMAd}
\end{table}

\begin{table}
    \centering
    \begin{tabular}{|m{80pt}|m{80pt}|m{120pt}|}
    \hline 
    \textbf{Element} 
    &
\textbf{Figure}
    &
    \textbf{Expression}
    \\
    \hline
    Internal vertex
    &
       \begin{tikzpicture}
       \node[] at (-1.2,1.6) 
    {$\vec{k}_1$};
    \node[] at (-0.2,0.4) 
    {$\alpha_1$};
       \draw[black, thick] (-1,1)--(0,0);
       \draw[black, thick,->] (-1.2,1.4)--(-1,1.2);
           \node[] at (0.1,1.4) 
    {$\vec{k}_2$};
    \node[] at (0.4,0.4) 
    {$\alpha_2$};
       \draw[black, thick] (0.2,1)--(0,0);
       \draw[black, thick,->] (0.4,1.5)--(0.3,1);
         \node[] at (0.4,-1.2) 
    {$\vec{k}_3$};
    \node[] at (0.5,-0.4) 
    {$\alpha_3$};
       \draw[black, thick] (0.5,-1)--(0,0);
       \draw[black, thick,->] (0.8,-1.2)--(0.7,-1);
\filldraw[black] (0,0) circle (2pt) node[anchor=west] {$\tau$};
 \node[] at (-0.3,-0.7) 
    {\huge $\ldots$};
     \node[] at (-1,-1.2) 
    {$\vec{k}_n$};
    \node[] at (-0.6,-0.2) 
    {$\alpha_n$};
       \draw[black, thick] (-1.5,-1)--(0,0);
       \draw[black, thick,->] (-1.6,-1.2)--(-1.3,-1);
\end{tikzpicture}  & $$\mathcal{F}^{\alpha_1,\ldots,\alpha_n}(\vec{k}_1,\ldots,\vec{k}_n)\cdot$$
$$\delta(\vec{k}_1+\cdots+\vec{k}_n)$$ \\
         \hline 
         Internal line
         & 
         \begin{tikzpicture}
         \node[] at (1,0.8) 
    {$\vec{k}$};
    \node[] at (0,0.2) 
    {$\alpha$};
    \node[] at (2,0.2) 
    {$\beta$};
        \node[] at (0,-0.2) 
    {$t$};
    \node[] at (2,-0.2) 
    {$t'$};
         \filldraw[black] (0,0)  circle (2pt);
         \filldraw[black] (2,0)  circle (2pt);
         \draw[black, thick] (0,0)--(2,0);
         \draw[black, thick,->] (0.5,0.5)--(1.5,0.5);
         \end{tikzpicture}
         &
         $$\delta_{\beta,-\gamma} \theta(\beta(t-t'))\frac{e^{-i\omega_0({\vec{k}})(t-t')}}{2\omega_0({\vec{k}})(2\pi)^3}$$
         \\
         \hline
        External vertex
         &
      \begin{tikzpicture}
        \node[] at (0.7,0.2) 
    {$\vec{p}$};
        \node[] at (1,-1.3) 
    {$i$};
         \draw[black, thick] (0,0)--(1,-1);
         \draw[black, thick,<-] (0.2,0.3)--(.8,-.2);
         \filldraw[black,thick,fill=white] (1,-1)  circle (2pt);
    \end{tikzpicture}
    &
    $$\delta^{(3)}(\vec{k}_j-\alpha_j\vec{p})$$
         \\
         \hline
    \end{tabular}

    \caption{Unordered Feynman rules in the time-momentum space in the adiabatic limit}
    \label{tab:FRules-TM-unord}
\end{table}

\begin{table}
    \centering
    \begin{tabular}{|m{70pt}|m{80pt}|m{140pt}|}
        \hline 
    \textbf{Element} 
    &
\textbf{Figure}
    &
    \textbf{Factor}
    \\

    \hline
    Internal vertex
    &
       \begin{tikzpicture}
       \node[] at (-1.2,1.6) 
    {$\vec{k}_1,\omega_1$};
    \node[] at (-0.2,0.4) 
    {$\alpha_1$};
       \draw[black, thick] (-1,1)--(0,0);
       \draw[black, thick,->] (-1.2,1.4)--(-1,1.2);
           \node[] at (0.1,1.8) 
    {$\vec{k}_2,\omega_2$};
    \node[] at (0.4,0.4) 
    {$\alpha_2$};
       \draw[black, thick] (0.2,1)--(0,0);
       \draw[black, thick,->] (0.4,1.5)--(0.3,1);
         \node[] at (0.4,-1.2) 
    {$\vec{k}_3,\omega_3$};
    \node[] at (0.5,-0.4) 
    {$\alpha_3$};
       \draw[black, thick] (0.5,-1)--(0,0);
       \draw[black, thick,->] (0.8,-1.2)--(0.7,-1);
\filldraw[black] (0,0) circle (2pt);
 \node[] at (-0.3,-0.7) 
    {\huge $\ldots$};
     \node[] at (-1,-1.2) 
    {$\vec{k}_n,\omega_n$};
    \node[] at (-0.6,-0.2) 
    {$\alpha_n$};
       \draw[black, thick] (-1.5,-1)--(0,0);
       \draw[black, thick,->] (-1.6,-1.2)--(-1.3,-1);
\filldraw[black] (0,0) circle (2pt) node[anchor=west] {$\tau$};
\end{tikzpicture}  &  $$-i(2\pi)^4 F^{\beta_1,\ldots,\beta_m}( \vec{k}_{1},\ldots ,\vec{k}_{m})$$
$$\delta(\omega'_1+\cdots +\omega'_n)\delta^{(3)}(\vec{k}_1+\cdots \vec{k}_m),$$\\
         \hline 
         Line
         &
         \begin{tikzpicture}
         \node[] at (1,0.8) 
    {$\vec{k},\omega$};
    \node[] at (0,0.2) 
    {$\alpha$};
    \node[] at (2,0.2) 
    {$\beta$};
         \filldraw[black] (0,0)  circle (2pt);
         \filldraw[black] (2,0)  circle (2pt);
         \draw[black, thick] (0,0)--(2,0);
         \draw[black, thick,->] (0.5,0.5)--(1.5,0.5);
         \end{tikzpicture}
         &
         $$\delta_{\alpha,-\beta}\cdot$$
         $$\frac{-i \alpha}{2\omega_0({\vec{p}})(2\pi)^4(\omega+\alpha (\omega_0({\vec{p}})-i0))}$$
         \\
         \hline
         External vertex
         &
        \begin{tikzpicture}
        \node[] at (0.7,0.2) 
    {$\vec{p},\omega$};
        \node[] at (0.8,-1.3) 
    {$i$};
         \draw[black, thick] (0,0)--(1,-1);
         \draw[black, thick,<-] (0.2,0.3)--(.8,-.2);
         \filldraw[black,thick,fill=white] (1,-1)  circle (2pt);
        \end{tikzpicture}
        &
        \begin{center}
            $$\delta^{(3)}(\vec{k}_i-\alpha\vec{p})\delta(\omega_i-\alpha_i\omega)$$    \end{center}

         \\
         \hline
    \end{tabular}

    \caption{Unordered Feynman rules in the energy-momentum space in the adiabatic limit}
    \label{tab:FRules-EM}
\end{table}

\begin{rmk} \label{rmk:FR/MT-ord}
For $n\in\NN_0$, $f\in\SRR{4n}$ $\arrs{\alpha}{n}\in\{+,-\}$ and a formal parameter $g$ the symbolic expression for 
$\GreenUnagAd{n}{\alpha}{g}[f]$ may be computed by the following Feynman rules:
\begin{itemize}
 \item The relevant are all partially ordered Feynman graphs $(\Gamma,\prec_{\Gamma})$ with $V$ non enumerated internal vertices $\bullet$, $j=1,\ldots, v$ and $n$ external vertices $\circ_j$, $j=1,\ldots,n$ and no vacuum corrections;
    \item To each line a momentum flux, directed from earlier to later vertex is assigned;
    \item To each internal vertex is assigned a timestamp (a variable in $\RR{}$); 
    \item For convenience with the external vertex $\circ_i$ we associate the timestamp $t_i$ and the momentum $\vec{p}_i$ which are treated as free variables;
    \item Factors corresponding to elements of the diagrams are presented in Table \ref{tab:FRules-TMAd};
    \item The overall factor is $f((t_i,\vec{k}_i)_{i=1\ldots n})$;
    \item The integration region is bounded by the condition that timestamps corresponding to earlier vertices (internal or external) have smaller values than the ones corresponding to the later.
\end{itemize}

\end{rmk}

For practice it is often more convenient to work with energy and momentum presentation rather than with the mixed one as above. We define
$$
\GreenEUnagAd{n}{\alpha}{g}\in\SpRR{4n}
$$
by
$$
\GreenUnagAd{n}{\alpha}{g}((\vec{p}_i,t_i)_{i=1\ldots n})=\int \GreenEUnagAd{n}{\alpha}{g}((\vec{p}_i,\omega_i)_{i=1\ldots n})e^{-\ii\sum_{i=1}^{n}\omega_i t_i}d^{n}\arrs{\omega}{n}.
$$
Then it is convenient to compute the integrals over the timestamps of both internal and external timestamps similarly to treatment of the outer vertices in Subsubsection \ref{Main/Temp} following result. For that we should use total rather than partial order.
\begin{rmk}\label{rmk:FR/ME-ord}
For $n\in\NN_0$, $f\in\SRR{4n}$ $\arrs{\alpha}{n}\in\{+,-\}$ and a formal parameter $g$ the symbolic expression for 
$\GreenUnagAd{n}{\alpha}{g}[f]$ may be computed by the following Feynman rules:
\begin{itemize}
 \item The relevant are all totally ordered Feynman graphs $(\Gamma,\prec_{\Gamma})$ with $V$ non enumerated internal vertices $\bullet$, $j=1,\ldots, v$ and $n$ external vertices $\circ_j$, $j=1,\ldots,n$ and no vacuum corrections;
    \item To each line a momentum flux and energy, directed from the earlier to the later vertex is assigned;
    \item For convenience with the external vertex $\circ_i$ we associate the energy $\omega_i$ and the momentum $\vec{p}_i$ which are treated as free variables;
    \item Factors corresponding to elements of the diagrams are presented in Table \ref{tab:FRules-EMAd};
    \item The overall factor is $f((\omega_i,\vec{k}_i)_{i=1\ldots n})\frac{\delta(\sum_{\diamond} \Omega_{\diamond})}{\prod'_{\diamond} \sum_{\diamond'\prec\diamond }\Omega_{\diamond'}}$, where the product goes over all (internal and external) vertices $\diamond$ except for the last one, sum goes over all vertices $\diamond'$ preceding $\diamond$ and the on-shell energy defects are defined as
    $\Omega_{\diamond}$ are computed according to Table \ref{rmk:FR/MT-unord};
    \item Alternatively the vertex factor can be replaced by
    $f((\omega_i,\vec{k}_i)_{i=1\ldots n})\frac{\delta(\sum_{\diamond} \Omega_{\diamond})}{\prod'_{\diamond} \sum_{-\diamond'\succ\diamond }\Omega_{\diamond'}+i0}$, where this time the product goes over all (internal and external) vertices except for the first one;
\end{itemize}
\end{rmk}
The two versions of the Feynman rules are quite effective for computation of the contribution of a small graph. But they are too much different from the traditional formulation in the physical literature which makes even qualitative comparison with the local theories more difficult. In addition to that the standard techniques to manage large classes of Feynman graphs, such as resummations of the one-particle corrections,  are not available for the formulation of Remark \ref{rmk:FR/ME-ord} where the total ordering makes all the pieces entangled. 
\par 
To avoid that we replace summations over ordering with summations over orientations of lines. Moreover, we formally allow lines that are incoming or outgoing at both ends. For that we modify the propagator so that it vanishes whenever the order of timestamps and orientations of the lines are not consistent. Note that in this formulation the direction of the energy and momentum flows is not fixed anymore.

\begin{rmk}\label{rmk:FR/MT-unord}
For $n\in\NN_0$, $f\in\SRR{4n}$ $\arrs{\alpha}{n}\in\{+,-\}$ and a formal parameter $g$ the symbolic expression for 
$\GreenUnagAd{n}{\alpha}{g}[f]$ may be computed by the following Feynman rules:
\begin{itemize}
 \item The relevant are all (unordered) Feynman graphs $\Gamma$ with $V$ non enumerated internal vertices $\bullet$,  and $n$ external vertices $\circ_j$, $j=1,\ldots,n$, and no vacuum corrections;
    \item To each line a momentum flux is assigned;
    \item To each end of each line is assigned a time-orientation (i.e. a sign $+$ or $-$);
    \item To each internal vertex is assigned a timestamp (a variable in $\RR{}$); 
    \item For convenience with the external vertex $\circ_i$ we associate the timestamp $t_i$ and the momentum $\vec{p}_i$ which are treated as free variables;
    \item Factors corresponding to elements of the diagrams are presented in Table \ref{tab:FRules-EM};
    \item The overall factor is $f((t_i,\vec{k}_i)_{i=1\ldots n})$;
    \item The integration region is bounded by the condition that timestamps corresponding to earlier vertices (internal or external) have smaller values than the ones corresponding to the later.
\end{itemize}
\end{rmk}

After the Fourier transform we get
\begin{rmk} \label{rmk:FR/ME-unord}For $n\in\NN_0$, $f\in\SRR{4n}$ $\arrs{\alpha}{n}\in\{+,-\}$ and a formal parameter $g$ the symbolic expression for 
$\GreenEUnagAd{n}{\alpha}{g}[f]$ may be computed by the following Feynman rules:
\begin{itemize}
 \item The relevant are all (unordered) Feynman graphs $\Gamma$ with $V$ non enumerated internal vertices $\bullet$,  and $n$ external vertices $\circ_j$, $j=1,\ldots,n$, and no vacuum corrections;
    \item To each line an energy and a momentum flux;
    \item To each end of each line is assigned a time-orientation (i.e. a sign $+$ or $-$);
    \item For convenience with the external vertex $\circ_i$ we associate the timestamp $t_i$ and the momentum $\vec{p}_i$ which are treated as free variables;
    \item Factors corresponding to elements of the diagrams are presented in Table \ref{tab:FRules-EM};
    \item The overall factor is $f((t_i,\vec{k}_i)_{i=1\ldots n})$;
    \item The integration region is bounded by the condition that timestamps corresponding to earlier vertices (internal or external) have less values than the ones corresponding to the later.
\end{itemize}
\end{rmk}
One may note that in the local case $$\mathcal{F}_{\alpha_1,\ldots,\alpha_m}(\vec{k}_1,\ldots,\vec{k}_m)=\lambda_m$$
summation over the time-orientation of the propagators may be performed leading to the standard Feynman propagator.
\par 
We refer to \cite{PhDThesis} where the Feynman rules of Remark \ref{rmk:FR/ME-unord} were used to make contact of the time-dependent Hamiltonian perturbation theory with the time-independent one used in \cite{Kossow} to compute the quantum corrections to dispersion relations. It was shown, that as in the conventional QFT the dispersion relation corrections appear as shifts of the pole of the corrected propagator. We postpone the precise interpretation of this result to further publications.

\end{document}